\documentclass{article}
\usepackage[margin=1in]{geometry}
\usepackage{authblk}

\usepackage{preamble}

\newcommand{\kelly}{\texttt{ApproxKelly}\xspace}

\newcommand{\Wtilde}{\widetilde{W}\xspace}
\newcommand{\Mtilde}{\widetilde{A}\xspace}

\newcommand{\wor}{WoR\xspace}

\newcommand{\var}{\mathbb{V}}
\newcommand{\cov}{\mathrm{Cov}}

\newcommand{\mhat}{\widehat{m}}
\newcommand{\vhat}{\widehat{v}}

\newcommand{\iid}{i.i.d.}
\newcommand{\rlfa}{(\varepsilon, \delta)\text{-}\mathrm{RLFA}\xspace}

\title{ Risk-limiting Financial Audits  \\via Weighted Sampling without Replacement
}
\date{}

\author[1]{Shubhanshu Shekhar}
\author[1]{Ziyu Xu}
\author[2,3]{\authorcr Zachary C. Lipton}
\author[3]{Pierre J. Liang} 
\author[1,2]{Aaditya Ramdas}
\affil[1]{Department of Statistics and Data Science, Carnegie Mellon University}
\affil[2]{Machine Learning Department, Carnegie Mellon University}
\affil[3]{Tepper School of Business, Carnegie Mellon University}

\begin{document}

\maketitle

\begin{abstract}
    
We introduce the notion of a risk-limiting financial auditing~(RLFA): given $N$ transactions, the goal is to estimate the total misstated monetary fraction~($m^*$) to a given accuracy $\epsilon$, with confidence $1-\delta$.  We do this by constructing new confidence sequences~(CSs)  for the weighted average of $N$ unknown values,  based on samples drawn without replacement  according to a (randomized) weighted sampling scheme.  Using the idea of importance weighting to construct test martingales,  we first develop a framework to construct CSs  for arbitrary sampling strategies.  Next, we develop methods to improve the quality of CSs by incorporating side information about  the unknown values associated with each item. We show that when the side information is sufficiently predictive, it can directly drive the sampling. Addressing the case where the accuracy is unknown \emph{a priori}, we introduce a method that incorporates  side information via control variates. Crucially, our construction is adaptive: if the side information is highly predictive of the unknown misstated amounts, then the benefits of incorporating it are significant;  but if the side information is uncorrelated,  our methods learn to ignore it.  Our methods recover state-of-the-art bounds  for the special case when the weights are equal, which has already found applications in election auditing. The harder weighted case solves our more challenging problem of AI-assisted financial auditing.

\end{abstract}
\tableofcontents

\section{Introduction}
    \label{sec:introduction}

    Consider the following scenario: in a given year, 
    a company has $N$ recorded financial transactions 
    with reported monetary values $M(i) \in (0, \infty)$ 
    for each $i \in [N] \coloneqq \{1, \dots, N\}$.
    As required by law, an external auditor is required 
    to attest with ``reasonable assurance'' 
    about whether the financial records 
    as a whole are free from ``material misstatement.''
   For example, the company has cash receipts for sales of products, 
   and it wants to ensure that the reported monetary value 
   matches the true amount that was made on the sales 
   according to prescribed accounting rules as some receipts 
   may actually represent past sales or future deliveries. 
   This can be done, for instance, 
   by manually examining the entire sales process 
   to determine the true sales amount 
   against the the amount recorded by the company.
    Since the task of \emph{auditing} each transaction
    can be %
    complex %
    requires substantial human labor
    it can be prohibitively expensive to perform
    a comprehensive audit of a company's records.

    Suppose that the auditor has built 
    an AI system for ``automated auditing'',
    i.e., this AI system can output predictions 
    about the accuracy of a transaction value, 
    based on receipts, OCR (optical character recognition), databases, etc. 
    Such systems are in a state of active development and deployment,
    and the high level of industry demand is unsurprising
    given the remarkable predictive capabilities 
    of modern machine learning algorithms.
    But there's a catch: because the system is trained and deployed
    on differently distributed data, 
    its accuracy on a new set of records 
    in a new time period
    is unknown \emph{a priori}.
    Even if anecdotally, the AI system seems to perform reasonably well
    on data collected from a variety of companies, 
    we cannot make statistically certifiable conclusions
    based solely on the output of the AI system 
    on a new company and/or in a new time period.
    Thus we can think of AI systems in deployment
    as black boxes for which we have
    (reasonable) hopes of high accuracy
    but lack formal guarantees.

    The auditor's goal is to minimize the amount of manual auditing that must be done by a person, while accurately estimating the true monetary amount of those transactions that have not manually audited. When the AI system is accurate, we want to reduce the amount of human auditing effort required. More importantly, we want a statistically rigorous conclusion regardless of the AI system accuracy. Hence, our method should interpolate between using predictions to reduce its uncertainty rapidly when the system is accurate, and the most efficient AI-free strategy when the system is inaccurate.

    \paragraph{Problem setup and notation.} Denote the unknown misstated fraction of the $i$th transaction as $f(i) \in [0,1]$, for each $i \in [N]$. In other words, if $M^*(i)$ denotes the true value of the transaction $i$, and $M(i)$ is the reported value, then~\footnote{We are primarily concerned with estimating the downside that arises from misstatement, e.g., $M(i)$ represents the money that should have been received for a sale, and $M^*(i)$ represents the actual money received. In this scenario, we may lose at most $M(i)$ amount of money if $M^*(i) = 0$. Hence, we assume $f(i) \in [0, 1]$. 
    } $f(i) = |M^*(i)-M(i)|/M(i)$. We can normalize the reported transaction values by the sum over all transaction values to get a weight $\pi(i) \coloneqq M(i) / (\sum_{i = 1}^N M(i))$ for each $i \in [N]$, where $\sum_{i=1}^n \pi(i) = 1$. The auditor wishes to obtain an estimate of  $m^* = \sum_{i=1}^N \pi(i)f(i)$, the fraction of the total monetary value that is misstated, up to an accuracy $\varepsilon \in [0,1]$. By $S(i)$, we denote the \emph{side information}, a score for the $i$th transaction 
    that (ideally) predicts $f(i)$.
    In our setup, the side information can be generated through any method, e.g., through an AI system that automatically analyzes the documents 
    a human auditor would use, may also be available to the auditor.
    Each transaction can be evaluated 
    by the auditor to reveal $M^*(i)$ 
    (or equivalently, $f(i)$). 
    Thus, \emph{given an $\varepsilon>0$, 
    in what order should the transactions be audited 
    to estimate $m^*$ within $\varepsilon$ additive accuracy, 
    using the fewest number of calls to the auditor?}

    If we allow for no uncertainty, i.e., 
    we want to produce a confidence interval (CI) 
    for $m^*$ with 100\% confidence, 
    then the best strategy is to audit the transactions 
    in decreasing order of their reported value,
    and stop when the remaining transactions 
    constitute smaller than an $\varepsilon$ fraction of the total. 
    However, we can show that if we want to provide an estimate of $m^*$ 
    that is $\varepsilon$-accurate with probability at least $1-\delta$, 
    for a tolerance level $\delta$~(e.g., $0.01$), 
    there exist strategies based on randomized sampling \wor 
    that allow us to stop much earlier. 
    In other words, for each $t \in [N]$, 
    we adaptively construct a sampling distribution $q_t$ 
    over the remaining $N-t+1$ unaudited transactions,
    and sample $I_t$, the index of the $t$th transaction to audit, 
    according to $q_t$. 
    We then obtain $f(I_t)$ through manual auditing, 
    and incorporate this new information 
    to update our estimate of $m^*$. 
    If our residual uncertainty is sufficiently small~(i.e., 
    smaller than $\varepsilon$), we stop sampling. 
    Otherwise, we continue the process 
    by drawing the next index, $I_{t+1}$, 
    according to an appropriately chosen distribution $q_{t+1}$.  
    
    Before presenting the technical details,
    we note that we use $(X_t)_{t \in \mathbb{I}}$ 
    to denote a sequence of objects indexed by a set $\mathbb{I}$, 
    and the $t$th object is $X_t$. 
    We drop the indexing subscript if it is clear from context. 
    For any $t \in [N]$, 
    we use $\filtration_t \coloneqq \sigma(\{I_i\}_{i \in [t]})$ 
    to denote the sigma-algebra over our query selections 
    for the first $t$ queries.

    \paragraph{Risk-limiting financial audit (RLFA).} Formally, a $(\varepsilon, \delta)$-\textit{risk-limiting financial audit (RFLA)} is a procedure that outputs an interval $\mc{C}$ where $|\mc{C}| \leq \varepsilon$ and $\mc{C}$ contains the true misstated fraction, $m^*$, with probability at least $1 - \delta$. This is a natural generalization of risk-limiting audits that are used to ensure statistically valid election auditing \citep{stark_conservative_statistical_2008a,stark_cast_canvass_2009,lindeman2012gentle} to the financial setting, where each transaction is weighted by its reported monetary value (as opposed to uniform weighting for all votes in the election setting). We also consider other possible definitions of an RLFA in \Cref{sec:alt-defs}. Our goal is to produce $\mc{C}$ that satisfies the conditions of an RLFA with as few audits, i.e, queries of $f$, as possible. To produce such an interval, we propose a framework for building RLFAs by constructing confidence sequences, which we introduce next.

    \paragraph{Confidence sequences for sequential estimation.}  Let $T \in [N]$ be a random stopping time, that is,  a random variable for which the event $\{T = t\}$ belongs to $\filtration_t$ for each $t \in [N]$, and  let $\mc{T}$ denote the universe of all such stopping times. \emph{Confidence sequences}~\citep{lai1976confidence, howard2021time} (CSs), or time-uniform confidence intervals, are sequences of intervals, $(\mc{C}_t)_{t \in [N]}$, that satisfy
    \begin{align}
        \label{eq:conf-seq}
        \underset{T \in \mc{T}}{\sup}\ \prob{m^* \not\in \mc{C}_T} \leq \delta \Leftrightarrow
        \mathbb{P} \lp \exists t \in [N]:  m^* \not \in \mc{C}_t \rp \leq \delta,
    \end{align} where $\delta \in (0, 1)$ is a fixed error level. \citet{ramdas_admissible_anytime-valid_2020} showed the equivalence above, i.e., that any sequence of intervals $(\mc{C}_t)$ that satisfies one side of the implication will immediately satisfy the other as well. 
    
    Using this equivalence, we can define a simple $\rlfa$ procedure: construct a CS for $m^*$, denoted by $(\mc{C}_t)$, and produce $\mc{C}_\tau$ where $\tau$ is the following stopping time:
    \begin{gather}
    \tau = \tau(\varepsilon, \delta) \defined  \min \{t \geq 1: |\mc{C}_t| \leq \varepsilon \}. \label{eq:TauDef}
    \end{gather}
    The width of all nontrivial CSs converges to zero as $t \to N$, 
    and thus the above stopping time is well-defined, 
    and is usually smaller than $N$. 
    
    Note that the only source of randomness in this problem 
    is the randomized sampling strategy $(q_t)_{t \in [N]}$, 
    used to select transactions for manual evaluation. 
    Hence, $(q_t)_{t \in [N]}$ is another design choice for us to make. 
    To summarize, our goal in  this paper is 
    to \textbf{(i)} design sampling strategies $(q_t)$, 
    and \textbf{(ii)} develop methods of aggregating 
    the information so collected with any available side information, 
    in order to construct CSs for $m^*$
    whose width decays rapidly to $0$. 
    
    Among existing works in literature, the recent papers by~\citet{waudby2020estimating, waudby2020confidence} 
    are the most closely related to our work. 
    In these works, the authors considered 
    the problem of estimating the average value of $N$ items 
    via \wor sampling---however, they considered only uniform sampling, 
    and estimating only the unweighted mean of the population. 
    Our methods work with any sampling scheme, 
    and can estimate any weighted mean; 
    we recover their existing results 
    in \Cref{sec:HoefEBComparison}.

        \paragraph{\wor confidence intervals for a fixed sample size.}~Most existing results on concentration inequalities for observations drawn via \wor sampling focus on the fixed sample size setting, starting with  \citet{hoeffding1963probability}, who bounded the probability of deviation of the unweighted empirical mean with \wor sampling in terms of the range of the observations. In particular, \citet{hoeffding1963probability} showed that for observations $X_{I_1}, \ldots, X_{I_n} \in [a,b]$ drawn uniformly \wor from $N$ values $(X_i)_{i \in [N]}$, we have
        \begin{align}
        {
            \label{eq:hoeffding-fixed-time}
            \mathbb{P}\lp \tfrac{\sum_{t=1}^n X_{I_t}}{n}  - \tfrac{\sum_{t=1}^N X_i}{N}  > \varepsilon \rp \leq \exp \lp - \tfrac{2 n \varepsilon^2}{(b-a)^2} \rp.}
        \end{align}
        In \wor sampling, as the sample size $n$ approaches $N$, the total number of items, we expect the empirical estimate to approximate the true average very accurately. This observation, not captured by the above bound, was made formal by \citet{serfling1974probability}, who showed that the $n$ in~\eqref{eq:hoeffding-fixed-time} can be replaced by $\frac{n}{1-(n-1)/N}$, thus highlighting the significant improvement possible for larger $n$ values.
        \citet{ben2018weighted} prove a Hoeffding style concentration inequality on the unweighted sample mean to its own expectation, which is a different estimand than the weighted population mean considered in this paper.  Finally, in the unweighted case, \citet{bardenet2015concentration} obtained variance adaptive Bernstein and empirical-Bernstein variants of Serfling's results, that are tighter in cases where the variance of the observations is small. 
        These results appear to be incomparable to those of~\citet{waudby2020confidence,waudby2020estimating}, 
        that have found successful application to auditing elections~\cite{waudby2021rilacs}.
        In this paper, we develop techniques that generalize the CS constructions of~\citet{waudby2021rilacs, waudby2020confidence} in order to estimate the weighted average of $N$ quantities~(instead of simple, unweighted average) sampled via an adaptive scheme (instead of uniform), motivated by financial auditing applications.

    \subsection{Contributions}
    \label{subsec:overview}
    We introduce the concept of $\rlfa$ that generalizes 
    the notion of a risk-limiting audit introduced by \citet{stark_conservative_statistical_2008a} for election auditing. Unlike risk-limiting audits, where the main concern is testing an announced result, the objective of an RLFA is to precisely estimate the misstated monetary fraction of the reported financial transactions. 
    To accomplish this, we make the following key technical contributions:
        \begin{enumerate}[itemsep=0em, leftmargin=*]
            \item \emph{New CSs for weighted means with non-uniform sampling.} To design an $\rlfa$ procedure, we construct novel CSs for $m^*$ that are based on a betting method that was pioneered in \citep{waudby2020estimating} in \Cref{sec:side information}, as well as Hoeffding and empirical-Bernstein CSs in \Cref{sec:hoeffding-empirical-bernstein} (which are looser but have a simple analytical form). Our results generalize previous methods in two ways: \textbf{(i)} they can estimate the weighted mean of $N$ items, and \textbf{(ii)} they work with adaptive, data-dependent, sampling strategies.
            In particular, our betting CSs, which we show empirically are the most powerful in \Cref{sec:HoefEBExperiments}) are based on simultaneously playing  gambling games with an aim to disprove the possibility that $m^* =m$, for each $m \in [0, 1]$. Values for $m$, where we accumulate much wealth are eliminated from the CS. Consequently, we develop a simple, lucrative betting strategy for this setting (\kelly), which is equivalent to formulating narrower CSs. 
            \item \emph{Adaptive sampling strategies that minimize CS width.} In addition to designing CSes that are intrinsically narrow, we are also able to change the sampling distribution of the transactions at each time step, and develop a sampling strategy that will minimize CS width in concert with any valid CS construction. 
            We propose two sampling strategies, \propM and \propMS, the latter of which can incorporate approximately accurate scores $(S(i))_{i \in [N]}$ to improve the sample efficiency of our CSs. This is accomplished by choosing the sampling distribution, at each time step, that maximizes the wealth accumulated by the betting strategies that underlie our CSs. We find that this is approximately equivalent to choosing the sampling distribution with the minimal variance, and we show that our sampling strategies result in a noticeable improvement over uniform sampling through simulations in \Cref{sec:experiments}. 
            \item \emph{Robust use of side information to tighten CSs.} Finally, in~\Cref{sec:side information}, we develop a principled way of leveraging any available side information, inspired by the idea of control variates used for variance reduction in  Monte Carlo sampling. Interestingly, our method adapts to the quality of the side information---if $(S(i))_{i \in [N]}$ and $(f(i))_{i \in [N]}$ are highly correlated, the resulting CSs are tighter, while in the case of uncorrelated $(S(i))$, we simply learn to discard the side information. 
\end{enumerate}

\section{Betting-based CS construction}
    \label{subsec:betting-CS-no-side-info}
    We derive our CSs by designing sequential tests to simultaneously check the hypotheses that $m^*=m$, for all $m \in [0,1]$. 
    By the principle of \emph{testing by betting}~\citep{shafer_testing_betting_2021}, 
    this is equivalent to playing repeated gambling games 
    aimed at disproving the null $m^*=m$, for each $m \in [0,1]$. 
    Formally, for all $m \in [0,1]$, 
    we construct a process $(W_t(m))_{t \in [N]}$~(the wealth process), 
    such that \textbf{(i)} if $m=m^*$, 
    then $(W_t(m))$ is a \emph{test martingale}, 
    i.e., a nonnegative martingale with initial value $1$, 
    and \textbf{(ii)} if $m \neq m^*$,
    then $W_t(m)$ grows at an exponential rate.
    Recall that a process $(W_t)_{t \in [N]}$ adapted to $(\filtration_t)_{t \in [N]}$ is a supermartingale iff  $\expect[W_t \mid \filtration_{t - 1}] \leq W_{t - 1}$ for all $t \in [N]$,
    and a martingale if the inequality is replaced with an equality. 
    Assuming we can construct such a process, 
    we define the confidence set at any time $t$ 
    as the set of those $m \in [0,1]$ 
    for which $(W_t(m))$ is `small', 
    because a nonnegative martingale 
    is unlikely to take large values. 
    
    As mentioned earlier, this approach requires us 
    to design sampling distributions $(q_t)$, 
    and a method for constructing a CS $(\mathcal{C}_t)$ 
    from the queried indices.  
    We begin by formally defining a sampling strategy.
    \begin{definition}[Sampling Strategy]
        \label{def:sampling-strategy}
        A sampling strategy consists of a sequence $(q_t)_{t \in [N]}$, where $q_t$ is a probability distribution on the set $\mc{N}_t \defined [N] \setminus \{I_1, \ldots, I_{t-1}\}$. Here $I_j$ denotes the index drawn according to the predictable~(i.e., $\filtration_{j-1}$-measurable) distribution $q_j$. 
    \end{definition}

    A natural baseline sampling strategy is to set $q_t$ to be uniform over  $\mc{N}_t$ for all $t \in [N]$. We will develop other, more powerful, sampling strategies that are more suited to our problem in~\Cref{sec:sampling-strategies}. 

    We now describe how to construct the wealth process for an arbitrary sampling strategy. First, define the following:
    \begin{align}
        Z_t \coloneqq f(I_t) \tfrac{\pi(I_t)}{q_t(I_t)},\text{ and } \mu_t(m) \defined m - \sum_{j=1}^{t-1} \pi(I_j) f(I_j).
    \end{align} Note that $\mu_t(m)$ is the remaining misstated fraction after accounting for the first $t - 1$ queries to $f$ if $m$ is truly the total misstated fraction.
    Now, we can define the \emph{wealth process}:
    \begin{align}
        W_t(m) = W_{t-1}(m) \times \lp 1 + \lambda_t(m)\lp Z_t - \mu_t(m) \rp\rp, \label{eq:wealth-process-0}
    \end{align}
    with $W_0=1$. $(\lambda_t(m))_{t \in [N]}$ is a predictable sequence with values in $[0,1/u_t(m)]$, and  $u_t(m)$ is the largest value in the support of $Z_t - \mu_t(m)$, for each $t \in [N]$. 
     Note that this constraint on $(\lambda_t(m))$  ensures that  $W_t(m)$ is nonnegative for each $t \in [N]$. We also let $W_0(m) = 1$ for all $m \in [0, 1]$. If we view the wealth process as the wealth we earn from gambling on the outcome of $Z_t - \mu_t(m)$, then $(\lambda_t(m))$ represents a betting strategy, i.e.,  how much money to gamble each turn. Hence, we refer to $(\lambda_t(m))$ as a \emph{betting strategy}.

    It is easy to verify that $(W_t(m^*))$ is a nonnegative martingale  for any sampling strategy $(q_t)$ and betting strategy $(\lambda_t(m^*))$). Hence, it is unlikely to take large values, as we describe next. 
    \begin{proposition}
       \label{prop:type-I}
       For any sampling and betting strategies $(q_t)$ and $(\lambda_t(m^*))$, the following holds:
       \begin{align}
           \mathbb{P}\lp \exists t \geq 1: W_t(m^*) \geq 1/\delta \rp \leq \delta.
       \end{align}
    \end{proposition}
    This is a consequence of Ville's inequality, first obtained by~\citet{ville1939etude}, which is a time-uniform version of Markov's inequality for nonnegative supermartingales.
    This result immediately implies that  for any sampling strategy, and any betting strategy, the term $m^*$ must lie in the set
    \begin{align}
        \label{eq:conf-seq-def-1}
        \mc{C}_t = \{m : W_t(m) < 1/\delta\}
    \end{align}    with probability at least $1-\delta$, making $(\mc{C}_t)$ a $(1 - \delta)$-CS. 
    \begin{theorem}
        $(\mc{C}_t)$ is an $(1 - \delta)$-CS, where $\mc{C}_t$ defined by \eqref{eq:conf-seq-def-1}. Hence, a procedure that outputs $\mc{C}_\tau$ is an $\rlfa$, for any sampling strategy $(q_t)$ and betting strategies $(\lambda_t(m))$ for each $m \in [0, 1]$. Recall that the $\tau$ is defined in \eqref{eq:TauDef} as the first time where $|\mc{C}_t| \leq \varepsilon$.
        \label{thm:CSthm}
    \end{theorem}
    This methodology gives us flexible framework for constructing different $(\mc{C}_t)$ that result in different RLFAs. Now, we can turn our attention to finding betting strategies $(\lambda_t(m))$ that reduces the CS width quickly and minimizes $\tau$.

    \begin{remark}
        \label{remark:betting-CS-width}
        Note that the set $\mc{C}_t$ in~\eqref{eq:conf-seq-def-1}, does not admit a closed form expression, and is computed numerically in practice by choosing $m$ values over a sufficiently fine grid on $[0,1]$. In~\Cref{sec:hoeffding-empirical-bernstein}, we design CSs based on nonnegative supermartingales~(instead of martingales) that do admit closed form representation. However, this analytical tractability comes as the price of empirical performance, as we demonstrate in~\Cref{sec:HoefEBExperiments}. 
    \end{remark}    

    \begin{remark}
        \label{remark:cs-optimality} 
        Ville's inequality~(\Cref{fact:ville} in ~\Cref{appendix:betting-cs-construction}), used for proving~\Cref{prop:type-I},  is known to be tight for continuous-time nonnegative martingales with infinite quadratic variation, and incurs a slight looseness as we move to the case of discrete time martingales. As a result, the martingale-based CSs constructed in this section provide nearly tight coverage guarantees, that are strictly better than the supermartingale based closed-form CSs discussed in~\Cref{sec:hoeffding-empirical-bernstein}. This near-tightness of the error probability of our betting-based CSs implies that there exists no other CS that is uniformly tighter than ours, while also controlling the error probability below $\alpha$. In other words, our CSs satisfy a notion of admissibility or Pareto-optimality. 
    \end{remark}
    \subsection{Powerful betting strategies}
    \label{sec:powerful-bet-strat}
        Besides validity, we also want the size of the CS to shrink rapidly. This depends on how quickly the values of $W_t(m)$ for $m \neq m^*$ grow with $t$.
        One such criterion is to consider the \emph{growth rate}, 
        i.e., the expected logarithm of the outcome of each bet.
        We can define the \emph{one-step growth rate} $D_n$, 
        for each $n \in [N]$ as follows:
        \begin{align}
            D_{n}(m, \lambda)&\defined  \log (1 + \lambda(Z_t - \mu_t(m))).
        \end{align}
        We are interested in maximizing the expected logarithm of the wealth process \citep{grunwald_safe_testing_2020,shafer_testing_betting_2021}, since it is equivalent to minimizing the expected time for a wealth process to exceed a fixed threshold (asymptotically, as the threshold grows larger) \citep{breiman_optimal_gambling_1961}. Thus, \textit{in the context of the auditing problem, 
        maximizing $\expect[D_t(\lambda, m) \mid \filtration_{t - 1}]$, approximately minimizes $\expect[\tau]$}. The one-step growth rate is a broadly studied objective known as the ``Kelly criterion'' \citep{kelly_new_interpretation_1956}.
        In general, finding the best sequence of bets $\lambda_t(m)$ for different values of $n$ is non-tractable. Instead we consider the approximation $\log(1+x) \geq x - x^2$ for $|x| \leq 1/2$, and define the best constant bet $\lambda^*_n$ in hindsight, as
        \begin{align}
            B_t(m, \lambda) &\defined \lambda \lp Z_t - \mu_t(m)\rp - \lambda^2 \lp Z_t - \mu_t(m) \rp^2, \label{eq:Bn}\\
            \lambda^*_n &\defined \underset{\lambda \in [\pm 1/2c]}{\argmax}\ \frac{1}{n} \sum_{t = 1}^n B_t(m, \lambda),
            \label{eq:approx-lambda-star}
        \end{align}
        where  $c = \max \{|Z_t - \mu_t(m)|: t\in[n]\}$. We get the following result on $\lambda^*_n$ for each $n \in [N]$:
        \begin{gather}
            \lambda^*_n \propto \frac{\sum_{t = 1}^n Z_t - \mu_t(m)}{\sum_{t = 1}^n (Z_t - \mu_t(m))^2} \defined \frac{A_n}{V_n}.
        \end{gather}
        Since $\lambda^*_n$ depends on the $n$th sample itself, $Z_n$, we cannot use this strategy in our CS construction. Instead, at any $n \in [N]$, we can use a predictable approximation of this strategy, that we shall refer to as the \kelly betting strategy. This strategy sets $\lambda_t(m)$ as follows:
        \begin{align}
            \lambda_t(m) = c_t \frac{A_{t - 1}}{V_{t -1}}, \label{eq:approx-kelly} \tag*{(\kelly)}
        \end{align}
        where the (predictable) factor $c_t$ is selected to ensure that $\lambda_t(m) \times \lp Z_t - \mu_t(m) \rp \in (-1, \infty)$, i.e., to satisfy the nonnegativity constraint of $(W_t(m))$.

        \begin{remark}
            \label{remark:other-betting-methods}
            We note there exist several other betting schemes in literature besides \kelly, such as those based on alternative approximations of $\log(1+x)$~\citep{fan_exponential_inequalities_2015, waudby2020estimating, ryu2022confidence}, or the ONS strategy that relies on the exp-concavity of the $\log$-loss~\citep{cutkosky2018black}. In practice, however, we did not observe significant difference in their performance, and we focus on the \kelly\ strategy in this paper due to its conceptual simplicity.
        \end{remark}
    
    \subsection{Logical CS}\label{subsec:logical-CS}
        Irrespective of the choice of the sampling and betting strategies, we can construct a CS that contains $m^*$ with probability 1, based on purely logical considerations. After sampling $t$ transactions, we know that $m^*$ is lower bounded by quantities derived from the the misstatement fraction accumulated in the items we have sampled already. %
        Hence, we can derive the following deterministic bounds:
        \begin{align}
            L_l(t) \defined \sum_{j=1}^t \pi(I_j)f(I_j) \leq m^*, \quad \text{and} \quad 
            U_l(t) \defined L_l(t) + \sum_{i \in \mc{U}_t} \pi(i) \geq m^* .
        \end{align}
        Note that $L_l(t)$~(resp. $U_l(t)$) values are obtained by noting that all the remaining unknown $f$ values must be larger than $0$~(resp. smaller than $1$).
        Additionally, due to the time-uniform nature of confidence sequences,%
         we can intersect the logical CS with a `probabilistic' CS constructed in~\eqref{eq:conf-seq-def-1}, and obtain the following CS:
        \begin{align}
            \label{eq:combined-conf-seq}
            \widetilde{\mc{C}}_t \defined \mc{C}_t \cap [L_\ell(t), U_\ell(t)] \cap \widetilde{\mc{C}}_{t - 1},
        \end{align}
        where $\widetilde{C}_{0} \defined [0,1]$. Note that we may take the running intersection of a CS since it remains a CS, simply by definition. Consequently, the combined CS in \eqref{eq:combined-conf-seq} dominates the probabilistic CS.%
\section{Sampling Strategies}
    \label{sec:sampling-strategies}
    The choice of the sampling strategy, $(q_t)$, is also critical to reducing uncertainty about $m^*$ quickly. Recall that $q_t$ is a probability distribution on the remaining indices $\remain_t$ for each $t \in [N]$. To motivate the choice of our sampling strategy, we first consider the following question: \emph{what is the randomized sampling strategy that leads to the fastest reduction in uncertainty about $m^*$?}
    
    In general, it is difficult to characterize this strategy in closed form~(other than the computational aspect of the strategy being the solution of a multistage optimization problem). Thus, we consider a simplified question, that of finding the sampling strategy that maximizes the expectation of the one-step growth rate, $D_n(\lambda, m)$, for each $n \in [N]$. We seek to maximize the lower bound, $B_n(\lambda, m)$, introduced in~\eqref{eq:Bn}:
    \begin{align}
        q_n^* \defined \underset{q \in \Delta^{\mc{N}_{n}}}{\argmax}\ \expect_{I_n \sim q }\left[ B_n(\lambda, m) \right],  
        \label{eq:max-bound}
    \end{align}
    where $\Delta^{\mc{N}_n}$ is the universe of distributions supported on $\mc{N}_n$. Our next result presents a closed-form characterization of $q_n^*$. 
    \begin{proposition}
        \label{theorem:oracle-strategy}
        Note that $q_n^* = \argmin_{q \in \Delta^{\mc{N}_{n}}}\ \mathbb{V}_{I_n \sim q}[Z_n]$, which implies that $q_n^*(i) \propto \pi(i)f(i)$. Hence, for any valid betting strategy $(\lambda_t)$ and sampling strategy $(q_t)$, 
        we have $\expect_{I \sim q_t}[B_t(\lambda_t,  m) ] \leq \expect_{I \sim q^*_t}[B_t(\lambda_t,  m)]$.
    \end{proposition}
    We defer the proof to \Cref{sec:oracle-proof}, which proceeds by showing  that maximizing the lower bound on the one-step growth rate is equivalent to minimizing the variance of $Z_n$. It turns out that $q_n^*(i) \propto \pi(i)f(i)$ is the minimum~(in fact, zero) variance sampling distribution, and thus, $(q^*_t)$ dominates any other sampling strategy w.r.t.\ maximizing the expected bound on the one-step growth rate.
    \begin{remark}
        The oracle strategy in \Cref{theorem:oracle-strategy} can be considered as a solution of an alternative question: suppose there is an oracle who knows the true values of $f(i)$, and needs to convince an observer that the value $m^*$ is within an interval of width $\varepsilon$ with probability at least $1-\delta$. The oracle wishes to do so by revealing as few $f(i)$ values to the observer as possible. Clearly, any deterministic sampling strategy from the oracle will lead to skepticism from the observer (i.e., the observer will only be convinced once the $\pi(i)$ corresponding to the unrevealed $f(i)$ sum to $\varepsilon$). Hence, the sampling strategy used by the oracle must be random, and according to~\Cref{theorem:oracle-strategy}, it should draw transactions with probability $\propto \pi(i) \times f(i)$. 
    \end{remark}
    \paragraph{Sampling without side information.} Since the $(f(i))$ values are unknown by definition of the problem, we cannot use $(q_t^*)$ in practice. Instead, we consider a sampling strategy that selects a index $i \in \mc{N}_t$ in proportion to its $\pi(i)$ value --- we refer to this strategy as the $\propM$ strategy.  This strategy is also known as ``sampling proportional to size'' in auditing literature~\citep{bickel1992inference}, and is similar to the best deterministic strategy, which queries indices in descending order w.r.t.\ $\pi(i)$.
    \begin{align}
        q_t(i) = \frac{\pi(i)}{\sum_{j \in \mc{N}_t} \pi(j)}, \tag*{(\propM)}
    \end{align} for each $i \in \mc{N}_t$.
    Sampling with \propM minimizes the ``worst case'' support range, and max value, of $Z_t$.
    This allows for the largest possible choice of $\lambda_t$, i.e., our bet.
    
    \paragraph{Using accurate side information for sampling.}
        \Cref{theorem:oracle-strategy} motivates a natural sampling strategy in situations where we have access to side information $(S(i))$ that is known to be a high-fidelity approximation of the true $(f(i))$ values---draw indices proportional to $\pi(i) \times S(i)$. We will refer to this strategy as the $\propMS$ strategy:
        \begin{align}
            q_t(i) = \frac{\pi(i) S(i)}{\sum_{j \in \mc{N}_t} \pi(j) S(j)} ~ .
            \tag*{(\propMS)}
        \end{align}
        Under certain relative accuracy guarantees on the side information, we can characterize the performance achieved by the \propMS strategy as compared to the optimal strategy of~\Cref{theorem:oracle-strategy}, as we state next.  
        \begin{corollary}
            \label{corollary:accurate-side-info}
            Assume that the side information, $(S(i))$, is an accurate prediction of $(f(i))$, i.e., there exists a known parameter $a \in [0,1)$, such that
            \begin{align}
                S(i) / f(i) \in [1 \pm a]
                \label{eq:accurate-side-info-assumption}
            \end{align} for all $i \in [N]$. With the \propMS strategy for $(q_t)$, we can ensure $\expect_{I_t \sim q_t}[B_t(\lambda_t, m)] \geq \expect_{I_t \sim q_t^*}[B_t(\lambda_t, m)]\lp \frac{1}{1+a} \rp^2$, where $(q^*_t)$ is the optimal sampling strategy of \Cref{theorem:oracle-strategy}.
        \end{corollary}
        In many cases, we do not have the accuracy guarantees on the side information required by~\Cref{corollary:accurate-side-info}, and we develop an approach to properly incorporate such side information in the next section. 

\section{Using possibly inaccurate side information}
\label{sec:side information}
    Often, we do not have a uniform guarantee on accuracy on $(S(i))$ as we assumed in the previous section. In such cases, we cannot continue to use the \propMS strategy, as it requires 
    knowledge of the range of $f(i)/S(i)$ 
    in order to select the betting fractions 
    that ensure non-negativity of the process $(W_t(m))$. 
    Nevertheless, we develop new techniques in this section 
    that can exploit the side information 
    without the uniform accuracy guarantees, 
    provided that the side information 
    is correlated with the unknown $(f(i))$ values. 
    In particular, the method developed in this section 
    for incorporating the side information 
    is orthogonal to the choice of the sampling strategy; 
    and thus, it can be combined with any sampling strategy 
    that ensures the non-negativity of the process $(W_t(m))$.  
    
    Our approach is based on the idea of control variates~\citep[\S~V.2]{asmussen2007stochastic} that are used to reduce the variance of Monte Carlo~(MC) estimates of an unknown quantity, using some correlated side information whose expected value is known.  More specifically, let $\mhat$ denote an unbiased estimate of an unknown parameter $m$, and let $\vhat$ denote another (possibly correlated to $\mhat$) statistic with zero mean. Then, the new statistic, $\mhat_\beta = \mhat + \beta \vhat$ is also an unbiased estimate of $m$, for all $\beta \in \mathbb{R}$. Furthermore, it is easy to check that $\var(\mhat_\beta) = \var(\mhat) + \beta^2\var(\vhat) + 2 \beta \cov(\mhat, \vhat)$, which implies that the variance of this new estimate is minimized at  $\beta=\beta^* \defined - \lp \cov(\mhat, \vhat) / \var(\vhat) \rp$. Finally, note that the variance of $\mhat_{\beta^*}$ cannot be larger than the variance of the original estimate $\mhat$, since $\var(\mhat_{\beta^*}) \leq \var(\mhat_0) = \var(\mhat)$ by the definition of $\beta^*$.

    Returning to our problem, given some possibly inaccurate side information $(S(i))$,
     define the control variate~(that is, an analog of the term $\vhat$)  as 
    \begin{align}
        U_t \defined S(I_t) - \mathbb{E}_{I' \sim q_t}[S(I')], 
    \end{align}
    and let $(\beta_t)$ denote a sequence of predictable terms taking values in $[-1, 1]$ used to weigh the effect of $(U_t)$. Note that, similar to $\vhat$, the term  $U_t$ has zero mean for each $t \in [N]$. We now define the wealth process with control variates, denoted by $(\Wtilde_t(m))$, and its corresponding CS as follows: 
     \begin{align}
         \Wtilde_t(m) &\defined \prod_{t=1}^n \lp 1 + \lambda_t(m) (Z_t + \beta_t U_t - \mu_t(m)) \rp,\\
         \mc{C}_t &= \{m \in [0,1]: \Wtilde_n(m) < 1/\alpha \},
        \label{def:CS-control-variates}
     \end{align}
     where $(\lambda_t(m))$ is a betting strategy for each $m \in [0, 1]$.%
     \begin{theorem}
        For any set of side information $(S(i))$, sequence $(\beta_t)$, sampling strategy $(q_t)$, and betting strategies $(\lambda_t(m))$, $(\mc{C}_t)$ as defined in \eqref{def:CS-control-variates} is an $(1 - \delta)$-CS for $m^*$. Consequently, the procedure that produces $\mc{C}_\tau$ is an $\rlfa$.
        \label{thm:SideCSthm}
    \end{theorem}
    The discussion above suggests that by a suitable choice of the parameters $(\beta_t)$, we can reduce the variance of the first term. To see why this is desirable, recall that the optimal value  of the approximate growth rate after $n$ steps of the new wealth process satisfies the following:
    \begin{gather}
        \begin{aligned}
            \widetilde{B}_n(\lambda, m)\coloneqq &\lambda (Z_t + \beta_t U_t  - \mu_t(m))\\
            &- \lambda^2 (Z_t + \beta_tU_t - \mu_t(m))^2, %
        \end{aligned}\\
            \max_{\lambda} \widetilde{B}_n(\lambda, m)  \propto \frac{\sum_{t = 1}^n Z_t + \beta_tU_t - \mu_t(m)}{\sum_{t = 1}^n (Z_t + \beta_tU_t - \mu_t(m))^2}.
        \label{eq:control-variates-exponent-1}
    \end{gather}
    Note that by setting $\beta_t=0$ for all $t \in [n]$, we recover $\widetilde{B}_n(\lambda, m) = B_n(\lambda, m)$, i.e., the  wealth lower bound with no side information. Next, we observe that $\sum_{t=1}^n \beta_t U_t$ concentrates strongly around its mean~($0$).
    \begin{proposition}
        \label{prop:control-variates-1}
        For any $\delta \in (0,1)$ and sequence $(\beta_t)$, the following statement is simultaneously true for all $n \in [N]$ with probability at least $1-\delta$
        \begin{align}
            \left \lvert \frac{1}{n}\sum_{t = 1}^n \beta_tU_t \right \rvert = \mc{O}\lp \sqrt{ \frac{ \log(\log n /\delta)}{n} } \rp.
        \end{align}
    \end{proposition}
    This result, proved in~\Cref{proof:control-variates-1}, implies that in order to select the parameters $(\beta_t)$, we can focus on its effect on the second order term in the denominator. In particular, the best value of $\beta$ for the first $n$ observations, is the one that minimizes the denominator, and can be defined as follows:
    \begin{align}
        \beta^*_n &\defined  \underset{\beta \in [-1,1]}{\argmin} \; \sum_{t=1}^n {(Z_t - \mu_t(m) + \beta U_t)^2}
                  \propto-\frac{\sum_{t = 1}^n(Z_t - \mu_t(m))U_t}{\sum_{t = 1}^nU_t^2}.
    \end{align}
    The numerator of $\beta_n^*$ varies with $\sum_{t = 1}^n f(I_t)S(I_t)$---hence, the magnitude of $\beta_t$ increases with the amount of correlation between $f(i)$ and $S(i)
    $.
    Since $\beta^*_n$ is not predictable (it is $\mc{F}_n$ instead of $\mc{F}_{n-1}$ measurable), we will use the following strategy of approximating $\beta_n^*$ at each $n \in [N]$:
    \begin{align}
        \beta_n \propto -\frac{\sum_{t = 1}^{n - 1}(Z_t - \mu_t(m))U_t}{\sum_{t = 1}^{n - 1}U_t^2}.
    \end{align} for $n \geq 2$ and we let $\beta_1 = 0$. This provides a principled way of incorporating side information even when the relationship between the side information and the ground truth is unclear.

    \begin{remark}
        \label{remark:risk-vs-correlation}
        Our work is motivated by applications where the side-information is generated by an ML model trained on historical transaction data. In practice, ML models are trained via empirical risk minimization, and we expect that models with lower risk should result in side-information with higher correlation. For some simple cases, such as least-squares linear regressors, we can obtain a precise relation between correlation and risk: $\rho^2 = 1- MSE/S_y$, where $S_y$ is the empirical variance of the target variable $y$ used in training the model.  Characterizing this relation for more general models is left for future work. 
    \end{remark}
\section{Experiments}
\label{sec:experiments}
    We conduct simulations of our RLFA methods on a variety of scenarios for $\pi$ and $f$. For each simulation setup, we choose two positive integers $\Nlarge$ and $\Nsmall$ such that $\Nlarge+\Nsmall=N$. We generate the weight distribution $\pi$, consisting of $\Nlarge$ `large' values and $\Nsmall$ `small' values. The exact range of values taken by these terms are varied across experiments, but on an average the ratio of `large' over `small' $\pi$ values lie between $10$ and $10^3$. We then generate the $f$ values in one of two ways: (1) $f \propto \pi$, where indices with where large $\pi$ values take $f$ values in $[0.4, 0.5]$ and small $\pi$ values take on $f$ values in $[0.001, 0.01]$, or (2) $f \propto 1/\pi$, where the $f$ value ranges are swapped for large and small values. The simulations in this section focus on the different sampling strategies as well as the efficacy of control variates --- we provide additional experiments comparing the betting CS with other types of CS in \Cref{sec:HoefEBExperiments}.
    \begin{figure}[htb!]
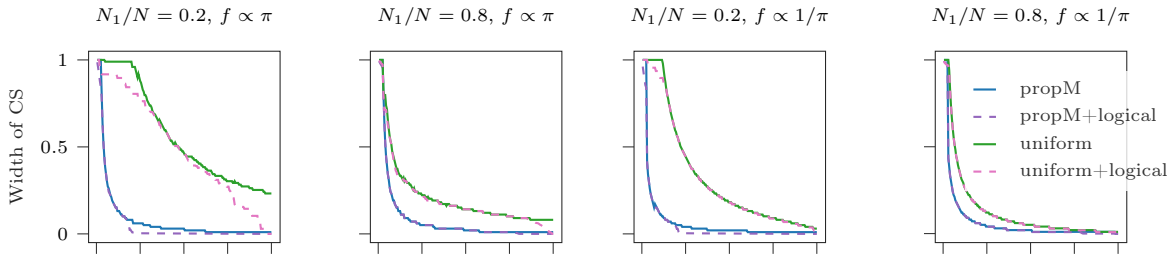

        \centering
            \def\figwidth{0.25\hsize}
            \def\figheight{0.25\hsize} %
            \begin{tabular}{cccc}
            \input{Figures/NoSideInfoCSf_propto_M_large_40} &
            \input{Figures/NoSideInfoCSf_propto_M_large_160} & 
            \input{Figures/NoSideInfoCSf_inv_propto_M_large_40} &
            \input{Figures/NoSideInfoCSf_inv_propto_M_large_160}
            \end{tabular}
            \caption{A comparison of \propM vs. uniform sampling distributions, and the impact of intersecting with the logical CS (\Cref{subsec:logical-CS}) on the width of the CSs, when $\delta = 0.05$. The \propM strategy produces tighter CSs that results, and intersecting with the logical CS further reduces the width, particularly when few transactions are large ($\Nlarge = 0.2$).}
            \label{fig:no-side-info-CS}
    \end{figure}

     \begin{figure}[htb!]
        \centering
            \def\figwidth{0.25\hsize}
            \def\figheight{0.25\hsize} %
            \begin{tabular}{cccc}
             \begin{tikzpicture}

\definecolor{darkslategray38}{RGB}{38,38,38}
\definecolor{forestgreen4416044}{RGB}{44,160,44}
\definecolor{lightgray204}{RGB}{204,204,204}
\definecolor{mediumpurple148103189}{RGB}{148,103,189}
\definecolor{orchid227119194}{RGB}{227,119,194}
\definecolor{steelblue31119180}{RGB}{31,119,180}

\tikzstyle{every node}=[font=\scriptsize]
\begin{axis}[
axis line style={darkslategray38},
height=\figheight,
legend cell align={left},
legend style={
  fill opacity=0.2,
  draw opacity=1,
  text opacity=1,
  at={(0.08,0.97)},
  anchor=north west,
  draw=none
},
tick align=outside,
tick pos=left,
title={$N_1/N=0.2$,    \(\displaystyle f \propto \pi\)},
width=\figwidth,
x grid style={lightgray204},
xmin=25.675, xmax=208.825,
xtick style={color=darkslategray38},
y grid style={lightgray204},
ylabel=\textcolor{darkslategray38}{Density},
ymin=0, ymax=0.5,
ytick style={color=darkslategray38}, 
yticklabels={,,}
]
\draw[draw=none,fill=steelblue31119180,fill opacity=0.7] (axis cs:53,0) rectangle (axis cs:59.1,0.101311475409836);
\addlegendimage{ybar,ybar legend,draw=none,fill=steelblue31119180,fill opacity=0.7}
\addlegendentry{propM}

\draw[draw=none,fill=steelblue31119180,fill opacity=0.7] (axis cs:59.1,0) rectangle (axis cs:65.2,0.0554098360655738);
\draw[draw=none,fill=steelblue31119180,fill opacity=0.7] (axis cs:65.2,0) rectangle (axis cs:71.3,0.00360655737704918);
\draw[draw=none,fill=steelblue31119180,fill opacity=0.7] (axis cs:71.3,0) rectangle (axis cs:77.4,0.00163934426229508);
\draw[draw=none,fill=steelblue31119180,fill opacity=0.7] (axis cs:77.4,0) rectangle (axis cs:83.5,0.00131147540983607);
\draw[draw=none,fill=steelblue31119180,fill opacity=0.7] (axis cs:83.5,0) rectangle (axis cs:89.6,0.000327868852459017);
\draw[draw=none,fill=steelblue31119180,fill opacity=0.7] (axis cs:89.6,0) rectangle (axis cs:95.7,0);
\draw[draw=none,fill=steelblue31119180,fill opacity=0.7] (axis cs:95.7,0) rectangle (axis cs:101.8,0);
\draw[draw=none,fill=steelblue31119180,fill opacity=0.7] (axis cs:101.8,0) rectangle (axis cs:107.9,0);
\draw[draw=none,fill=steelblue31119180,fill opacity=0.7] (axis cs:107.9,0) rectangle (axis cs:114,0.000327868852459017);
\draw[draw=none,fill=mediumpurple148103189,fill opacity=0.7] (axis cs:34,0) rectangle (axis cs:34.7,0.0285714285714285);
\addlegendimage{ybar,ybar legend,draw=none,fill=mediumpurple148103189,fill opacity=0.7}
\addlegendentry{propM+logical}

\draw[draw=none,fill=mediumpurple148103189,fill opacity=0.7] (axis cs:34.7,0) rectangle (axis cs:35.4,0.111428571428572);
\draw[draw=none,fill=mediumpurple148103189,fill opacity=0.7] (axis cs:35.4,0) rectangle (axis cs:36.1,0.305714285714284);
\draw[draw=none,fill=mediumpurple148103189,fill opacity=0.7] (axis cs:36.1,0) rectangle (axis cs:36.8,0);
\draw[draw=none,fill=mediumpurple148103189,fill opacity=0.7] (axis cs:36.8,0) rectangle (axis cs:37.5,0.457142857142855);
\draw[draw=none,fill=mediumpurple148103189,fill opacity=0.7] (axis cs:37.5,0) rectangle (axis cs:38.2,0.322857142857142);
\draw[draw=none,fill=mediumpurple148103189,fill opacity=0.7] (axis cs:38.2,0) rectangle (axis cs:38.9,0);
\draw[draw=none,fill=mediumpurple148103189,fill opacity=0.7] (axis cs:38.9,0) rectangle (axis cs:39.6,0.157142857142857);
\draw[draw=none,fill=mediumpurple148103189,fill opacity=0.7] (axis cs:39.6,0) rectangle (axis cs:40.3,0.0428571428571431);
\draw[draw=none,fill=mediumpurple148103189,fill opacity=0.7] (axis cs:40.3,0) rectangle (axis cs:41,0.00285714285714285);
\draw[draw=none,fill=forestgreen4416044,fill opacity=0.7] (axis cs:199.5,0) rectangle (axis cs:199.6,0);
\addlegendimage{ybar,ybar legend,draw=none,fill=forestgreen4416044,fill opacity=0.7}
\addlegendentry{uniform}

\draw[draw=none,fill=forestgreen4416044,fill opacity=0.7] (axis cs:196.6,0) rectangle (axis cs:199.7,0);
\draw[draw=none,fill=forestgreen4416044,fill opacity=0.7] (axis cs:199.7,0) rectangle (axis cs:199.8,0);
\draw[draw=none,fill=forestgreen4416044,fill opacity=0.7] (axis cs:199.8,0) rectangle (axis cs:199.9,0);
\draw[draw=none,fill=forestgreen4416044,fill opacity=0.7] (axis cs:199.9,0) rectangle (axis cs:200,0);
\draw[draw=none,fill=forestgreen4416044,fill opacity=0.7] (axis cs:200.1,0) rectangle (axis cs:200.2,0);
\draw[draw=none,fill=forestgreen4416044,fill opacity=0.7] (axis cs:200.2,0) rectangle (axis cs:200.3,0);
\draw[draw=none,fill=forestgreen4416044,fill opacity=0.7] (axis cs:200.3,0) rectangle (axis cs:200.4,0);
\draw[draw=none,fill=forestgreen4416044,fill opacity=0.7] (axis cs:200.4,0) rectangle (axis cs:200.5,0);
\draw[draw=none,fill=orchid227119194,fill opacity=0.7] (axis cs:127,0) rectangle (axis cs:134.2,0.000277777777777778);
\addlegendimage{ybar,ybar legend,draw=none,fill=orchid227119194,fill opacity=0.7}
\addlegendentry{uniform+logical}

\draw[draw=none,fill=orchid227119194,fill opacity=0.7] (axis cs:134.2,0) rectangle (axis cs:141.4,0.000277777777777777);
\draw[draw=none,fill=orchid227119194,fill opacity=0.7] (axis cs:141.4,0) rectangle (axis cs:148.6,0);
\draw[draw=none,fill=orchid227119194,fill opacity=0.7] (axis cs:148.6,0) rectangle (axis cs:155.8,0.000277777777777777);
\draw[draw=none,fill=orchid227119194,fill opacity=0.7] (axis cs:155.8,0) rectangle (axis cs:163,0.000833333333333335);
\draw[draw=none,fill=orchid227119194,fill opacity=0.7] (axis cs:163,0) rectangle (axis cs:170.2,0.00333333333333334);
\draw[draw=none,fill=orchid227119194,fill opacity=0.7] (axis cs:170.2,0) rectangle (axis cs:177.4,0.00777777777777776);
\draw[draw=none,fill=orchid227119194,fill opacity=0.7] (axis cs:177.4,0) rectangle (axis cs:184.6,0.0247222222222223);
\draw[draw=none,fill=orchid227119194,fill opacity=0.7] (axis cs:184.6,0) rectangle (axis cs:191.8,0.0433333333333332);
\draw[draw=none,fill=orchid227119194,fill opacity=0.7] (axis cs:191.8,0) rectangle (axis cs:199,0.0580555555555557);
\legend{}
\end{axis}

\end{tikzpicture} &
             \begin{tikzpicture}

\definecolor{darkslategray38}{RGB}{38,38,38}
\definecolor{forestgreen4416044}{RGB}{44,160,44}
\definecolor{lightgray204}{RGB}{204,204,204}
\definecolor{mediumpurple148103189}{RGB}{148,103,189}
\definecolor{orchid227119194}{RGB}{227,119,194}
\definecolor{steelblue31119180}{RGB}{31,119,180}

\tikzstyle{every node}=[font=\scriptsize]
\begin{axis}[
axis line style={darkslategray38},
height=\figheight,
legend cell align={left},
legend style={
  fill opacity=0.8,
  draw opacity=1,
  text opacity=1,
  at={(0.03,0.97)},
  anchor=north west,
  draw=none
},
tick align=outside,
tick pos=left,
title={
$N_1/N=0.8$,    \(\displaystyle f \propto \pi\)},
width=\figwidth,
x grid style={lightgray204},
xmin=45.625, xmax=207.875,
xtick style={color=darkslategray38},
y grid style={lightgray204},
ylabel=\textcolor{darkslategray38}{},
ymin=0, ymax=0.2,
ytick style={color=darkslategray38}, 
yticklabels={,,}
]
\draw[draw=none,fill=steelblue31119180,fill opacity=0.7] (axis cs:53,0) rectangle (axis cs:55.8,0.0657142857142858);
\addlegendimage{ybar,ybar legend,draw=none,fill=steelblue31119180,fill opacity=0.7}
\addlegendentry{propM}

\draw[draw=none,fill=steelblue31119180,fill opacity=0.7] (axis cs:55.8,0) rectangle (axis cs:58.6,0.0985714285714284);
\draw[draw=none,fill=steelblue31119180,fill opacity=0.7] (axis cs:58.6,0) rectangle (axis cs:61.4,0.0907142857142858);
\draw[draw=none,fill=steelblue31119180,fill opacity=0.7] (axis cs:61.4,0) rectangle (axis cs:64.2,0.0849999999999999);
\draw[draw=none,fill=steelblue31119180,fill opacity=0.7] (axis cs:64.2,0) rectangle (axis cs:67,0.0128571428571429);
\draw[draw=none,fill=steelblue31119180,fill opacity=0.7] (axis cs:67,0) rectangle (axis cs:69.8,0.00357142857142858);
\draw[draw=none,fill=steelblue31119180,fill opacity=0.7] (axis cs:69.8,0) rectangle (axis cs:72.6,0);
\draw[draw=none,fill=steelblue31119180,fill opacity=0.7] (axis cs:72.6,0) rectangle (axis cs:75.4,0);
\draw[draw=none,fill=steelblue31119180,fill opacity=0.7] (axis cs:75.4,0) rectangle (axis cs:78.2,0);
\draw[draw=none,fill=steelblue31119180,fill opacity=0.7] (axis cs:78.2,0) rectangle (axis cs:81,0.000714285714285715);
\draw[draw=none,fill=mediumpurple148103189,fill opacity=0.7] (axis cs:53,0) rectangle (axis cs:55.8,0.0657142857142858);
\addlegendimage{ybar,ybar legend,draw=none,fill=mediumpurple148103189,fill opacity=0.7}
\addlegendentry{propM+logical}

\draw[draw=none,fill=mediumpurple148103189,fill opacity=0.7] (axis cs:55.8,0) rectangle (axis cs:58.6,0.0992857142857141);
\draw[draw=none,fill=mediumpurple148103189,fill opacity=0.7] (axis cs:58.6,0) rectangle (axis cs:61.4,0.0900000000000001);
\draw[draw=none,fill=mediumpurple148103189,fill opacity=0.7] (axis cs:61.4,0) rectangle (axis cs:64.2,0.0849999999999999);
\draw[draw=none,fill=mediumpurple148103189,fill opacity=0.7] (axis cs:64.2,0) rectangle (axis cs:67,0.0128571428571429);
\draw[draw=none,fill=mediumpurple148103189,fill opacity=0.7] (axis cs:67,0) rectangle (axis cs:69.8,0.00357142857142858);
\draw[draw=none,fill=mediumpurple148103189,fill opacity=0.7] (axis cs:69.8,0) rectangle (axis cs:72.6,0);
\draw[draw=none,fill=mediumpurple148103189,fill opacity=0.7] (axis cs:72.6,0) rectangle (axis cs:75.4,0);
\draw[draw=none,fill=mediumpurple148103189,fill opacity=0.7] (axis cs:75.4,0) rectangle (axis cs:78.2,0);
\draw[draw=none,fill=mediumpurple148103189,fill opacity=0.7] (axis cs:78.2,0) rectangle (axis cs:81,0.000714285714285715);
\draw[draw=none,fill=forestgreen4416044,fill opacity=0.7] (axis cs:199.5,0) rectangle (axis cs:199.6,0);
\addlegendimage{ybar,ybar legend,draw=none,fill=forestgreen4416044,fill opacity=0.7}
\addlegendentry{uniform}

\draw[draw=none,fill=forestgreen4416044,fill opacity=0.7] (axis cs:189.6,0) rectangle (axis cs:199.7,0);
\draw[draw=none,fill=forestgreen4416044,fill opacity=0.7] (axis cs:199.7,0) rectangle (axis cs:199.8,0);
\draw[draw=none,fill=forestgreen4416044,fill opacity=0.7] (axis cs:199.8,0) rectangle (axis cs:199.9,0);
\draw[draw=none,fill=forestgreen4416044,fill opacity=0.7] (axis cs:199.9,0) rectangle (axis cs:200,0);
\draw[draw=none,fill=forestgreen4416044,fill opacity=0.7] (axis cs:200.1,0) rectangle (axis cs:200.2,0);
\draw[draw=none,fill=forestgreen4416044,fill opacity=0.7] (axis cs:200.2,0) rectangle (axis cs:200.3,0);
\draw[draw=none,fill=forestgreen4416044,fill opacity=0.7] (axis cs:200.3,0) rectangle (axis cs:200.4,0);
\draw[draw=none,fill=forestgreen4416044,fill opacity=0.7] (axis cs:200.4,0) rectangle (axis cs:200.5,0);
\draw[draw=none,fill=orchid227119194,fill opacity=0.7] (axis cs:90,0) rectangle (axis cs:100.5,0.00019047619047619);
\addlegendimage{ybar,ybar legend,draw=none,fill=orchid227119194,fill opacity=0.7}
\addlegendentry{uniform+logical}

\draw[draw=none,fill=orchid227119194,fill opacity=0.7] (axis cs:100.5,0) rectangle (axis cs:111,0);
\draw[draw=none,fill=orchid227119194,fill opacity=0.7] (axis cs:111,0) rectangle (axis cs:121.5,0);
\draw[draw=none,fill=orchid227119194,fill opacity=0.7] (axis cs:121.5,0) rectangle (axis cs:132,0.000380952380952381);
\draw[draw=none,fill=orchid227119194,fill opacity=0.7] (axis cs:132,0) rectangle (axis cs:142.5,0.00019047619047619);
\draw[draw=none,fill=orchid227119194,fill opacity=0.7] (axis cs:142.5,0) rectangle (axis cs:153,0.00019047619047619);
\draw[draw=none,fill=orchid227119194,fill opacity=0.7] (axis cs:153,0) rectangle (axis cs:163.5,0.000571428571428571);
\draw[draw=none,fill=orchid227119194,fill opacity=0.7] (axis cs:163.5,0) rectangle (axis cs:174,0.000761904761904762);
\draw[draw=none,fill=orchid227119194,fill opacity=0.7] (axis cs:174,0) rectangle (axis cs:184.5,0.00819047619047619);
\draw[draw=none,fill=orchid227119194,fill opacity=0.7] (axis cs:184.5,0) rectangle (axis cs:195,0.0847619047619048);
\legend{};
\end{axis}
\end{tikzpicture} &
             \begin{tikzpicture}

\definecolor{darkslategray38}{RGB}{38,38,38}
\definecolor{forestgreen4416044}{RGB}{44,160,44}
\definecolor{lightgray204}{RGB}{204,204,204}
\definecolor{mediumpurple148103189}{RGB}{148,103,189}
\definecolor{orchid227119194}{RGB}{227,119,194}
\definecolor{steelblue31119180}{RGB}{31,119,180}

\tikzstyle{every node}=[font=\scriptsize]
\begin{axis}[
axis line style={darkslategray38},
height=\figheight,
legend cell align={left},
legend style={fill opacity=0.8, draw opacity=1, text opacity=1, draw=none},
tick align=outside,
tick pos=left,
title={
$N_1/N=0.2$,    \(\displaystyle f \propto 1/\pi\)},
width=\figwidth,
x grid style={lightgray204},
xmin=22.8, xmax=203.2,
xtick style={color=darkslategray38},
y grid style={lightgray204},
ymin=0, ymax=0.6,
ytick style={color=darkslategray38}, 
yticklabels={,,}
]
\draw[draw=none,fill=steelblue31119180,fill opacity=0.7] (axis cs:36,0) rectangle (axis cs:39,0.222666666666667);
\addlegendimage{ybar,ybar legend,draw=none,fill=steelblue31119180,fill opacity=0.7}
\addlegendentry{propM}

\draw[draw=none,fill=steelblue31119180,fill opacity=0.7] (axis cs:39,0) rectangle (axis cs:42,0.0586666666666667);
\draw[draw=none,fill=steelblue31119180,fill opacity=0.7] (axis cs:42,0) rectangle (axis cs:45,0.0233333333333333);
\draw[draw=none,fill=steelblue31119180,fill opacity=0.7] (axis cs:45,0) rectangle (axis cs:48,0.0126666666666667);
\draw[draw=none,fill=steelblue31119180,fill opacity=0.7] (axis cs:48,0) rectangle (axis cs:51,0.008);
\draw[draw=none,fill=steelblue31119180,fill opacity=0.7] (axis cs:51,0) rectangle (axis cs:54,0.00266666666666667);
\draw[draw=none,fill=steelblue31119180,fill opacity=0.7] (axis cs:54,0) rectangle (axis cs:57,0.00266666666666667);
\draw[draw=none,fill=steelblue31119180,fill opacity=0.7] (axis cs:57,0) rectangle (axis cs:60,0.00133333333333333);
\draw[draw=none,fill=steelblue31119180,fill opacity=0.7] (axis cs:60,0) rectangle (axis cs:63,0.000666666666666667);
\draw[draw=none,fill=steelblue31119180,fill opacity=0.7] (axis cs:63,0) rectangle (axis cs:66,0.000666666666666667);
\draw[draw=none,fill=mediumpurple148103189,fill opacity=0.7] (axis cs:31,0) rectangle (axis cs:32,0.024);
\addlegendimage{ybar,ybar legend,draw=none,fill=mediumpurple148103189,fill opacity=0.7}
\addlegendentry{propM+logical}

\draw[draw=none,fill=mediumpurple148103189,fill opacity=0.7] (axis cs:32,0) rectangle (axis cs:33,0.554);
\draw[draw=none,fill=mediumpurple148103189,fill opacity=0.7] (axis cs:33,0) rectangle (axis cs:34,0.174);
\draw[draw=none,fill=mediumpurple148103189,fill opacity=0.7] (axis cs:34,0) rectangle (axis cs:35,0.066);
\draw[draw=none,fill=mediumpurple148103189,fill opacity=0.7] (axis cs:35,0) rectangle (axis cs:36,0.036);
\draw[draw=none,fill=mediumpurple148103189,fill opacity=0.7] (axis cs:36,0) rectangle (axis cs:37,0.038);
\draw[draw=none,fill=mediumpurple148103189,fill opacity=0.7] (axis cs:37,0) rectangle (axis cs:38,0.044);
\draw[draw=none,fill=mediumpurple148103189,fill opacity=0.7] (axis cs:38,0) rectangle (axis cs:39,0.036);
\draw[draw=none,fill=mediumpurple148103189,fill opacity=0.7] (axis cs:39,0) rectangle (axis cs:40,0.02);
\draw[draw=none,fill=mediumpurple148103189,fill opacity=0.7] (axis cs:40,0) rectangle (axis cs:41,0.008);
\draw[draw=none,fill=forestgreen4416044,fill opacity=0.7] (axis cs:172,0) rectangle (axis cs:174.3,0.0034782608695652);
\addlegendimage{ybar,ybar legend,draw=none,fill=forestgreen4416044,fill opacity=0.7}
\addlegendentry{uniform}

\draw[draw=none,fill=forestgreen4416044,fill opacity=0.7] (axis cs:174.3,0) rectangle (axis cs:176.6,0.0078260869565218);
\draw[draw=none,fill=forestgreen4416044,fill opacity=0.7] (axis cs:176.6,0) rectangle (axis cs:178.9,0.0104347826086956);
\draw[draw=none,fill=forestgreen4416044,fill opacity=0.7] (axis cs:178.9,0) rectangle (axis cs:181.2,0.0426086956521742);
\draw[draw=none,fill=forestgreen4416044,fill opacity=0.7] (axis cs:181.2,0) rectangle (axis cs:183.5,0.0660869565217388);
\draw[draw=none,fill=forestgreen4416044,fill opacity=0.7] (axis cs:183.5,0) rectangle (axis cs:185.8,0.0913043478260865);
\draw[draw=none,fill=forestgreen4416044,fill opacity=0.7] (axis cs:185.8,0) rectangle (axis cs:188.1,0.142608695652175);
\draw[draw=none,fill=forestgreen4416044,fill opacity=0.7] (axis cs:188.1,0) rectangle (axis cs:190.4,0.043478260869565);
\draw[draw=none,fill=forestgreen4416044,fill opacity=0.7] (axis cs:190.4,0) rectangle (axis cs:192.7,0.0226086956521741);
\draw[draw=none,fill=forestgreen4416044,fill opacity=0.7] (axis cs:192.7,0) rectangle (axis cs:195,0.0043478260869565);
\draw[draw=none,fill=orchid227119194,fill opacity=0.7] (axis cs:159,0) rectangle (axis cs:162,0.00266666666666667);
\addlegendimage{ybar,ybar legend,draw=none,fill=orchid227119194,fill opacity=0.7}
\addlegendentry{uniform+logical}

\draw[draw=none,fill=orchid227119194,fill opacity=0.7] (axis cs:162,0) rectangle (axis cs:165,0.002);
\draw[draw=none,fill=orchid227119194,fill opacity=0.7] (axis cs:165,0) rectangle (axis cs:168,0.00466666666666667);
\draw[draw=none,fill=orchid227119194,fill opacity=0.7] (axis cs:168,0) rectangle (axis cs:171,0.014);
\draw[draw=none,fill=orchid227119194,fill opacity=0.7] (axis cs:171,0) rectangle (axis cs:174,0.034);
\draw[draw=none,fill=orchid227119194,fill opacity=0.7] (axis cs:174,0) rectangle (axis cs:177,0.0733333333333333);
\draw[draw=none,fill=orchid227119194,fill opacity=0.7] (axis cs:177,0) rectangle (axis cs:180,0.102666666666667);
\draw[draw=none,fill=orchid227119194,fill opacity=0.7] (axis cs:180,0) rectangle (axis cs:183,0.0693333333333333);
\draw[draw=none,fill=orchid227119194,fill opacity=0.7] (axis cs:183,0) rectangle (axis cs:186,0.026);
\draw[draw=none,fill=orchid227119194,fill opacity=0.7] (axis cs:186,0) rectangle (axis cs:189,0.00466666666666667);
\legend{};
\end{axis}

\end{tikzpicture} & 
             \begin{tikzpicture}

\definecolor{darkslategray38}{RGB}{38,38,38}
\definecolor{forestgreen4416044}{RGB}{44,160,44}
\definecolor{lightgray204}{RGB}{204,204,204}
\definecolor{mediumpurple148103189}{RGB}{148,103,189}
\definecolor{orchid227119194}{RGB}{227,119,194}
\definecolor{steelblue31119180}{RGB}{31,119,180}

\tikzstyle{every node}=[font=\scriptsize]
\begin{axis}[
axis line style={darkslategray38},
height=\figheight,
legend cell align={left},
legend style={fill opacity=0.8, draw opacity=1, text opacity=1, draw=none, at={(0.25,0.9)}, legend columns=2, anchor=north west},
tick align=outside,
tick pos=left,
title={ $N_1/N=0.8$,    \(\displaystyle f \propto 1/\pi\)},
width=\figwidth,
x grid style={lightgray204},
xmin=37.65, xmax=111.35,
xtick style={color=darkslategray38},
y grid style={lightgray204},
ylabel=\textcolor{darkslategray38}{},
ymin=0, ymax=0.65,
ytick style={color=darkslategray38}, 
yticklabels={,,}
]
\draw[draw=none,fill=steelblue31119180,fill opacity=0.7] (axis cs:48,0) rectangle (axis cs:49.3,0.603076923076924);
\addlegendimage{ybar,ybar legend,draw=none,fill=steelblue31119180,fill opacity=0.7}
\addlegendentry{propM}

\draw[draw=none,fill=steelblue31119180,fill opacity=0.7] (axis cs:49.3,0) rectangle (axis cs:50.6,0.155384615384615);
\draw[draw=none,fill=steelblue31119180,fill opacity=0.7] (axis cs:50.6,0) rectangle (axis cs:51.9,0);
\draw[draw=none,fill=steelblue31119180,fill opacity=0.7] (axis cs:51.9,0) rectangle (axis cs:53.2,0);
\draw[draw=none,fill=steelblue31119180,fill opacity=0.7] (axis cs:53.2,0) rectangle (axis cs:54.5,0);
\draw[draw=none,fill=steelblue31119180,fill opacity=0.7] (axis cs:54.5,0) rectangle (axis cs:55.8,0.00461538461538463);
\draw[draw=none,fill=steelblue31119180,fill opacity=0.7] (axis cs:55.8,0) rectangle (axis cs:57.1,0.0046153846153846);
\draw[draw=none,fill=steelblue31119180,fill opacity=0.7] (axis cs:57.1,0) rectangle (axis cs:58.4,0);
\draw[draw=none,fill=steelblue31119180,fill opacity=0.7] (axis cs:58.4,0) rectangle (axis cs:59.7,0);
\draw[draw=none,fill=steelblue31119180,fill opacity=0.7] (axis cs:59.7,0) rectangle (axis cs:61,0.00153846153846154);
\draw[draw=none,fill=mediumpurple148103189,fill opacity=0.7] (axis cs:41,0) rectangle (axis cs:42.2,0.0833333333333331);
\addlegendimage{ybar,ybar legend,draw=none,fill=mediumpurple148103189,fill opacity=0.7}
\addlegendentry{propM+logical}

\draw[draw=none,fill=mediumpurple148103189,fill opacity=0.7] (axis cs:42.2,0) rectangle (axis cs:43.4,0.0750000000000003);
\draw[draw=none,fill=mediumpurple148103189,fill opacity=0.7] (axis cs:43.4,0) rectangle (axis cs:44.6,0.115);
\draw[draw=none,fill=mediumpurple148103189,fill opacity=0.7] (axis cs:44.6,0) rectangle (axis cs:45.8,0.153333333333334);
\draw[draw=none,fill=mediumpurple148103189,fill opacity=0.7] (axis cs:45.8,0) rectangle (axis cs:47,0.13);
\draw[draw=none,fill=mediumpurple148103189,fill opacity=0.7] (axis cs:47,0) rectangle (axis cs:48.2,0.168333333333333);
\draw[draw=none,fill=mediumpurple148103189,fill opacity=0.7] (axis cs:48.2,0) rectangle (axis cs:49.4,0.0733333333333336);
\draw[draw=none,fill=mediumpurple148103189,fill opacity=0.7] (axis cs:49.4,0) rectangle (axis cs:50.6,0.0333333333333333);
\draw[draw=none,fill=mediumpurple148103189,fill opacity=0.7] (axis cs:50.6,0) rectangle (axis cs:51.8,0);
\draw[draw=none,fill=mediumpurple148103189,fill opacity=0.7] (axis cs:51.8,0) rectangle (axis cs:53,0.00166666666666666);
\draw[draw=none,fill=forestgreen4416044,fill opacity=0.7] (axis cs:92,0) rectangle (axis cs:93.6,0.00500000000000002);
\addlegendimage{ybar,ybar legend,draw=none,fill=forestgreen4416044,fill opacity=0.7}
\addlegendentry{uniform}; 

\draw[draw=none,fill=forestgreen4416044,fill opacity=0.7] (axis cs:93.6,0) rectangle (axis cs:95.2,0.0274999999999999);
\draw[draw=none,fill=forestgreen4416044,fill opacity=0.7] (axis cs:95.2,0) rectangle (axis cs:96.8,0.0312500000000001);
\draw[draw=none,fill=forestgreen4416044,fill opacity=0.7] (axis cs:96.8,0) rectangle (axis cs:98.4,0.134999999999999);
\draw[draw=none,fill=forestgreen4416044,fill opacity=0.7] (axis cs:98.4,0) rectangle (axis cs:100,0.0937500000000003);
\draw[draw=none,fill=forestgreen4416044,fill opacity=0.7] (axis cs:100,0) rectangle (axis cs:101.6,0.161250000000001);
\draw[draw=none,fill=forestgreen4416044,fill opacity=0.7] (axis cs:101.6,0) rectangle (axis cs:103.2,0.108749999999999);
\draw[draw=none,fill=forestgreen4416044,fill opacity=0.7] (axis cs:103.2,0) rectangle (axis cs:104.8,0.0287500000000001);
\draw[draw=none,fill=forestgreen4416044,fill opacity=0.7] (axis cs:104.8,0) rectangle (axis cs:106.4,0.0249999999999999);
\draw[draw=none,fill=forestgreen4416044,fill opacity=0.7] (axis cs:106.4,0) rectangle (axis cs:108,0.00875000000000003);
\draw[draw=none,fill=orchid227119194,fill opacity=0.7] (axis cs:80,0) rectangle (axis cs:81.5,0.00533333333333333);
\addlegendimage{ybar,ybar legend,draw=none,fill=orchid227119194,fill opacity=0.7}
\addlegendentry{uniform+logical}

\draw[draw=none,fill=orchid227119194,fill opacity=0.7] (axis cs:81.5,0) rectangle (axis cs:83,0.00666666666666667);
\draw[draw=none,fill=orchid227119194,fill opacity=0.7] (axis cs:83,0) rectangle (axis cs:84.5,0.0773333333333333);
\draw[draw=none,fill=orchid227119194,fill opacity=0.7] (axis cs:84.5,0) rectangle (axis cs:86,0.0866666666666667);
\draw[draw=none,fill=orchid227119194,fill opacity=0.7] (axis cs:86,0) rectangle (axis cs:87.5,0.174666666666667);
\draw[draw=none,fill=orchid227119194,fill opacity=0.7] (axis cs:87.5,0) rectangle (axis cs:89,0.106666666666667);
\draw[draw=none,fill=orchid227119194,fill opacity=0.7] (axis cs:89,0) rectangle (axis cs:90.5,0.138666666666667);
\draw[draw=none,fill=orchid227119194,fill opacity=0.7] (axis cs:90.5,0) rectangle (axis cs:92,0.0306666666666667);
\draw[draw=none,fill=orchid227119194,fill opacity=0.7] (axis cs:92,0) rectangle (axis cs:93.5,0.0306666666666667);
\draw[draw=none,fill=orchid227119194,fill opacity=0.7] (axis cs:93.5,0) rectangle (axis cs:95,0.00933333333333333);
\end{axis}

\end{tikzpicture}
            \end{tabular}
            \caption{Distribution of the number of transactions audited ($\tau$) for the same experiments as~\Cref{fig:no-side-info-CS}, with $\varepsilon=0.05$. We omit the uniform~(without logical CS) CS histograms, as they are concentrated entirely at $N$}
            \label{fig:no-side-info-Hist}
    \end{figure}
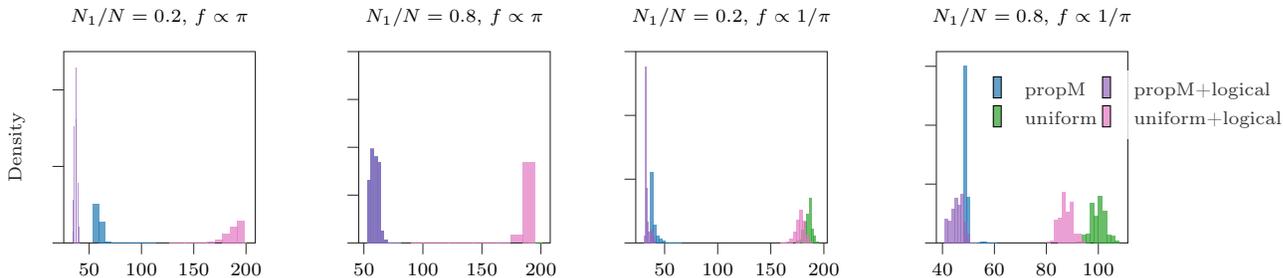

    \paragraph{No side information: uniform vs.\ \propM\ sampling.}
        In the first experiment, we compare the performance of the \propM strategy with the uniform baseline. In addition to this, we also illustrate the significance of logical CS~(introduced in \Cref{subsec:logical-CS}) especially in cases when there are a few large $\pi$ values.
        From the widths of the CSs plotted in \Cref{fig:no-side-info-CS}, we can see that \propM outperforms the uniform baseline in all four cases. The gap in performance increases when $\Nlarge$ is small since $\pi$ deviates more significantly from the uniform weighting: it consists of a few large weights with the rest close to $0$. On the other hand, when $\Nlarge$ is large, the weights  resemble the uniform distribution, leading to the competitive performance of the uniform baseline. The logical CSs are most useful in the case of small $\Nlarge$, especially with $f \propto \pi$. This is because for small $\Nlarge$, every query to an index with large $\pi$ value leads to a significant reduction in the uncertainty about $m^*$.
    
         Next, in~\Cref{fig:no-side-info-Hist}, we plot the distribution of the stopping time $\tau$ for an RLFA
 with $\varepsilon = \delta=0.05$,  over $500$ independent trials. The \propM strategy leads to a significant reduction in the sample size requirement to obtain an $\varepsilon$-accurate estimate of $m^*$ as compared to the uniform baseline, both with and without the logical CS. Furthermore, the distribution of $\tau$ with the \propM strategy often has  less variability than the uniform strategy. Hence, \propM has demonstrated itself empirically to be a better sampling strategy than simply sampling uniformly, as one would do when all the weights are equal.

    \paragraph{Using \propMS\ with accurate side information.} 
        In the second experiment, we study the benefit of incorporating accurate side information in the design of our CSs, by comparing the performance of \propMS strategy with that of the \propM strategy. We generate $S$ randomly while ensuring $S(i) / f(i) \in [1 \pm a]$ (from \eqref{eq:accurate-side-info-assumption}) for some $a \in (0,1)$. Thus smaller values of $a$ imply that the scores $S(i)$ are more accurate approximations of $f(i)$ for all $i \in [N]$.

\begin{figure}[t]
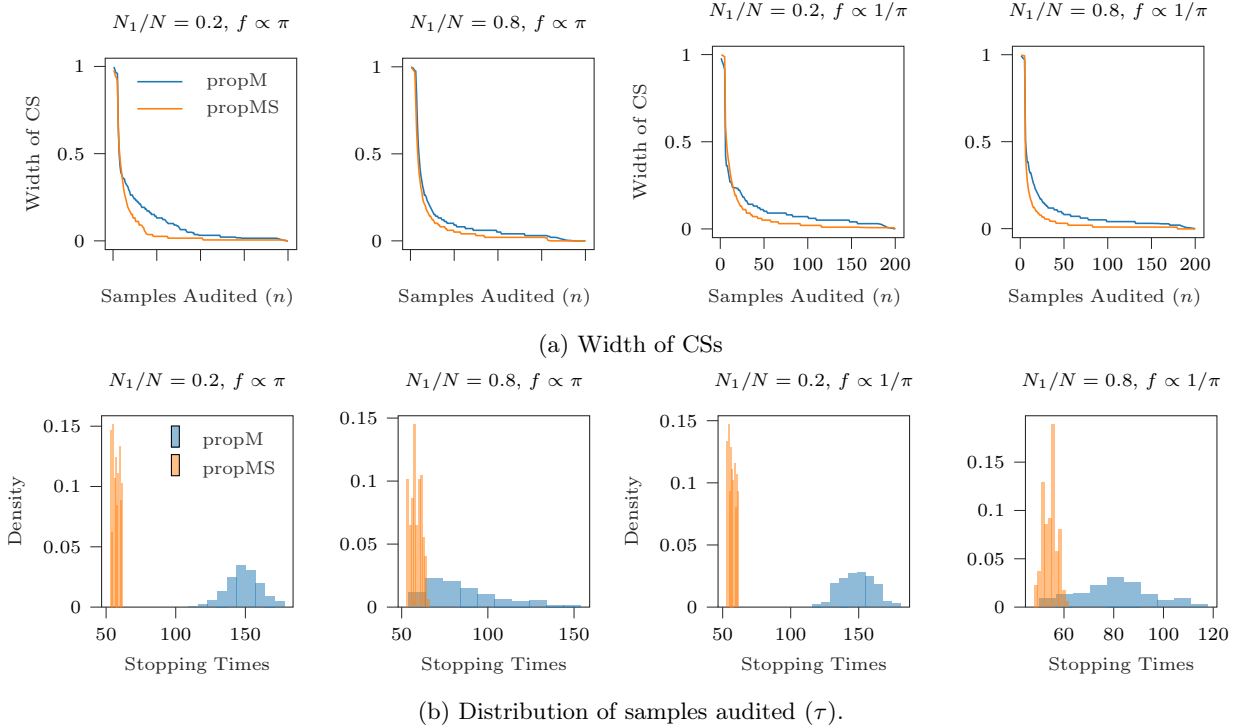

\begin{subfigure}{\columnwidth}
     \def\figwidth{0.25\columnwidth}
     \def\figheight{0.25\columnwidth} %
     \centering
     \hspace*{-1em}
     \input{Figures/AccurateSideInfoCSf_propto_M_large_40_A_0_1}
     \begin{tikzpicture}

\definecolor{darkorange25512714}{RGB}{255,127,14}
\definecolor{darkslategray38}{RGB}{38,38,38}
\definecolor{lightgray204}{RGB}{204,204,204}
\definecolor{steelblue31119180}{RGB}{31,119,180}

\tikzstyle{every node}=[font=\scriptsize]
\begin{axis}[
axis line style={darkslategray38},
height=\figheight,
legend cell align={left},
legend style={fill opacity=0.8, draw opacity=1, text opacity=1, draw=none},
tick align=outside,
tick pos=left,
title={$N_1/N=0.8$,    \(\displaystyle f \propto \pi\)},
width=\figwidth,
x grid style={lightgray204},
xlabel=\textcolor{darkslategray38}{Samples Audited (\(\displaystyle n\)) },
xmin=-8.95, xmax=209.95,
xtick style={color=darkslategray38},
xticklabels={,,},
y grid style={lightgray204},
ylabel=\textcolor{darkslategray38}{},
ymin=-0.0500106766164154, ymax=1.04939080940641,
ytick style={color=darkslategray38}
]
\addplot [semithick, steelblue31119180]
table {%
1 0.999418014587189
2 0.995346215047418
3 0.992641442399086
4 0.985880871871744
5 0.977225642726343
6 0.9749674005879
7 0.844463904591544
8 0.666666666666667
9 0.545454545454545
10 0.474747474747475
11 0.414141414141414
12 0.373737373737374
13 0.343434343434343
14 0.313131313131313
15 0.292929292929293
16 0.262626262626263
17 0.262626262626263
18 0.252525252525253
19 0.232323232323232
20 0.222222222222222
21 0.212121212121212
22 0.202020202020202
23 0.191919191919192
24 0.181818181818182
25 0.171717171717172
26 0.161616161616162
27 0.151515151515152
28 0.151515151515152
29 0.141414141414141
30 0.141414141414141
31 0.141414141414141
32 0.131313131313131
33 0.131313131313131
34 0.131313131313131
35 0.131313131313131
36 0.121212121212121
37 0.121212121212121
38 0.121212121212121
39 0.111111111111111
40 0.111111111111111
41 0.111111111111111
42 0.101010101010101
43 0.101010101010101
44 0.101010101010101
45 0.101010101010101
46 0.101010101010101
47 0.101010101010101
48 0.101010101010101
49 0.101010101010101
50 0.0909090909090909
51 0.0909090909090909
52 0.0909090909090909
53 0.0808080808080808
54 0.0808080808080808
55 0.0808080808080808
56 0.0808080808080808
57 0.0808080808080808
58 0.0808080808080808
59 0.0808080808080808
60 0.0808080808080808
61 0.0808080808080808
62 0.0808080808080808
63 0.0707070707070707
64 0.0707070707070707
65 0.0707070707070707
66 0.0707070707070707
67 0.0707070707070707
68 0.0707070707070707
69 0.0707070707070707
70 0.0707070707070707
71 0.0707070707070707
72 0.0606060606060606
73 0.0606060606060606
74 0.0606060606060606
75 0.0606060606060606
76 0.0606060606060606
77 0.0606060606060606
78 0.0606060606060606
79 0.0606060606060606
80 0.0606060606060606
81 0.0606060606060606
82 0.0606060606060606
83 0.0606060606060606
84 0.0606060606060606
85 0.0606060606060606
86 0.0606060606060606
87 0.0606060606060606
88 0.0606060606060606
89 0.0606060606060606
90 0.0606060606060606
91 0.0606060606060606
92 0.0606060606060606
93 0.0606060606060606
94 0.0606060606060606
95 0.0606060606060606
96 0.0606060606060606
97 0.0606060606060606
98 0.0606060606060606
99 0.0505050505050505
100 0.0505050505050505
101 0.0505050505050505
102 0.0505050505050505
103 0.0505050505050505
104 0.0404040404040404
105 0.0404040404040404
106 0.0404040404040404
107 0.0404040404040404
108 0.0404040404040404
109 0.0404040404040404
110 0.0404040404040404
111 0.0404040404040404
112 0.0404040404040404
113 0.0404040404040404
114 0.0404040404040404
115 0.0404040404040404
116 0.0404040404040404
117 0.0404040404040404
118 0.0404040404040404
119 0.0404040404040404
120 0.0404040404040404
121 0.0404040404040404
122 0.0404040404040404
123 0.0404040404040404
124 0.0404040404040404
125 0.0404040404040404
126 0.0404040404040404
127 0.0404040404040404
128 0.0404040404040404
129 0.0404040404040404
130 0.0404040404040404
131 0.0303030303030303
132 0.0303030303030303
133 0.0303030303030303
134 0.0303030303030303
135 0.0303030303030303
136 0.0303030303030303
137 0.0303030303030303
138 0.0303030303030303
139 0.0303030303030303
140 0.0303030303030303
141 0.0303030303030303
142 0.0303030303030303
143 0.0303030303030303
144 0.0303030303030303
145 0.0303030303030303
146 0.0303030303030303
147 0.0303030303030303
148 0.0303030303030303
149 0.0303030303030303
150 0.0303030303030303
151 0.0303030303030303
152 0.0303030303030303
153 0.0303030303030303
154 0.0303030303030303
155 0.0303030303030303
156 0.0295853212282953
157 0.0268079150135939
158 0.0256275150785858
159 0.025624293963876
160 0.0202020202020202
161 0.0202020202020202
162 0.0202020202020202
163 0.0202020202020202
164 0.0202020202020202
165 0.0202020202020202
166 0.0202020202020202
167 0.0185129441595042
168 0.0185100792913299
169 0.0165765531287208
170 0.0143348639614696
171 0.0131745693982788
172 0.0119148405111511
173 0.0102739604661783
174 0.00855878107667651
175 0.00855679148784072
176 0.00756072812537095
177 0.00499859261353419
178 0.00358157800336323
179 0.00357985877956529
180 0.00357738308042838
181 0.00222001206220451
182 0.00221784580606454
183 0.00221465752252176
184 0.00221377578222925
185 -1.70382500269683e-05
186 -1.89569366346731e-05
187 -2.00804030042745e-05
188 -2.15979870909488e-05
189 -2.34559016422109e-05
190 -2.67333602817499e-05
191 -2.83846937993992e-05
192 -2.90721392929294e-05
193 -3.09175424067498e-05
194 -3.16809069669799e-05
195 -3.32973368125655e-05
196 -3.43624968358602e-05
197 -3.4990433847415e-05
198 -3.63842174364026e-05
199 -3.70906123617343e-05
200 -3.78817971961598e-05
};
\addlegendentry{propM}
\addplot [semithick, darkorange25512714]
table {%
1 0.997036467100548
2 0.989438832262953
3 0.983065506335134
4 0.976266557769051
5 0.969505987241709
6 0.808080808080808
7 0.636363636363636
8 0.525252525252525
9 0.444444444444444
10 0.383838383838384
11 0.343434343434343
12 0.303030303030303
13 0.282828282828283
14 0.242424242424242
15 0.222222222222222
16 0.212121212121212
17 0.202020202020202
18 0.181818181818182
19 0.171717171717172
20 0.161616161616162
21 0.151515151515151
22 0.141414141414141
23 0.141414141414141
24 0.131313131313131
25 0.121212121212121
26 0.121212121212121
27 0.111111111111111
28 0.101010101010101
29 0.101010101010101
30 0.101010101010101
31 0.101010101010101
32 0.0909090909090909
33 0.0909090909090909
34 0.0808080808080808
35 0.0808080808080808
36 0.0808080808080808
37 0.0808080808080808
38 0.0808080808080808
39 0.0808080808080808
40 0.0707070707070707
41 0.0707070707070707
42 0.0606060606060606
43 0.0606060606060606
44 0.0606060606060606
45 0.0606060606060606
46 0.0606060606060606
47 0.0606060606060606
48 0.0606060606060606
49 0.0606060606060606
50 0.0505050505050505
51 0.0505050505050505
52 0.0505050505050505
53 0.0505050505050505
54 0.0505050505050505
55 0.0505050505050505
56 0.0505050505050505
57 0.0505050505050505
58 0.0404040404040404
59 0.0404040404040404
60 0.0404040404040404
61 0.0404040404040404
62 0.0404040404040404
63 0.0404040404040404
64 0.0404040404040404
65 0.0404040404040404
66 0.0404040404040404
67 0.0404040404040404
68 0.0404040404040404
69 0.0404040404040404
70 0.0404040404040404
71 0.0404040404040404
72 0.0404040404040404
73 0.0404040404040404
74 0.0404040404040404
75 0.0303030303030303
76 0.0303030303030303
77 0.0303030303030303
78 0.0303030303030303
79 0.0303030303030303
80 0.0303030303030303
81 0.0303030303030303
82 0.0303030303030303
83 0.0303030303030303
84 0.0303030303030303
85 0.0202020202020202
86 0.0202020202020202
87 0.0202020202020202
88 0.0202020202020202
89 0.0202020202020202
90 0.0202020202020202
91 0.0202020202020202
92 0.0202020202020202
93 0.0202020202020202
94 0.0202020202020202
95 0.0202020202020202
96 0.0202020202020202
97 0.0202020202020202
98 0.0202020202020202
99 0.0202020202020202
100 0.0202020202020202
101 0.0202020202020202
102 0.0202020202020202
103 0.0202020202020202
104 0.0202020202020202
105 0.0202020202020202
106 0.0202020202020202
107 0.0202020202020202
108 0.0202020202020202
109 0.0202020202020202
110 0.0202020202020202
111 0.0202020202020202
112 0.0202020202020202
113 0.0202020202020202
114 0.0202020202020202
115 0.0202020202020202
116 0.0202020202020202
117 0.0202020202020202
118 0.0202020202020202
119 0.0202020202020202
120 0.0202020202020202
121 0.0202020202020202
122 0.0202020202020202
123 0.0202020202020202
124 0.0202020202020202
125 0.0202020202020202
126 0.0202020202020202
127 0.0202020202020202
128 0.0202020202020202
129 0.0202020202020202
130 0.0202020202020202
131 0.0202020202020202
132 0.0202020202020202
133 0.0202020202020202
134 0.0202020202020202
135 0.0202020202020202
136 0.0202020202020202
137 0.0202020202020202
138 0.0202020202020202
139 0.0202020202020202
140 0.0202020202020202
141 0.0202020202020202
142 0.0202020202020202
143 0.0202020202020202
144 0.0202020202020202
145 0.0202020202020202
146 0.0202020202020202
147 0.0202020202020202
148 0.0202020202020202
149 0.0202020202020202
150 0.0202020202020202
151 0.0202020202020202
152 0.0202020202020202
153 0.0202020202020202
154 0.0183557767551641
155 0.0173243228189133
156 0.0162125773394556
157 0.00452111597035237
158 0.00204857319769031
159 0.000966389872843088
160 3.78989050804357e-05
161 3.61429495593368e-05
162 3.29218348494864e-05
163 0
164 0
165 0
166 0
167 0
168 0
169 0
170 0
171 0
172 0
173 0
174 0
175 0
176 -4.66382187258585e-07
177 -3.74384082679757e-06
178 -6.03868279802633e-06
179 -6.74507772335797e-06
180 -7.5362625577835e-06
181 -8.60142258107821e-06
182 -9.48606621309356e-06
183 -1.19755266317356e-05
184 -1.33693102207233e-05
185 -1.51997732311115e-05
186 -1.71184598388163e-05
187 -1.81508316655643e-05
188 -2.01404205013578e-05
189 -2.19983350526198e-05
190 -2.35159191392942e-05
191 -2.57949377974187e-05
192 -2.74113676430043e-05
193 -2.91305914409401e-05
194 -3.00123317334422e-05
195 -3.06997772269724e-05
196 -3.24925352749306e-05
197 -3.4337938388751e-05
198 -3.51013029489811e-05
199 -3.7253860184383e-05
200 -3.78817971959378e-05
};
\addlegendentry{propMS};
\legend{};
\end{axis}

\end{tikzpicture}
     \begin{tikzpicture}

\definecolor{darkorange25512714}{RGB}{255,127,14}
\definecolor{darkslategray38}{RGB}{38,38,38}
\definecolor{lightgray204}{RGB}{204,204,204}
\definecolor{steelblue31119180}{RGB}{31,119,180}

\tikzstyle{every node}=[font=\scriptsize]
\begin{axis}[
axis line style={darkslategray38},
height=\figheight,
legend cell align={left},
legend style={fill opacity=0.8, draw opacity=1, text opacity=1, draw=none},
tick align=outside,
tick pos=left,
title={$N_1/N=0.2$,    \(\displaystyle f \propto 1/\pi\)},
width=\figwidth,
x grid style={lightgray204},
xlabel=\textcolor{darkslategray38}{Samples Audited (\(\displaystyle n\)) },
xmin=-8.95, xmax=209.95,
xtick style={color=darkslategray38},
y grid style={lightgray204},
ylabel=\textcolor{darkslategray38}{Width of CS},
ymin=-0.0498894058902771, ymax=1.04767752369582,
ytick style={color=darkslategray38}
]
\addplot [semithick, steelblue31119180]
table {%
1 0.979064300831915
2 0.961444550441976
3 0.950166366899512
4 0.935000846568604
5 0.915276734570577
6 0.424107027158049
7 0.363475337560522
8 0.361807979086153
9 0.331489721759017
10 0.301157846458991
11 0.270822488565697
12 0.269490494762087
13 0.259362721076908
14 0.239124170629484
15 0.237702341952787
16 0.236071925930226
17 0.236057084945166
18 0.235125881112661
19 0.232323232323232
20 0.232151440587066
21 0.222222222222222
22 0.222222222222222
23 0.212121212121212
24 0.202020202020202
25 0.191919191919192
26 0.181818181818182
27 0.181818181818182
28 0.171717171717172
29 0.161616161616162
30 0.161616161616162
31 0.151515151515152
32 0.141414141414141
33 0.141414141414141
34 0.141414141414141
35 0.141414141414141
36 0.141414141414141
37 0.141414141414141
38 0.131313131313131
39 0.131313131313131
40 0.131313131313131
41 0.131313131313131
42 0.131313131313131
43 0.121212121212121
44 0.121212121212121
45 0.121212121212121
46 0.111111111111111
47 0.111111111111111
48 0.111111111111111
49 0.111111111111111
50 0.101010101010101
51 0.101010101010101
52 0.101010101010101
53 0.101010101010101
54 0.0909090909090909
55 0.0909090909090909
56 0.0909090909090909
57 0.0909090909090909
58 0.0909090909090909
59 0.0909090909090909
60 0.0909090909090909
61 0.0909090909090909
62 0.0909090909090909
63 0.0909090909090909
64 0.0909090909090909
65 0.0909090909090909
66 0.0909090909090909
67 0.0909090909090909
68 0.0909090909090909
69 0.0909090909090909
70 0.0909090909090909
71 0.0909090909090909
72 0.0909090909090909
73 0.0909090909090909
74 0.0909090909090909
75 0.0909090909090909
76 0.0808080808080808
77 0.0808080808080808
78 0.0808080808080808
79 0.0808080808080808
80 0.0808080808080808
81 0.0808080808080808
82 0.0808080808080808
83 0.0707070707070707
84 0.0707070707070707
85 0.0707070707070707
86 0.0707070707070707
87 0.0707070707070707
88 0.0707070707070707
89 0.0707070707070707
90 0.0707070707070707
91 0.0707070707070707
92 0.0707070707070707
93 0.0707070707070707
94 0.0707070707070707
95 0.0707070707070707
96 0.0707070707070707
97 0.0707070707070707
98 0.0707070707070707
99 0.0707070707070707
100 0.0707070707070707
101 0.0606060606060606
102 0.0606060606060606
103 0.0606060606060606
104 0.0606060606060606
105 0.0606060606060606
106 0.0606060606060606
107 0.0606060606060606
108 0.0606060606060606
109 0.0606060606060606
110 0.0505050505050505
111 0.0505050505050505
112 0.0505050505050505
113 0.0505050505050505
114 0.0505050505050505
115 0.0505050505050505
116 0.0505050505050505
117 0.0505050505050505
118 0.0505050505050505
119 0.0505050505050505
120 0.0505050505050505
121 0.0505050505050505
122 0.0505050505050505
123 0.0505050505050505
124 0.0505050505050505
125 0.0505050505050505
126 0.0505050505050505
127 0.0505050505050505
128 0.0505050505050505
129 0.0505050505050505
130 0.0505050505050505
131 0.0505050505050505
132 0.0505050505050505
133 0.0505050505050505
134 0.0505050505050505
135 0.0505050505050505
136 0.0505050505050505
137 0.0505050505050505
138 0.0505050505050505
139 0.0505050505050505
140 0.0505050505050505
141 0.0505050505050505
142 0.0505050505050505
143 0.0505050505050505
144 0.0505050505050505
145 0.0505050505050505
146 0.0505050505050505
147 0.0505050505050505
148 0.0404040404040404
149 0.0404040404040404
150 0.0404040404040404
151 0.0404040404040404
152 0.0404040404040404
153 0.0404040404040404
154 0.0404040404040404
155 0.0404040404040404
156 0.0404040404040404
157 0.0303030303030303
158 0.0303030303030303
159 0.0303030303030303
160 0.0303030303030303
161 0.0303030303030303
162 0.0303030303030303
163 0.0303030303030303
164 0.0303030303030303
165 0.0303030303030303
166 0.0303030303030303
167 0.0303030303030303
168 0.0303030303030303
169 0.0303030303030303
170 0.0303030303030303
171 0.0303030303030303
172 0.0303030303030303
173 0.0303030303030303
174 0.0303030303030303
175 0.0303030303030303
176 0.0303030303030303
177 0.0303030303030303
178 0.0303030303030303
179 0.0303030303030303
180 0.0302814701924397
181 0.0293142391740967
182 0.0283766703815596
183 0.0273362701674035
184 0.0260660870438815
185 0.0242596405145423
186 0.0205428264225801
187 0.0184140398002789
188 0.0169323778989864
189 0.0132628249162917
190 0.0110509427218345
191 0.00817961385073762
192 0.00743546453141408
193 0.00654975753853382
194 0.00556802671349443
195 0.00469491492571039
196 0.00397336388965375
197 0.00340350710638704
198 0.00265617669230145
199 0.000771300689186516
200 0
};
\addlegendentry{propM}
\addplot [semithick, darkorange25512714]
table {%
1 0.997788117805543
2 0.997016817116356
3 0.994626248584766
4 0.990330863099223
5 0.987982645756018
6 0.589223900415084
7 0.508067514915229
8 0.445389440220627
9 0.40392966918987
10 0.372670903100543
11 0.333333333333333
12 0.303030303030303
13 0.272727272727273
14 0.252525252525253
15 0.232323232323232
16 0.212121212121212
17 0.191919191919192
18 0.191919191919192
19 0.171717171717172
20 0.161616161616162
21 0.151515151515152
22 0.141414141414141
23 0.131313131313131
24 0.131313131313131
25 0.121212121212121
26 0.111111111111111
27 0.111111111111111
28 0.111111111111111
29 0.111111111111111
30 0.0909090909090909
31 0.0909090909090909
32 0.0909090909090909
33 0.0909090909090909
34 0.0909090909090909
35 0.0808080808080808
36 0.0808080808080808
37 0.0707070707070707
38 0.0707070707070707
39 0.0707070707070707
40 0.0707070707070707
41 0.0707070707070707
42 0.0707070707070707
43 0.0606060606060606
44 0.0606060606060606
45 0.0606060606060606
46 0.0606060606060606
47 0.0606060606060606
48 0.0505050505050505
49 0.0505050505050505
50 0.0505050505050505
51 0.0505050505050505
52 0.0505050505050505
53 0.0505050505050505
54 0.0505050505050505
55 0.0505050505050505
56 0.0505050505050505
57 0.0505050505050505
58 0.0404040404040404
59 0.0404040404040404
60 0.0404040404040404
61 0.0404040404040404
62 0.0404040404040404
63 0.0404040404040404
64 0.0404040404040404
65 0.0404040404040404
66 0.0404040404040404
67 0.0303030303030303
68 0.0303030303030303
69 0.0303030303030303
70 0.0303030303030303
71 0.0303030303030303
72 0.0303030303030303
73 0.0303030303030303
74 0.0303030303030303
75 0.0303030303030303
76 0.0303030303030303
77 0.0303030303030303
78 0.0303030303030303
79 0.0303030303030303
80 0.0303030303030303
81 0.0303030303030303
82 0.0303030303030303
83 0.0303030303030303
84 0.0303030303030303
85 0.0303030303030303
86 0.0303030303030303
87 0.0303030303030303
88 0.0303030303030303
89 0.0303030303030303
90 0.0303030303030303
91 0.0303030303030303
92 0.0202020202020202
93 0.0202020202020202
94 0.0202020202020202
95 0.0202020202020202
96 0.0202020202020202
97 0.0202020202020202
98 0.0202020202020202
99 0.0202020202020202
100 0.0202020202020202
101 0.0202020202020202
102 0.0202020202020202
103 0.0202020202020202
104 0.0202020202020202
105 0.0202020202020202
106 0.0202020202020202
107 0.0202020202020202
108 0.0202020202020202
109 0.0202020202020202
110 0.0202020202020202
111 0.0202020202020202
112 0.0202020202020202
113 0.0202020202020202
114 0.0202020202020202
115 0.0202020202020202
116 0.0101010101010101
117 0.0101010101010101
118 0.0101010101010101
119 0.0101010101010101
120 0.0101010101010101
121 0.0101010101010101
122 0.0101010101010101
123 0.0101010101010101
124 0.0101010101010101
125 0.0101010101010101
126 0.0101010101010101
127 0.0101010101010101
128 0.0101010101010101
129 0.0101010101010101
130 0.0101010101010101
131 0.0101010101010101
132 0.0101010101010101
133 0.0101010101010101
134 0.0101010101010101
135 0.0101010101010101
136 0.0101010101010101
137 0.0101010101010101
138 0.0101010101010101
139 0.0101010101010101
140 0.0101010101010101
141 0.0101010101010101
142 0.0101010101010101
143 0.0101010101010101
144 0.0101010101010101
145 0.0101010101010101
146 0.0101010101010101
147 0.0101010101010101
148 0.0101010101010101
149 0.0101010101010101
150 0.0101010101010101
151 0.0101010101010101
152 0.0101010101010101
153 0.0101010101010101
154 0.0101010101010101
155 0.0101010101010101
156 0.0101010101010101
157 0.0101010101010101
158 0.0101010101010101
159 0.00957905288116373
160 0.00911809160469765
161 0.0084282607216008
162 0.00839132680385035
163 0.00836479005100949
164 0.00835216844262721
165 0.00803909748315082
166 0.00801786065920435
167 0.00800158753854585
168 0.00797183667674972
169 0.00757754178531511
170 0.00754684662688057
171 0.00751800162988542
172 0.00748921322967949
173 0.00745268298427562
174 0.00742128589470137
175 0.00740280225424267
176 0.00738865892751636
177 0.00737805031540828
178 0.00736822657938865
179 0.00735338559432838
180 0.00731942000272437
181 0.00729379101125738
182 0.0072870999934472
183 0.00727357886768601
184 0.00725557507678931
185 0.00722324748652597
186 0.00719648390235753
187 0.00716501044162599
188 0.00714511221907491
189 0.00712482054257091
190 0.00709195918347136
191 0.00707673215936525
192 0.0070670428921937
193 0.00704150457532082
194 0.00702536916714444
195 0.00700495215917096
196 0.00699539465238377
197 0.00697644024739896
198 0.00696630784931232
199 0.00523952128003868
200 0
};
\addlegendentry{propMS};
\legend{};
\end{axis}

\end{tikzpicture}
     \input{Figures/AccurateSideInfoCSf_inv_propto_M_large_160_A_0_1}
     \caption{Width of CSs}
     \label{fig:accurate-side-info-CS}
\end{subfigure}
\begin{subfigure}{\columnwidth}
         \def\figwidth{0.25\columnwidth}
         \def\figheight{0.25\columnwidth} %
         \centering
         \hspace*{-1em}
         \begin{tikzpicture}

\definecolor{darkorange25512714}{RGB}{255,127,14}
\definecolor{darkslategray38}{RGB}{38,38,38}
\definecolor{lightgray204}{RGB}{204,204,204}
\definecolor{steelblue31119180}{RGB}{31,119,180}

\tikzstyle{every node}=[font=\scriptsize]
\begin{axis}[
axis line style={darkslategray38},
height=\figheight,
legend cell align={left},
legend style={fill opacity=0.8, draw opacity=1, text opacity=1, draw=none},
tick align=outside,
tick pos=left,
title={
$N_1/N=0.2$,  \(\displaystyle f \propto \pi\)},
width=\figwidth,
x grid style={lightgray204},
xlabel=\textcolor{darkslategray38}{Stopping Times},
xmin=46.75, xmax=184.25,
xtick style={color=darkslategray38},
y grid style={lightgray204},
ylabel=\textcolor{darkslategray38}{Density},
ymin=0, ymax=0.158666666666666,
ytick style={color=darkslategray38}
]
\draw[draw=none,fill=steelblue31119180,fill opacity=0.5] (axis cs:109,0) rectangle (axis cs:115.9,0.000579710144927536);
\addlegendimage{ybar,ybar legend,draw=none,fill=steelblue31119180,fill opacity=0.5}
\addlegendentry{propM}

\draw[draw=none,fill=steelblue31119180,fill opacity=0.5] (axis cs:115.9,0) rectangle (axis cs:122.8,0.00231884057971015);
\draw[draw=none,fill=steelblue31119180,fill opacity=0.5] (axis cs:122.8,0) rectangle (axis cs:129.7,0.00579710144927537);
\draw[draw=none,fill=steelblue31119180,fill opacity=0.5] (axis cs:129.7,0) rectangle (axis cs:136.6,0.0133333333333333);
\draw[draw=none,fill=steelblue31119180,fill opacity=0.5] (axis cs:136.6,0) rectangle (axis cs:143.5,0.0243478260869565);
\draw[draw=none,fill=steelblue31119180,fill opacity=0.5] (axis cs:143.5,0) rectangle (axis cs:150.4,0.0347826086956521);
\draw[draw=none,fill=steelblue31119180,fill opacity=0.5] (axis cs:150.4,0) rectangle (axis cs:157.3,0.0301449275362319);
\draw[draw=none,fill=steelblue31119180,fill opacity=0.5] (axis cs:157.3,0) rectangle (axis cs:164.2,0.0197101449275363);
\draw[draw=none,fill=steelblue31119180,fill opacity=0.5] (axis cs:164.2,0) rectangle (axis cs:171.1,0.00927536231884057);
\draw[draw=none,fill=steelblue31119180,fill opacity=0.5] (axis cs:171.1,0) rectangle (axis cs:178,0.00463768115942029);
\draw[draw=none,fill=darkorange25512714,fill opacity=0.5] (axis cs:53,0) rectangle (axis cs:53.9,0.146666666666667);
\addlegendimage{ybar,ybar legend,draw=none,fill=darkorange25512714,fill opacity=0.5}
\addlegendentry{propMS}

\draw[draw=none,fill=darkorange25512714,fill opacity=0.5] (axis cs:53.9,0) rectangle (axis cs:54.8,0.0622222222222223);
\draw[draw=none,fill=darkorange25512714,fill opacity=0.5] (axis cs:54.8,0) rectangle (axis cs:55.7,0.15111111111111);
\draw[draw=none,fill=darkorange25512714,fill opacity=0.5] (axis cs:55.7,0) rectangle (axis cs:56.6,0.106666666666667);
\draw[draw=none,fill=darkorange25512714,fill opacity=0.5] (axis cs:56.6,0) rectangle (axis cs:57.5,0.124444444444445);
\draw[draw=none,fill=darkorange25512714,fill opacity=0.5] (axis cs:57.5,0) rectangle (axis cs:58.4,0.0844444444444446);
\draw[draw=none,fill=darkorange25512714,fill opacity=0.5] (axis cs:58.4,0) rectangle (axis cs:59.3,0.111111111111111);
\draw[draw=none,fill=darkorange25512714,fill opacity=0.5] (axis cs:59.3,0) rectangle (axis cs:60.2,0.133333333333332);
\draw[draw=none,fill=darkorange25512714,fill opacity=0.5] (axis cs:60.2,0) rectangle (axis cs:61.1,0.088888888888889);
\draw[draw=none,fill=darkorange25512714,fill opacity=0.5] (axis cs:61.1,0) rectangle (axis cs:62,0.102222222222222);
\end{axis}

\end{tikzpicture}
         \begin{tikzpicture}

\definecolor{darkorange25512714}{RGB}{255,127,14}
\definecolor{darkslategray38}{RGB}{38,38,38}
\definecolor{lightgray204}{RGB}{204,204,204}
\definecolor{steelblue31119180}{RGB}{31,119,180}

\tikzstyle{every node}=[font=\scriptsize]
\begin{axis}[
axis line style={darkslategray38},
height=\figheight,
legend cell align={left},
legend style={fill opacity=0.8, draw opacity=1, text opacity=1, draw=none},
tick align=outside,
tick pos=left,
title={
$N_1/N=0.8$,    \(\displaystyle f \propto \pi\)},
width=\figwidth,
x grid style={lightgray204},
xlabel=\textcolor{darkslategray38}{Stopping Times},
xmin=47.95, xmax=159.05,
xtick style={color=darkslategray38},
y grid style={lightgray204},
ylabel=\textcolor{darkslategray38}{},
ymin=0, ymax=0.151846153846153,
ytick style={color=darkslategray38}
]
\draw[draw=none,fill=steelblue31119180,fill opacity=0.5] (axis cs:54,0) rectangle (axis cs:64,0.012);
\addlegendimage{ybar,ybar legend,draw=none,fill=steelblue31119180,fill opacity=0.5}
\addlegendentry{propM}

\draw[draw=none,fill=steelblue31119180,fill opacity=0.5] (axis cs:64,0) rectangle (axis cs:74,0.0228);
\draw[draw=none,fill=steelblue31119180,fill opacity=0.5] (axis cs:74,0) rectangle (axis cs:84,0.0204);
\draw[draw=none,fill=steelblue31119180,fill opacity=0.5] (axis cs:84,0) rectangle (axis cs:94,0.0152);
\draw[draw=none,fill=steelblue31119180,fill opacity=0.5] (axis cs:94,0) rectangle (axis cs:104,0.01);
\draw[draw=none,fill=steelblue31119180,fill opacity=0.5] (axis cs:104,0) rectangle (axis cs:114,0.0064);
\draw[draw=none,fill=steelblue31119180,fill opacity=0.5] (axis cs:114,0) rectangle (axis cs:124,0.0048);
\draw[draw=none,fill=steelblue31119180,fill opacity=0.5] (axis cs:124,0) rectangle (axis cs:134,0.0052);
\draw[draw=none,fill=steelblue31119180,fill opacity=0.5] (axis cs:134,0) rectangle (axis cs:144,0.0016);
\draw[draw=none,fill=steelblue31119180,fill opacity=0.5] (axis cs:144,0) rectangle (axis cs:154,0.0016);
\draw[draw=none,fill=darkorange25512714,fill opacity=0.5] (axis cs:53,0) rectangle (axis cs:54.3,0.101538461538462);
\addlegendimage{ybar,ybar legend,draw=none,fill=darkorange25512714,fill opacity=0.5}
\addlegendentry{propMS}

\draw[draw=none,fill=darkorange25512714,fill opacity=0.5] (axis cs:54.3,0) rectangle (axis cs:55.6,0.0646153846153844);
\draw[draw=none,fill=darkorange25512714,fill opacity=0.5] (axis cs:55.6,0) rectangle (axis cs:56.9,0.0861538461538463);
\draw[draw=none,fill=darkorange25512714,fill opacity=0.5] (axis cs:56.9,0) rectangle (axis cs:58.2,0.144615384615384);
\draw[draw=none,fill=darkorange25512714,fill opacity=0.5] (axis cs:58.2,0) rectangle (axis cs:59.5,0.0646153846153848);
\draw[draw=none,fill=darkorange25512714,fill opacity=0.5] (axis cs:59.5,0) rectangle (axis cs:60.8,0.101538461538462);
\draw[draw=none,fill=darkorange25512714,fill opacity=0.5] (axis cs:60.8,0) rectangle (axis cs:62.1,0.104615384615384);
\draw[draw=none,fill=darkorange25512714,fill opacity=0.5] (axis cs:62.1,0) rectangle (axis cs:63.4,0.0553846153846155);
\draw[draw=none,fill=darkorange25512714,fill opacity=0.5] (axis cs:63.4,0) rectangle (axis cs:64.7,0.0399999999999999);
\draw[draw=none,fill=darkorange25512714,fill opacity=0.5] (axis cs:64.7,0) rectangle (axis cs:66,0.00615384615384617);
\legend{};
\end{axis}

\end{tikzpicture}
         \begin{tikzpicture}

\definecolor{darkorange25512714}{RGB}{255,127,14}
\definecolor{darkslategray38}{RGB}{38,38,38}
\definecolor{lightgray204}{RGB}{204,204,204}
\definecolor{steelblue31119180}{RGB}{31,119,180}

\tikzstyle{every node}=[font=\scriptsize]
\begin{axis}[
axis line style={darkslategray38},
height=\figheight,
legend cell align={left},
legend style={fill opacity=0.8, draw opacity=1, text opacity=1, draw=none},
tick align=outside,
tick pos=left,
title={
$N_1/N=0.2$,    \(\displaystyle f \propto 1/\pi\)},
width=\figwidth,
x grid style={lightgray204},
xlabel=\textcolor{darkslategray38}{Stopping Times},
xmin=46.6, xmax=187.4,
xtick style={color=darkslategray38},
y grid style={lightgray204},
ylabel=\textcolor{darkslategray38}{Density},
ymin=0, ymax=0.154,
ytick style={color=darkslategray38}
]
\draw[draw=none,fill=steelblue31119180,fill opacity=0.5] (axis cs:116,0) rectangle (axis cs:122.5,0.00246153846153846);
\addlegendimage{ybar,ybar legend,draw=none,fill=steelblue31119180,fill opacity=0.5}
\addlegendentry{propM}

\draw[draw=none,fill=steelblue31119180,fill opacity=0.5] (axis cs:122.5,0) rectangle (axis cs:129,0.00369230769230769);
\draw[draw=none,fill=steelblue31119180,fill opacity=0.5] (axis cs:129,0) rectangle (axis cs:135.5,0.0190769230769231);
\draw[draw=none,fill=steelblue31119180,fill opacity=0.5] (axis cs:135.5,0) rectangle (axis cs:142,0.0233846153846154);
\draw[draw=none,fill=steelblue31119180,fill opacity=0.5] (axis cs:142,0) rectangle (axis cs:148.5,0.0270769230769231);
\draw[draw=none,fill=steelblue31119180,fill opacity=0.5] (axis cs:148.5,0) rectangle (axis cs:155,0.0276923076923077);
\draw[draw=none,fill=steelblue31119180,fill opacity=0.5] (axis cs:155,0) rectangle (axis cs:161.5,0.0246153846153846);
\draw[draw=none,fill=steelblue31119180,fill opacity=0.5] (axis cs:161.5,0) rectangle (axis cs:168,0.0184615384615385);
\draw[draw=none,fill=steelblue31119180,fill opacity=0.5] (axis cs:168,0) rectangle (axis cs:174.5,0.00430769230769231);
\draw[draw=none,fill=steelblue31119180,fill opacity=0.5] (axis cs:174.5,0) rectangle (axis cs:181,0.00307692307692308);
\draw[draw=none,fill=darkorange25512714,fill opacity=0.5] (axis cs:53,0) rectangle (axis cs:53.9,0.133333333333334);
\addlegendimage{ybar,ybar legend,draw=none,fill=darkorange25512714,fill opacity=0.5}
\addlegendentry{propMS}

\draw[draw=none,fill=darkorange25512714,fill opacity=0.5] (axis cs:53.9,0) rectangle (axis cs:54.8,0.146666666666667);
\draw[draw=none,fill=darkorange25512714,fill opacity=0.5] (axis cs:54.8,0) rectangle (axis cs:55.7,0.0933333333333327);
\draw[draw=none,fill=darkorange25512714,fill opacity=0.5] (axis cs:55.7,0) rectangle (axis cs:56.6,0.128888888888889);
\draw[draw=none,fill=darkorange25512714,fill opacity=0.5] (axis cs:56.6,0) rectangle (axis cs:57.5,0.111111111111111);
\draw[draw=none,fill=darkorange25512714,fill opacity=0.5] (axis cs:57.5,0) rectangle (axis cs:58.4,0.102222222222222);
\draw[draw=none,fill=darkorange25512714,fill opacity=0.5] (axis cs:58.4,0) rectangle (axis cs:59.3,0.115555555555556);
\draw[draw=none,fill=darkorange25512714,fill opacity=0.5] (axis cs:59.3,0) rectangle (axis cs:60.2,0.0799999999999995);
\draw[draw=none,fill=darkorange25512714,fill opacity=0.5] (axis cs:60.2,0) rectangle (axis cs:61.1,0.106666666666667);
\draw[draw=none,fill=darkorange25512714,fill opacity=0.5] (axis cs:61.1,0) rectangle (axis cs:62,0.0933333333333335);
\legend{};
\end{axis}

\end{tikzpicture}
         \begin{tikzpicture}

\definecolor{darkorange25512714}{RGB}{255,127,14}
\definecolor{darkslategray38}{RGB}{38,38,38}
\definecolor{lightgray204}{RGB}{204,204,204}
\definecolor{steelblue31119180}{RGB}{31,119,180}

\tikzstyle{every node}=[font=\scriptsize]
\begin{axis}[
axis line style={darkslategray38},
height=\figheight,
legend cell align={left},
legend style={fill opacity=0.8, draw opacity=1, text opacity=1, draw=none},
tick align=outside,
tick pos=left,
title={$N_1/N=0.8$,    \(\displaystyle f \propto 1/\pi\)},
width=\figwidth,
x grid style={lightgray204},
xlabel=\textcolor{darkslategray38}{Stopping Times},
xmin=44.5, xmax=121.5,
xtick style={color=darkslategray38},
y grid style={lightgray204},
ylabel=\textcolor{darkslategray38}{},
ymin=0, ymax=0.198,
ytick style={color=darkslategray38}
]
\draw[draw=none,fill=steelblue31119180,fill opacity=0.5] (axis cs:50,0) rectangle (axis cs:56.8,0.00882352941176471);
\addlegendimage{ybar,ybar legend,draw=none,fill=steelblue31119180,fill opacity=0.5}
\addlegendentry{propM}

\draw[draw=none,fill=steelblue31119180,fill opacity=0.5] (axis cs:56.8,0) rectangle (axis cs:63.6,0.0135294117647059);
\draw[draw=none,fill=steelblue31119180,fill opacity=0.5] (axis cs:63.6,0) rectangle (axis cs:70.4,0.0141176470588235);
\draw[draw=none,fill=steelblue31119180,fill opacity=0.5] (axis cs:70.4,0) rectangle (axis cs:77.2,0.0223529411764706);
\draw[draw=none,fill=steelblue31119180,fill opacity=0.5] (axis cs:77.2,0) rectangle (axis cs:84,0.0311764705882353);
\draw[draw=none,fill=steelblue31119180,fill opacity=0.5] (axis cs:84,0) rectangle (axis cs:90.8,0.0258823529411765);
\draw[draw=none,fill=steelblue31119180,fill opacity=0.5] (axis cs:90.8,0) rectangle (axis cs:97.6,0.0129411764705882);
\draw[draw=none,fill=steelblue31119180,fill opacity=0.5] (axis cs:97.6,0) rectangle (axis cs:104.4,0.00705882352941175);
\draw[draw=none,fill=steelblue31119180,fill opacity=0.5] (axis cs:104.4,0) rectangle (axis cs:111.2,0.00882352941176473);
\draw[draw=none,fill=steelblue31119180,fill opacity=0.5] (axis cs:111.2,0) rectangle (axis cs:118,0.00235294117647058);
\draw[draw=none,fill=darkorange25512714,fill opacity=0.5] (axis cs:48,0) rectangle (axis cs:49.4,0.0228571428571429);
\addlegendimage{ybar,ybar legend,draw=none,fill=darkorange25512714,fill opacity=0.5}
\addlegendentry{propMS}

\draw[draw=none,fill=darkorange25512714,fill opacity=0.5] (axis cs:49.4,0) rectangle (axis cs:50.8,0.0371428571428572);
\draw[draw=none,fill=darkorange25512714,fill opacity=0.5] (axis cs:50.8,0) rectangle (axis cs:52.2,0.128571428571428);
\draw[draw=none,fill=darkorange25512714,fill opacity=0.5] (axis cs:52.2,0) rectangle (axis cs:53.6,0.0857142857142858);
\draw[draw=none,fill=darkorange25512714,fill opacity=0.5] (axis cs:53.6,0) rectangle (axis cs:55,0.0914285714285715);
\draw[draw=none,fill=darkorange25512714,fill opacity=0.5] (axis cs:55,0) rectangle (axis cs:56.4,0.188571428571429);
\draw[draw=none,fill=darkorange25512714,fill opacity=0.5] (axis cs:56.4,0) rectangle (axis cs:57.8,0.0571428571428572);
\draw[draw=none,fill=darkorange25512714,fill opacity=0.5] (axis cs:57.8,0) rectangle (axis cs:59.2,0.0799999999999997);
\draw[draw=none,fill=darkorange25512714,fill opacity=0.5] (axis cs:59.2,0) rectangle (axis cs:60.6,0.0171428571428572);
\draw[draw=none,fill=darkorange25512714,fill opacity=0.5] (axis cs:60.6,0) rectangle (axis cs:62,0.00571428571428572);
\legend{}; 
\end{axis}

\end{tikzpicture}
         \caption{Distribution of samples audited ($\tau$).}
     \label{fig:accurate-side-info-Hist}
    \end{subfigure}
    \caption{Comparison of the \propMS vs. the \propM strategy with accurate side information $(S(i))$, i.e., $S(i) / f(i) \in [0.9, 1.1]$ where $\varepsilon= \delta = 0.05$. We see that \propMS outperforms \propM in both CS width and sample efficiency.}
\end{figure}
        In~\Cref{fig:accurate-side-info-CS}, we can see that the \propMS strategy with accurate side information dominates the \propM strategy. This is further reflected in the distribution of $\tau$ for an RLFA where $\varepsilon=\delta=0.05$ in~\Cref{fig:accurate-side-info-Hist}.
        Hence, in situations where we are confident in the accuracy of our side information, we should incorporate it directly into our sampling strategy to reduce the width of the CS.

    \paragraph{Control variates from possibly inaccurate side information.}
        Finally, we consider the case in which we do not have prior information about the accuracy of the side information. Thus, using the \propMS strategy in this scenario directly can lead to very conservative CSs~(this is because in the absence of tight guarantees on the range of the $S/f$ ratio, we will have to use the worst case range). Instead, we compare the performance of the \propM strategy, with and without using control variates described in~\Cref{sec:side information}.
        In this case, we set $S(i) =  c\times f(i) +  (1-c)\times R_i$ for $c \in (0,1)$, where $(R_i)_{i \in [N]}$  are \iid random variables distributed uniformly over $[0,1]$. The parameter $c$ controls the level of correlation between $f$ and $S$ values, with small $c$ values indicating low correlation.

        We generate the data with  $\Nlarge =  40$ and $N=200$. In~\Cref{fig:CV-CS}, we compare the CSs and the distribution of $\tau$ for an RLFA~(with $\varepsilon=\delta=0.05$) for the \propM strategy with and without control variates, when the side information is generate with $c=0.9$. Due to the high correlation, there is a significant decrease in the samples needed to reach an accuracy of $\varepsilon$, when using control variates.

        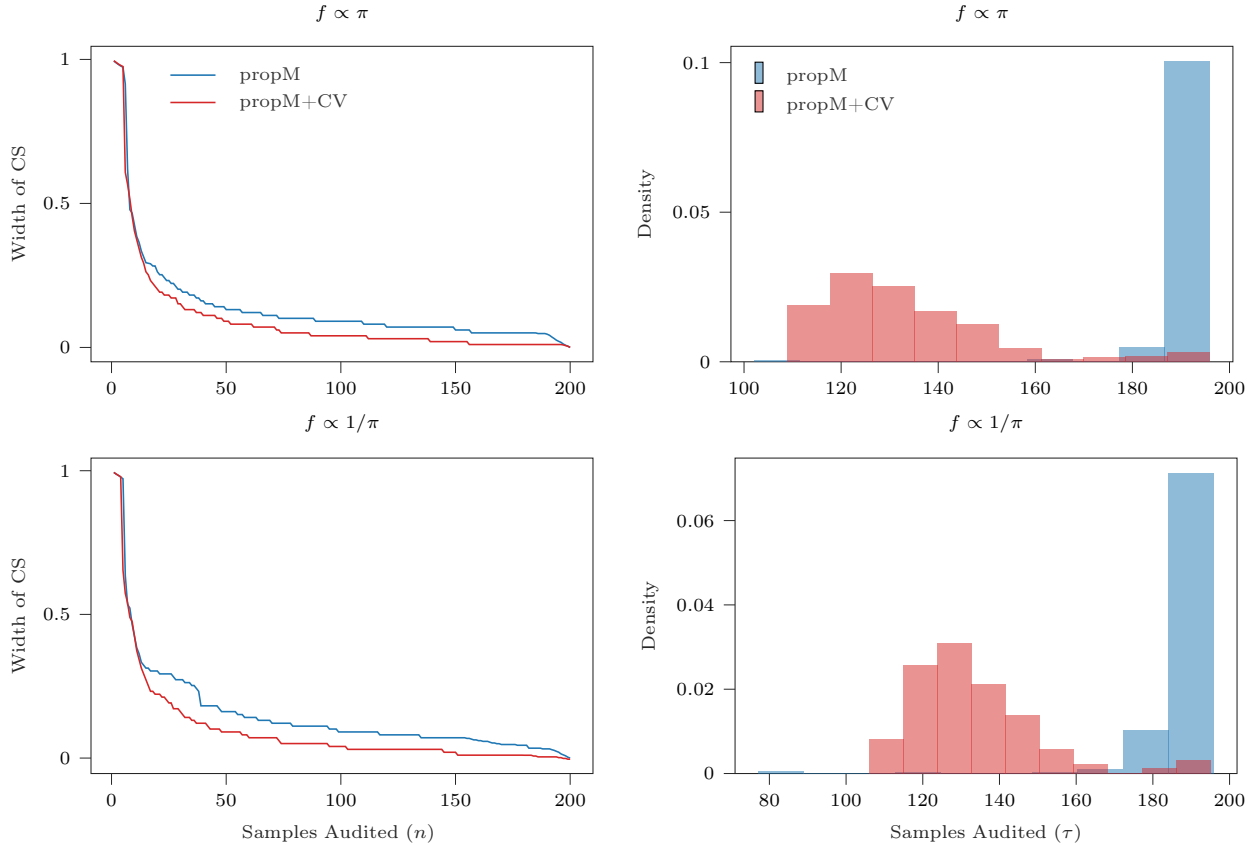
\begin{figure}[t]
            \def\figwidth{0.5\columnwidth}
            \def\figheight{0.35\columnwidth} %
            \hspace*{-1em}
            \begin{tikzpicture}

\definecolor{crimson2143940}{RGB}{214,39,40}
\definecolor{darkslategray38}{RGB}{38,38,38}
\definecolor{lightgray204}{RGB}{204,204,204}
\definecolor{steelblue31119180}{RGB}{31,119,180}

\tikzstyle{every node}=[font=\scriptsize]
\begin{axis}[
axis line style={darkslategray38},
height=\figheight,
legend cell align={left},
legend style={fill opacity=0.8,
    draw opacity=1,
    at={(0.14,0.97)},
    anchor=north west,
    text opacity=1, draw=none},
tick align=outside,
tick pos=left,
title={\(\displaystyle f \propto \pi\)},
width=\figwidth,
x grid style={lightgray204},
xmin=-8.95, xmax=209.95,
xtick style={color=darkslategray38},
y grid style={lightgray204},
ylabel=\textcolor{darkslategray38}{Width of CS},
ymin=-0.0497855178642922, ymax=1.04549587515009,
ytick style={color=darkslategray38}
]
\addplot [semithick, steelblue31119180]
table {%
1 0.994130457160991
2 0.988260914321981
3 0.982391371482972
4 0.977187040558113
5 0.972922319185073
6 0.91475778062222
7 0.619726796206534
8 0.478302367691526
9 0.465999227790727
10 0.42558597404335
11 0.385173626737221
12 0.364963305104036
13 0.334650487192459
14 0.314438084495461
15 0.294228318797445
16 0.292353031850916
17 0.290464204389029
18 0.282828282828283
19 0.282828282828283
20 0.262626262626263
21 0.252525252525253
22 0.252525252525253
23 0.242424242424242
24 0.232323232323232
25 0.232323232323232
26 0.222222222222222
27 0.222222222222222
28 0.212121212121212
29 0.202020202020202
30 0.202020202020202
31 0.191919191919192
32 0.191919191919192
33 0.191919191919192
34 0.181818181818182
35 0.181818181818182
36 0.181818181818182
37 0.171717171717172
38 0.171717171717172
39 0.161616161616162
40 0.161616161616162
41 0.151515151515152
42 0.151515151515152
43 0.151515151515152
44 0.151515151515152
45 0.141414141414141
46 0.141414141414141
47 0.141414141414141
48 0.141414141414141
49 0.141414141414141
50 0.131313131313131
51 0.131313131313131
52 0.131313131313131
53 0.131313131313131
54 0.131313131313131
55 0.131313131313131
56 0.131313131313131
57 0.121212121212121
58 0.121212121212121
59 0.121212121212121
60 0.121212121212121
61 0.121212121212121
62 0.121212121212121
63 0.121212121212121
64 0.121212121212121
65 0.121212121212121
66 0.111111111111111
67 0.111111111111111
68 0.111111111111111
69 0.111111111111111
70 0.111111111111111
71 0.111111111111111
72 0.111111111111111
73 0.101010101010101
74 0.101010101010101
75 0.101010101010101
76 0.101010101010101
77 0.101010101010101
78 0.101010101010101
79 0.101010101010101
80 0.101010101010101
81 0.101010101010101
82 0.101010101010101
83 0.101010101010101
84 0.101010101010101
85 0.101010101010101
86 0.101010101010101
87 0.101010101010101
88 0.101010101010101
89 0.0909090909090909
90 0.0909090909090909
91 0.0909090909090909
92 0.0909090909090909
93 0.0909090909090909
94 0.0909090909090909
95 0.0909090909090909
96 0.0909090909090909
97 0.0909090909090909
98 0.0909090909090909
99 0.0909090909090909
100 0.0909090909090909
101 0.0909090909090909
102 0.0909090909090909
103 0.0909090909090909
104 0.0909090909090909
105 0.0909090909090909
106 0.0909090909090909
107 0.0909090909090909
108 0.0909090909090909
109 0.0909090909090909
110 0.0808080808080808
111 0.0808080808080808
112 0.0808080808080808
113 0.0808080808080808
114 0.0808080808080808
115 0.0808080808080808
116 0.0808080808080808
117 0.0808080808080808
118 0.0808080808080808
119 0.0808080808080808
120 0.0707070707070707
121 0.0707070707070707
122 0.0707070707070707
123 0.0707070707070707
124 0.0707070707070707
125 0.0707070707070707
126 0.0707070707070707
127 0.0707070707070707
128 0.0707070707070707
129 0.0707070707070707
130 0.0707070707070707
131 0.0707070707070707
132 0.0707070707070707
133 0.0707070707070707
134 0.0707070707070707
135 0.0707070707070707
136 0.0707070707070707
137 0.0707070707070707
138 0.0707070707070707
139 0.0707070707070707
140 0.0707070707070707
141 0.0707070707070707
142 0.0707070707070707
143 0.0707070707070707
144 0.0707070707070707
145 0.0707070707070707
146 0.0707070707070707
147 0.0707070707070707
148 0.0707070707070707
149 0.0707070707070707
150 0.0606060606060606
151 0.0606060606060606
152 0.0606060606060606
153 0.0606060606060606
154 0.0606060606060606
155 0.0606060606060606
156 0.0606060606060606
157 0.0505050505050505
158 0.0505050505050505
159 0.0505050505050505
160 0.0505050505050505
161 0.0505050505050505
162 0.0505050505050505
163 0.0505050505050505
164 0.0505050505050505
165 0.0505050505050505
166 0.0505050505050505
167 0.0505050505050505
168 0.0505050505050505
169 0.0505050505050505
170 0.0505050505050505
171 0.0505050505050505
172 0.0505050505050505
173 0.0505050505050505
174 0.0505050505050505
175 0.0505050505050505
176 0.0505050505050505
177 0.0505050505050505
178 0.0505050505050505
179 0.0505050505050505
180 0.0505050505050505
181 0.0505050505050505
182 0.0505050505050505
183 0.0505050505050505
184 0.0505050505050505
185 0.0503589729107282
186 0.0482382366721203
187 0.0482279243547146
188 0.0482164073972541
189 0.0482051812006163
190 0.0459062364182213
191 0.0421548088859859
192 0.0362852660469766
193 0.0310282743894524
194 0.0251587315504431
195 0.0211720268963204
196 0.0176707973911546
197 0.0118012545521453
198 0.00744912322878544
199 0.00427133529291512
200 -1.99840144432528e-15
};
\addlegendentry{propM}
\addplot [semithick, crimson2143940]
table {%
1 0.995710357285802
2 0.989840814446792
3 0.983971271607783
4 0.979303710413013
5 0.975163132896874
6 0.608151420752057
7 0.566192378812067
8 0.515151515151515
9 0.454545454545455
10 0.404040404040404
11 0.373737373737374
12 0.343434343434343
13 0.313131313131313
14 0.292929292929293
15 0.262626262626263
16 0.252525252525253
17 0.232323232323232
18 0.222222222222222
19 0.212121212121212
20 0.202020202020202
21 0.191919191919192
22 0.191919191919192
23 0.181818181818182
24 0.181818181818182
25 0.181818181818182
26 0.171717171717172
27 0.171717171717172
28 0.171717171717172
29 0.151515151515152
30 0.151515151515152
31 0.141414141414141
32 0.131313131313131
33 0.131313131313131
34 0.131313131313131
35 0.131313131313131
36 0.131313131313131
37 0.121212121212121
38 0.121212121212121
39 0.121212121212121
40 0.111111111111111
41 0.111111111111111
42 0.111111111111111
43 0.111111111111111
44 0.111111111111111
45 0.111111111111111
46 0.101010101010101
47 0.101010101010101
48 0.101010101010101
49 0.0909090909090909
50 0.0909090909090909
51 0.0909090909090909
52 0.0808080808080808
53 0.0808080808080808
54 0.0808080808080808
55 0.0808080808080808
56 0.0808080808080808
57 0.0808080808080808
58 0.0808080808080808
59 0.0808080808080808
60 0.0808080808080808
61 0.0808080808080808
62 0.0707070707070707
63 0.0707070707070707
64 0.0707070707070707
65 0.0707070707070707
66 0.0707070707070707
67 0.0707070707070707
68 0.0707070707070707
69 0.0707070707070707
70 0.0707070707070707
71 0.0707070707070707
72 0.0606060606060606
73 0.0606060606060606
74 0.0505050505050505
75 0.0505050505050505
76 0.0505050505050505
77 0.0505050505050505
78 0.0505050505050505
79 0.0505050505050505
80 0.0505050505050505
81 0.0505050505050505
82 0.0505050505050505
83 0.0505050505050505
84 0.0505050505050505
85 0.0505050505050505
86 0.0505050505050505
87 0.0404040404040404
88 0.0404040404040404
89 0.0404040404040404
90 0.0404040404040404
91 0.0404040404040404
92 0.0404040404040404
93 0.0404040404040404
94 0.0404040404040404
95 0.0404040404040404
96 0.0404040404040404
97 0.0404040404040404
98 0.0404040404040404
99 0.0404040404040404
100 0.0404040404040404
101 0.0404040404040404
102 0.0404040404040404
103 0.0404040404040404
104 0.0404040404040404
105 0.0404040404040404
106 0.0404040404040404
107 0.0404040404040404
108 0.0404040404040404
109 0.0404040404040404
110 0.0404040404040404
111 0.0404040404040404
112 0.0303030303030303
113 0.0303030303030303
114 0.0303030303030303
115 0.0303030303030303
116 0.0303030303030303
117 0.0303030303030303
118 0.0303030303030303
119 0.0303030303030303
120 0.0303030303030303
121 0.0303030303030303
122 0.0303030303030303
123 0.0303030303030303
124 0.0303030303030303
125 0.0303030303030303
126 0.0303030303030303
127 0.0303030303030303
128 0.0303030303030303
129 0.0303030303030303
130 0.0303030303030303
131 0.0303030303030303
132 0.0303030303030303
133 0.0303030303030303
134 0.0303030303030303
135 0.0303030303030303
136 0.0303030303030303
137 0.0303030303030303
138 0.0303030303030303
139 0.0202020202020202
140 0.0202020202020202
141 0.0202020202020202
142 0.0202020202020202
143 0.0202020202020202
144 0.0202020202020202
145 0.0202020202020202
146 0.0202020202020202
147 0.0202020202020202
148 0.0202020202020202
149 0.0202020202020202
150 0.0202020202020202
151 0.0202020202020202
152 0.0202020202020202
153 0.0202020202020202
154 0.0202020202020202
155 0.0202020202020202
156 0.0101010101010101
157 0.0101010101010101
158 0.0101010101010101
159 0.0101010101010101
160 0.0101010101010101
161 0.0101010101010101
162 0.0101010101010101
163 0.0101010101010101
164 0.0101010101010101
165 0.0101010101010101
166 0.0101010101010101
167 0.0101010101010101
168 0.0101010101010101
169 0.0101010101010101
170 0.0101010101010101
171 0.0101010101010101
172 0.0101010101010101
173 0.0101010101010101
174 0.0101010101010101
175 0.0101010101010101
176 0.0101010101010101
177 0.0101010101010101
178 0.0101010101010101
179 0.0101010101010101
180 0.0101010101010101
181 0.0101010101010101
182 0.0101010101010101
183 0.0101010101010101
184 0.0101010101010101
185 0.0101010101010101
186 0.0101010101010101
187 0.0101010101010101
188 0.0101010101010101
189 0.0101010101010101
190 0.0101010101010101
191 0.0101010101010101
192 0.0101010101010101
193 0.0101010101010101
194 0.0101010101010101
195 0.0101010101010101
196 0.0101010101010101
197 0.00871322094948182
198 0.00641397016073644
199 0.00410241332619077
200 -1.99840144432528e-15
};
\addlegendentry{propM+CV}
\end{axis}

\end{tikzpicture}\hfill
            \begin{tikzpicture}

\definecolor{crimson2143940}{RGB}{214,39,40}
\definecolor{darkslategray38}{RGB}{38,38,38}
\definecolor{lightgray204}{RGB}{204,204,204}
\definecolor{steelblue31119180}{RGB}{31,119,180}

\tikzstyle{every node}=[font=\scriptsize]
\begin{axis}[
axis line style={darkslategray38},
height=\figheight,
legend cell align={left},
legend style={
  fill opacity=0.8,
  draw opacity=1,
  text opacity=1,
  at={(0.03,0.97)},
  anchor=north west,
  draw=none
},
tick align=outside,
tick pos=left,
title={
 \(\displaystyle f \propto \pi\)},
width=\figwidth,
x grid style={lightgray204},
xmin=97.3, xmax=200.7,
xtick style={color=darkslategray38},
y grid style={lightgray204},
ylabel=\textcolor{darkslategray38}{Density},
ymin=0, ymax=0.105446808510639,
ytick style={color=darkslategray38}
]
\draw[draw=none,fill=steelblue31119180,fill opacity=0.5] (axis cs:102,0) rectangle (axis cs:111.4,0.000425531914893617);
\addlegendimage{ybar,ybar legend,draw=none,fill=steelblue31119180,fill opacity=0.5}
\addlegendentry{propM}

\draw[draw=none,fill=steelblue31119180,fill opacity=0.5] (axis cs:111.4,0) rectangle (axis cs:120.8,0);
\draw[draw=none,fill=steelblue31119180,fill opacity=0.5] (axis cs:120.8,0) rectangle (axis cs:130.2,0);
\draw[draw=none,fill=steelblue31119180,fill opacity=0.5] (axis cs:130.2,0) rectangle (axis cs:139.6,0);
\draw[draw=none,fill=steelblue31119180,fill opacity=0.5] (axis cs:139.6,0) rectangle (axis cs:149,0);
\draw[draw=none,fill=steelblue31119180,fill opacity=0.5] (axis cs:149,0) rectangle (axis cs:158.4,0);
\draw[draw=none,fill=steelblue31119180,fill opacity=0.5] (axis cs:158.4,0) rectangle (axis cs:167.8,0.000851063829787233);
\draw[draw=none,fill=steelblue31119180,fill opacity=0.5] (axis cs:167.8,0) rectangle (axis cs:177.2,0);
\draw[draw=none,fill=steelblue31119180,fill opacity=0.5] (axis cs:177.2,0) rectangle (axis cs:186.6,0.00468085106382977);
\draw[draw=none,fill=steelblue31119180,fill opacity=0.5] (axis cs:186.6,0) rectangle (axis cs:196,0.100425531914894);
\draw[draw=none,fill=crimson2143940,fill opacity=0.5] (axis cs:109,0) rectangle (axis cs:117.7,0.0188505747126437);
\addlegendimage{ybar,ybar legend,draw=none,fill=crimson2143940,fill opacity=0.5}
\addlegendentry{propM+CV}

\draw[draw=none,fill=crimson2143940,fill opacity=0.5] (axis cs:117.7,0) rectangle (axis cs:126.4,0.0294252873563218);
\draw[draw=none,fill=crimson2143940,fill opacity=0.5] (axis cs:126.4,0) rectangle (axis cs:135.1,0.0252873563218391);
\draw[draw=none,fill=crimson2143940,fill opacity=0.5] (axis cs:135.1,0) rectangle (axis cs:143.8,0.0170114942528735);
\draw[draw=none,fill=crimson2143940,fill opacity=0.5] (axis cs:143.8,0) rectangle (axis cs:152.5,0.0124137931034483);
\draw[draw=none,fill=crimson2143940,fill opacity=0.5] (axis cs:152.5,0) rectangle (axis cs:161.2,0.00459770114942529);
\draw[draw=none,fill=crimson2143940,fill opacity=0.5] (axis cs:161.2,0) rectangle (axis cs:169.9,0.000919540229885059);
\draw[draw=none,fill=crimson2143940,fill opacity=0.5] (axis cs:169.9,0) rectangle (axis cs:178.6,0.00137931034482758);
\draw[draw=none,fill=crimson2143940,fill opacity=0.5] (axis cs:178.6,0) rectangle (axis cs:187.3,0.00183908045977011);
\draw[draw=none,fill=crimson2143940,fill opacity=0.5] (axis cs:187.3,0) rectangle (axis cs:196,0.00321839080459771);
\end{axis}

\end{tikzpicture}
        
            \hspace*{-1em}
            \begin{tikzpicture}

\definecolor{crimson2143940}{RGB}{214,39,40}
\definecolor{darkslategray38}{RGB}{38,38,38}
\definecolor{lightgray204}{RGB}{204,204,204}
\definecolor{steelblue31119180}{RGB}{31,119,180}

\tikzstyle{every node}=[font=\scriptsize]
\begin{axis}[
axis line style={darkslategray38},
height=\figheight,
legend cell align={left},
legend style={fill opacity=0.8, draw opacity=1, text opacity=1, draw=none},
tick align=outside,
tick pos=left,
title={\(\displaystyle f \propto 1/\pi\)},
width=\figwidth,
x grid style={lightgray204},
xlabel=\textcolor{darkslategray38}{Samples Audited (\(\displaystyle n\)) },
xmin=-8.95, xmax=209.95,
xtick style={color=darkslategray38},
y grid style={lightgray204},
ylabel=\textcolor{darkslategray38}{Width of CS},
ymin=-0.0538041940520794, ymax=1.04399619600975,
ytick style={color=darkslategray38}
]
\addplot [semithick, steelblue31119180]
table {%
1 0.994096178279667
2 0.988192356559335
3 0.983753604278858
4 0.977849782558525
5 0.971945960838193
6 0.635742232333494
7 0.534724329286462
8 0.521985563114468
9 0.464646464646465
10 0.424242424242424
11 0.383838383838384
12 0.363636363636364
13 0.333333333333333
14 0.323232323232323
15 0.313131313131313
16 0.313131313131313
17 0.303030303030303
18 0.303030303030303
19 0.303030303030303
20 0.303030303030303
21 0.292929292929293
22 0.292929292929293
23 0.292929292929293
24 0.292929292929293
25 0.292929292929293
26 0.292929292929293
27 0.282828282828283
28 0.272727272727273
29 0.272727272727273
30 0.272727272727273
31 0.272727272727273
32 0.262626262626263
33 0.262626262626263
34 0.262626262626263
35 0.252525252525252
36 0.252525252525252
37 0.242424242424242
38 0.232323232323232
39 0.181818181818182
40 0.181818181818182
41 0.181818181818182
42 0.181818181818182
43 0.181818181818182
44 0.181818181818182
45 0.181818181818182
46 0.181818181818182
47 0.171717171717172
48 0.161616161616162
49 0.161616161616162
50 0.161616161616162
51 0.161616161616162
52 0.161616161616162
53 0.161616161616162
54 0.161616161616162
55 0.151515151515152
56 0.151515151515152
57 0.151515151515152
58 0.141414141414141
59 0.141414141414141
60 0.141414141414141
61 0.141414141414141
62 0.141414141414141
63 0.141414141414141
64 0.131313131313131
65 0.131313131313131
66 0.131313131313131
67 0.131313131313131
68 0.131313131313131
69 0.131313131313131
70 0.121212121212121
71 0.121212121212121
72 0.121212121212121
73 0.121212121212121
74 0.121212121212121
75 0.121212121212121
76 0.121212121212121
77 0.121212121212121
78 0.121212121212121
79 0.111111111111111
80 0.111111111111111
81 0.111111111111111
82 0.111111111111111
83 0.111111111111111
84 0.111111111111111
85 0.111111111111111
86 0.111111111111111
87 0.111111111111111
88 0.111111111111111
89 0.111111111111111
90 0.111111111111111
91 0.111111111111111
92 0.111111111111111
93 0.111111111111111
94 0.111111111111111
95 0.101010101010101
96 0.101010101010101
97 0.101010101010101
98 0.101010101010101
99 0.0909090909090909
100 0.0909090909090909
101 0.0909090909090909
102 0.0909090909090909
103 0.0909090909090909
104 0.0909090909090909
105 0.0909090909090909
106 0.0909090909090909
107 0.0909090909090909
108 0.0909090909090909
109 0.0909090909090909
110 0.0909090909090909
111 0.0909090909090909
112 0.0909090909090909
113 0.0909090909090909
114 0.0909090909090909
115 0.0909090909090909
116 0.0909090909090909
117 0.0808080808080808
118 0.0808080808080808
119 0.0808080808080808
120 0.0808080808080808
121 0.0808080808080808
122 0.0808080808080808
123 0.0808080808080808
124 0.0808080808080808
125 0.0808080808080808
126 0.0808080808080808
127 0.0808080808080808
128 0.0808080808080808
129 0.0808080808080808
130 0.0808080808080808
131 0.0808080808080808
132 0.0808080808080808
133 0.0808080808080808
134 0.0808080808080808
135 0.0707070707070707
136 0.0707070707070707
137 0.0707070707070707
138 0.0707070707070707
139 0.0707070707070707
140 0.0707070707070707
141 0.0707070707070707
142 0.0707070707070707
143 0.0707070707070707
144 0.0707070707070707
145 0.0707070707070707
146 0.0707070707070707
147 0.0707070707070707
148 0.0707070707070707
149 0.0707070707070707
150 0.0707070707070707
151 0.0707070707070707
152 0.0707070707070707
153 0.0707070707070707
154 0.0707070707070707
155 0.0684885362448297
156 0.0684827349115747
157 0.0658296414840012
158 0.0632748919873287
159 0.0632708239598869
160 0.0606357824206832
161 0.0606296426287139
162 0.0581397098228297
163 0.0581364046580321
164 0.0581299980051845
165 0.0552230107976639
166 0.0525557641379759
167 0.0525500190625642
168 0.0500567985654651
169 0.0500518150774367
170 0.0472197153062718
171 0.0472143758678118
172 0.047207599464978
173 0.0471992897103125
174 0.0471930440211723
175 0.0471841662935815
176 0.047179386738076
177 0.0445724944538779
178 0.0445680078183746
179 0.0445622126263707
180 0.0445544529430474
181 0.0445489102071039
182 0.0344418193642676
183 0.0344366216507126
184 0.0344292789705756
185 0.0344246336151339
186 0.0344177711484391
187 0.0344129989369084
188 0.0319874119950789
189 0.0319817235733615
190 0.0319752985862017
191 0.0319696090277208
192 0.0291889157974013
193 0.0269194484185538
194 0.0235025833786729
195 0.0204283895469103
196 0.0145245678265776
197 0.0111368546990214
198 0.0078313557379317
199 0.00297155562274654
200 -1.55431223447522e-15
};
\addlegendentry{propM}
\addplot [semithick, crimson2143940]
table {%
1 0.994096178279667
2 0.988789310532055
3 0.983543557179251
4 0.979077620500345
5 0.653857349491058
6 0.573041593285548
7 0.540024885169565
8 0.489515305723792
9 0.474747474747475
10 0.424242424242424
11 0.373737373737374
12 0.343434343434343
13 0.313131313131313
14 0.292929292929293
15 0.272727272727273
16 0.252525252525253
17 0.232323232323232
18 0.232323232323232
19 0.222222222222222
20 0.222222222222222
21 0.222222222222222
22 0.212121212121212
23 0.212121212121212
24 0.202020202020202
25 0.191919191919192
26 0.191919191919192
27 0.171717171717172
28 0.171717171717172
29 0.171717171717172
30 0.161616161616162
31 0.151515151515152
32 0.141414141414141
33 0.141414141414141
34 0.141414141414141
35 0.131313131313131
36 0.131313131313131
37 0.121212121212121
38 0.121212121212121
39 0.121212121212121
40 0.121212121212121
41 0.121212121212121
42 0.111111111111111
43 0.101010101010101
44 0.101010101010101
45 0.101010101010101
46 0.101010101010101
47 0.101010101010101
48 0.0909090909090909
49 0.0909090909090909
50 0.0909090909090909
51 0.0909090909090909
52 0.0909090909090909
53 0.0909090909090909
54 0.0909090909090909
55 0.0909090909090909
56 0.0909090909090909
57 0.0808080808080808
58 0.0808080808080808
59 0.0808080808080808
60 0.0707070707070707
61 0.0707070707070707
62 0.0707070707070707
63 0.0707070707070707
64 0.0707070707070707
65 0.0707070707070707
66 0.0707070707070707
67 0.0707070707070707
68 0.0707070707070707
69 0.0707070707070707
70 0.0707070707070707
71 0.0707070707070707
72 0.0707070707070707
73 0.0606060606060606
74 0.0505050505050505
75 0.0505050505050505
76 0.0505050505050505
77 0.0505050505050505
78 0.0505050505050505
79 0.0505050505050505
80 0.0505050505050505
81 0.0505050505050505
82 0.0505050505050505
83 0.0505050505050505
84 0.0505050505050505
85 0.0505050505050505
86 0.0505050505050505
87 0.0505050505050505
88 0.0505050505050505
89 0.0505050505050505
90 0.0505050505050505
91 0.0505050505050505
92 0.0505050505050505
93 0.0505050505050505
94 0.0505050505050505
95 0.0404040404040404
96 0.0404040404040404
97 0.0404040404040404
98 0.0404040404040404
99 0.0404040404040404
100 0.0404040404040404
101 0.0404040404040404
102 0.0404040404040404
103 0.0303030303030303
104 0.0303030303030303
105 0.0303030303030303
106 0.0303030303030303
107 0.0303030303030303
108 0.0303030303030303
109 0.0303030303030303
110 0.0303030303030303
111 0.0303030303030303
112 0.0303030303030303
113 0.0303030303030303
114 0.0303030303030303
115 0.0303030303030303
116 0.0303030303030303
117 0.0303030303030303
118 0.0303030303030303
119 0.0303030303030303
120 0.0303030303030303
121 0.0303030303030303
122 0.0303030303030303
123 0.0303030303030303
124 0.0303030303030303
125 0.0303030303030303
126 0.0303030303030303
127 0.0303030303030303
128 0.0303030303030303
129 0.0303030303030303
130 0.0303030303030303
131 0.0303030303030303
132 0.0303030303030303
133 0.0303030303030303
134 0.0303030303030303
135 0.0303030303030303
136 0.0303030303030303
137 0.0303030303030303
138 0.0303030303030303
139 0.0303030303030303
140 0.0303030303030303
141 0.0303030303030303
142 0.0303030303030303
143 0.0303030303030303
144 0.0303030303030303
145 0.0202020202020202
146 0.0202020202020202
147 0.0202020202020202
148 0.0202020202020202
149 0.0202020202020202
150 0.0202020202020202
151 0.0101010101010101
152 0.0101010101010101
153 0.0101010101010101
154 0.0101010101010101
155 0.0101010101010101
156 0.0101010101010101
157 0.0101010101010101
158 0.0101010101010101
159 0.0101010101010101
160 0.0101010101010101
161 0.0101010101010101
162 0.0101010101010101
163 0.0101010101010101
164 0.0101010101010101
165 0.0101010101010101
166 0.0101010101010101
167 0.0101010101010101
168 0.0101010101010101
169 0.0101010101010101
170 0.0101010101010101
171 0.0101010101010101
172 0.0101010101010101
173 0.0101010101010101
174 0.0101010101010101
175 0.0101010101010101
176 0.0101010101010101
177 0.0101010101010101
178 0.0101010101010101
179 0.0101010101010101
180 0.00984063570828347
181 0.00983524754221776
182 0.00982955912050038
183 0.00982223722297992
184 0.0071061940706762
185 0.00710041406641665
186 0.00416035288646094
187 0.00415301020632391
188 0.00414911966596948
189 0.00414352337937268
190 0.00413435016068509
191 0.00412686664816853
192 0.00412063965891024
193 0.00411495010042934
194 0.00411164493563176
195 0.00152749887500936
196 0.00152259109089825
197 -0.00127652754085156
198 -0.00128115315053762
199 -0.00389897860844124
200 -0.00390417632199624
};
\addlegendentry{propM+CV};
\legend{}
\end{axis}

\end{tikzpicture}\hfill
            \begin{tikzpicture}

\definecolor{crimson2143940}{RGB}{214,39,40}
\definecolor{darkslategray38}{RGB}{38,38,38}
\definecolor{lightgray204}{RGB}{204,204,204}
\definecolor{steelblue31119180}{RGB}{31,119,180}

\tikzstyle{every node}=[font=\scriptsize]
\begin{axis}[
axis line style={darkslategray38},
height=\figheight,
legend cell align={left},
legend style={
  fill opacity=0.8,
  draw opacity=1,
  text opacity=1,
  at={(0.03,0.97)},
  anchor=north west,
  draw=none
},
tick align=outside,
tick pos=left,
title={
 \(\displaystyle f \propto 1/\pi\)},
width=\figwidth,
x grid style={lightgray204},
xlabel=\textcolor{darkslategray38}{Samples Audited ($\tau$)},
xmin=71.05, xmax=201.95,
xtick style={color=darkslategray38},
y grid style={lightgray204},
ylabel=\textcolor{darkslategray38}{Density},
ymin=0, ymax=0.0748235294117648,
ytick style={color=darkslategray38}
]
\draw[draw=none,fill=steelblue31119180,fill opacity=0.5] (axis cs:77,0) rectangle (axis cs:88.9,0.000672268907563025);
\addlegendimage{ybar,ybar legend,draw=none,fill=steelblue31119180,fill opacity=0.5}
\addlegendentry{propM}

\draw[draw=none,fill=steelblue31119180,fill opacity=0.5] (axis cs:88.9,0) rectangle (axis cs:100.8,0);
\draw[draw=none,fill=steelblue31119180,fill opacity=0.5] (axis cs:100.8,0) rectangle (axis cs:112.7,0);
\draw[draw=none,fill=steelblue31119180,fill opacity=0.5] (axis cs:112.7,0) rectangle (axis cs:124.6,0.000336134453781513);
\draw[draw=none,fill=steelblue31119180,fill opacity=0.5] (axis cs:124.6,0) rectangle (axis cs:136.5,0);
\draw[draw=none,fill=steelblue31119180,fill opacity=0.5] (axis cs:136.5,0) rectangle (axis cs:148.4,0);
\draw[draw=none,fill=steelblue31119180,fill opacity=0.5] (axis cs:148.4,0) rectangle (axis cs:160.3,0.000336134453781512);
\draw[draw=none,fill=steelblue31119180,fill opacity=0.5] (axis cs:160.3,0) rectangle (axis cs:172.2,0.00100840336134454);
\draw[draw=none,fill=steelblue31119180,fill opacity=0.5] (axis cs:172.2,0) rectangle (axis cs:184.1,0.0104201680672269);
\draw[draw=none,fill=steelblue31119180,fill opacity=0.5] (axis cs:184.1,0) rectangle (axis cs:196,0.0712605042016808);
\draw[draw=none,fill=crimson2143940,fill opacity=0.5] (axis cs:106,0) rectangle (axis cs:114.9,0.00808988764044943);
\addlegendimage{ybar,ybar legend,draw=none,fill=crimson2143940,fill opacity=0.5}
\addlegendentry{propM+CV}

\draw[draw=none,fill=crimson2143940,fill opacity=0.5] (axis cs:114.9,0) rectangle (axis cs:123.8,0.0256179775280899);
\draw[draw=none,fill=crimson2143940,fill opacity=0.5] (axis cs:123.8,0) rectangle (axis cs:132.7,0.0310112359550562);
\draw[draw=none,fill=crimson2143940,fill opacity=0.5] (axis cs:132.7,0) rectangle (axis cs:141.6,0.021123595505618);
\draw[draw=none,fill=crimson2143940,fill opacity=0.5] (axis cs:141.6,0) rectangle (axis cs:150.5,0.0139325842696629);
\draw[draw=none,fill=crimson2143940,fill opacity=0.5] (axis cs:150.5,0) rectangle (axis cs:159.4,0.00584269662921348);
\draw[draw=none,fill=crimson2143940,fill opacity=0.5] (axis cs:159.4,0) rectangle (axis cs:168.3,0.00224719101123595);
\draw[draw=none,fill=crimson2143940,fill opacity=0.5] (axis cs:168.3,0) rectangle (axis cs:177.2,0);
\draw[draw=none,fill=crimson2143940,fill opacity=0.5] (axis cs:177.2,0) rectangle (axis cs:186.1,0.00134831460674157);
\draw[draw=none,fill=crimson2143940,fill opacity=0.5] (axis cs:186.1,0) rectangle (axis cs:195,0.00314606741573035);
\legend{};
\end{axis}

\end{tikzpicture}
            \caption{The plots above show the width of the CSs and the distribution of $\tau$ for the $f \propto \pi$ and the $f \propto 1/\pi$ cases, where  $\Nlarge/N=0.2$ and $c=0.9$. }
            \label{fig:CV-CS}
        \end{figure}

    Finally, in~\Cref{fig:CV-Gain}, we study the variation in sample efficiency as the correlation between $S$ and $f$ changes~(i.e., by varying $c$). In particular, for $9$ linearly spaced $c$ values in the range $[0.1, 0.9]$, we compute the $\tau$ for an RLFA without ($\tau_{\text{no-CV}}$) and with control variates ($\tau_{\text{CV}}$) over $250$ trials, and then plot the variation of the mean of their ratio, $\tau_{\text{CV}} / \tau_{\text{no-CV}}$.
    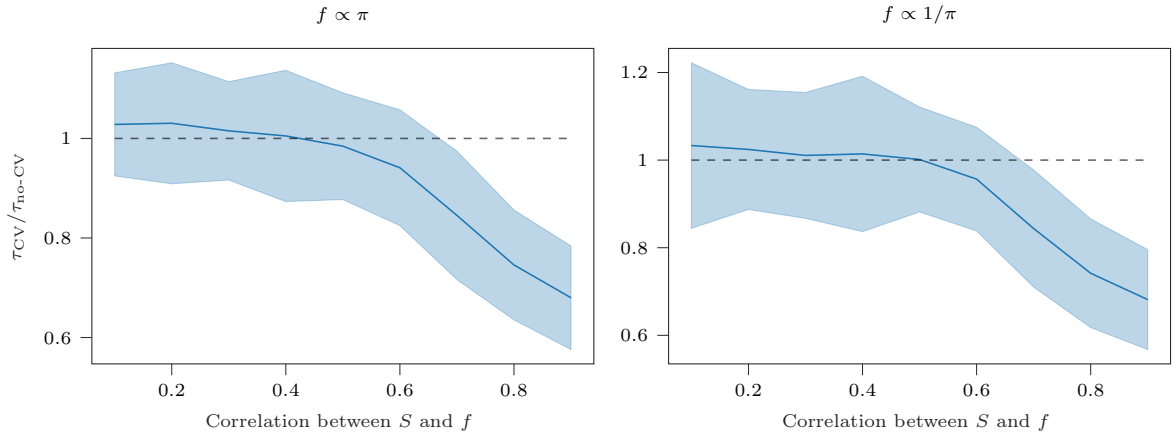
\begin{figure}[h!]
        \def\figwidth{0.5\columnwidth}
        \def\figheight{0.35\columnwidth} %
        \centering
        \hspace*{-1em}
        \begin{tikzpicture}

\definecolor{darkslategray38}{RGB}{38,38,38}
\definecolor{lightgray204}{RGB}{204,204,204}
\definecolor{steelblue31119180}{RGB}{31,119,180}

\tikzstyle{every node}=[font=\scriptsize]
\begin{axis}[
axis line style={darkslategray38},
height=\figheight,
tick align=outside,
tick pos=left,
title={$f \propto \pi$},
width=\figwidth,
x grid style={lightgray204},
xlabel=\textcolor{darkslategray38}{Correlation between \(\displaystyle S\) and \(\displaystyle f\) },
xmin=0.06, xmax=0.94,
xtick style={color=darkslategray38},
y grid style={lightgray204},
ylabel=\textcolor{darkslategray38}{$\tau_{\text{CV}} / \tau_{\text{no-CV}}$},
ymin=0.547065771174321, ymax=1.18082835565563,
ytick style={color=darkslategray38}
]
\path [draw=steelblue31119180, fill=steelblue31119180, opacity=0.3]
(axis cs:0.1,1.1316000527723)
--(axis cs:0.1,0.924691799634826)
--(axis cs:0.2,0.908885384661092)
--(axis cs:0.3,0.916365789053448)
--(axis cs:0.4,0.873148615712509)
--(axis cs:0.5,0.877236593401993)
--(axis cs:0.6,0.824895212114764)
--(axis cs:0.7,0.71635555802216)
--(axis cs:0.8,0.63545898189379)
--(axis cs:0.9,0.575873161378016)
--(axis cs:0.9,0.783970319525466)
--(axis cs:0.9,0.783970319525466)
--(axis cs:0.8,0.856179007956092)
--(axis cs:0.7,0.974431586573793)
--(axis cs:0.6,1.05739363587291)
--(axis cs:0.5,1.09171298392158)
--(axis cs:0.4,1.13688131443937)
--(axis cs:0.3,1.11417430485676)
--(axis cs:0.2,1.15202096545193)
--(axis cs:0.1,1.1316000527723)
--cycle;

\addplot [semithick, steelblue31119180]
table {%
0.1 1.02814592620356
0.2 1.03045317505651
0.3 1.0152700469551
0.4 1.00501496507594
0.5 0.984474788661788
0.6 0.941144423993837
0.7 0.845393572297977
0.8 0.745818994924941
0.9 0.679921740451741
};
\addplot [semithick, black, opacity=0.6, dashed]
table {%
0.1 1
0.2 1
0.3 1
0.4 1
0.5 1
0.6 1
0.7 1
0.8 1
0.9 1
};
\end{axis}

\end{tikzpicture}
        \begin{tikzpicture}

\definecolor{darkslategray38}{RGB}{38,38,38}
\definecolor{lightgray204}{RGB}{204,204,204}
\definecolor{steelblue31119180}{RGB}{31,119,180}

\tikzstyle{every node}=[font=\scriptsize]
\begin{axis}[
axis line style={darkslategray38},
height=\figheight,
tick align=outside,
tick pos=left,
title={$f \propto 1/\pi$},
width=\figwidth,
x grid style={lightgray204},
xlabel=\textcolor{darkslategray38}{Correlation between \(\displaystyle S\) and \(\displaystyle f\) },
xmin=0.06, xmax=0.94,
xtick style={color=darkslategray38},
y grid style={lightgray204},
ylabel={},
ymin=0.534775703427935, ymax=1.25496366491833,
ytick style={color=darkslategray38}
]
\path [draw=steelblue31119180, fill=steelblue31119180, opacity=0.3]
(axis cs:0.1,1.22222784848695)
--(axis cs:0.1,0.844299389345639)
--(axis cs:0.2,0.887274487110602)
--(axis cs:0.3,0.867043059931859)
--(axis cs:0.4,0.836745722151831)
--(axis cs:0.5,0.881643013402391)
--(axis cs:0.6,0.838142680624867)
--(axis cs:0.7,0.710054016237007)
--(axis cs:0.8,0.617690616905216)
--(axis cs:0.9,0.567511519859317)
--(axis cs:0.9,0.795625764730055)
--(axis cs:0.9,0.795625764730055)
--(axis cs:0.8,0.865728720917731)
--(axis cs:0.7,0.977514538624471)
--(axis cs:0.6,1.074994017744)
--(axis cs:0.5,1.12125089262325)
--(axis cs:0.4,1.19149522805915)
--(axis cs:0.3,1.15436051061344)
--(axis cs:0.2,1.16136266125258)
--(axis cs:0.1,1.22222784848695)
--cycle;

\addplot [semithick, steelblue31119180]
table {%
0.1 1.0332636189163
0.2 1.02431857418159
0.3 1.01070178527265
0.4 1.01412047510549
0.5 1.00144695301282
0.6 0.956568349184432
0.7 0.843784277430739
0.8 0.741709668911473
0.9 0.681568642294686
};
\addplot [semithick, black, opacity=0.6, dashed]
table {%
0.1 1
0.2 1
0.3 1
0.4 1
0.5 1
0.6 1
0.7 1
0.8 1
0.9 1
};
\end{axis}

\end{tikzpicture}
        \caption{The figures plot the variation of the reduction in $\tau$ for an RLFA with $\varepsilon=0.025, \delta=0.05$, when the CS is constructed with and without using control variates . The x-axis denotes the parameter $c \in [0.1., 0.9]$, and thus controls the amount of correlation between $S$ and $f$. As the amount of correlation between $S$ and $f$ increases, the CS with control variates decreases takes a decreasing fraction of the time it would take the CS w/o control variates. %
        }
        \label{fig:CV-Gain}
    \end{figure}

    \Cref{fig:CV-Gain} highlights the key advantage of our CS construction using control variates --- this method automatically adapts to the correlation between the side information and the $f$ values. In cases where the side information is highly correlated~(i.e., larger $c$ values), the reduction in samples is large; whereas when the correlation is small, our approach automatically reduces the impact of the side information.

\section{Conclusion}
\label{sec:conclusion}
    In this paper, we defined the concept of an $\rlfa$ and devised RLFA procedures from confidence sequences (CSs) for the weighted average of $N$ terms~(denoted by $m^*$), using adaptive randomized sampling \wor. First, for arbitrary randomized sampling strategies, we developed two methods of constructing CSs for $m^*$ using test martingales. We then addressed the question of effectively incorporating side information, with or without guarantees on their accuracy, to improve the quality of the CSs constructed.

    Our work opens up several interesting directions for future work. For instance, in~\Cref{theorem:oracle-strategy}, we characterized the sampling strategy that optimizes a lower bound on the one-step growth rate. 
    Future work could investigate whether we can obtain a more complete characterization of the optimal policy, without relying on approximations like~\Cref{theorem:oracle-strategy}.
    Another interesting issue, not addressed in our paper is that of considering more general types of side information available to us. As described in~\Cref{sec:introduction},  we have assumed that we have access to $[0,1]$ valued side information that is supposed to be a proxy for the true~(and unknown) $f$ values. However, in practical auditing problems, the side information is usually available in terms of  a collection of numeric, discrete and categorical features that are correlated with the unknown $f$ values. Developing methods for incorporating these more realistic forms of side information into our framework for designing CSs is another important question for future work. Furthermore, another type of side information is any knowledge from a prior audit. For example, auditors may know before reviewing any data (transactions or AI-generated side-info) that for this year, some accounts are likely to have smaller or bigger $f$ values than other accounts because of the specific performance incentives placed on the company managers by their supervisors or by the market conditions.

\vspace{2em}
\noindent \textbf{Acknowledgements.} We acknowledge Ian Waudby-Smith, Justin Whitehouse, Sang Wu, Tom Yan, and our collaborators at PricewaterhouseCoopers for their insightful discussions on this work. We also acknowledge PricewaterhouseCoopers for providing financial support for this research.
\bibliographystyle{abbrvnat}
\bibliography{ref}
\appendix
\newpage 
\section*{Organization}

We discuss some additional background material in~\Cref{appendix:additional-background}, and then formulate the proofs omitted from the body of the paper for \Cref{theorem:oracle-strategy} and \Cref{prop:control-variates-1} in \Cref{sec:Proofs}. We discuss a practical auditing consideration, where one may wish to estimate the remaining misstated fraction as opposed to the total misstated fraction, $m^*$, in \Cref{sec:alt-defs}. In \Cref{sec:hoeffding-empirical-bernstein}, we provide Hoeffding and empirical-Bernstein style CSs that are an alternative to the betting CSs we provide in the paper. We then use the Hoeffding CS constructed in the aforementioned to section as an example of the fact that previous CSs for unweighted mean estimation with uniform sampling from \citet{waudby2020confidence,waudby2020estimating} can be recovered by our method for a specific choice of $(\lambda_t)$ in \Cref{sec:HoefEBComparison}. Empirical results from simulations are shown for these CSs in \Cref{sec:HoefEBExperiments}, in which we compare the Hoeffding and empirical-Bernstein CSs to the betting CS discussed in the main body of the paper. Finally, we end by presenting the results of applying our methods on a real-world housing dataset in~\Cref{appendix:housing-data}. 

\section{Additional Background} 
\label{appendix:additional-background} 

    \subsection{Related Work on Confidence Sequences~(CS)}
    \label{appendix:related-work-CS}
        Confidence sequences~(CSs) are a fundamental tool in sequential analysis, and were introduced into this literature by Robbins and coauthors in a series of papers starting with~\citep{darling1967confidence}. Some other important early works in this area include~\citep{lai1976confidence} and~\citep{jennison1989interim}.
        More recently, there has been a resurgence of interest in confidence sequences, particularly motivated by its applications in anytime valid inference. In anytime valid inference, data samples are received sequentially, and the goal is to derive statistical guarantees that are valid even if one stops sampling and performs inference at a data-dependent time \citep{johari2015always, johari2017peeking, howard2021time, howard2022sequential}. Thus, confidence sequences can also be employed in the multi-armed bandit setting (where one can sample adaptively from multiple streams of data) to enable best arm identification \citep{jamieson2014best}. 
        Most of the papers mentioned above rely on making certain moment assumptions on the data-generating distributions. An important line of recent work aims to relax these assumptions, by  constructing confidence sequences for heavy-tailed data or contaminated data~\citep{wang2022catoni, bhatt2022catoni, mineiro2022lower, wang2023huber}. 
    
    \subsection{Betting-based CS construction}
    \label{appendix:betting-cs-construction}
        The betting-based approach for constructing confidence sequences builds upon the idea of \emph{testing-by-betting}, popularized recently by~\citet{shafer_testing_betting_2021}. This principle states that we can refute a claim~(equivalently, a null hypothesis $H_0$) about the probability distribution generating some data stream, if we can increase our wealth by repeatedly betting on the observations with the restriction that the betting payoffs are \emph{fair} under $H_0$. The restriction of fair payoff implies that that bettor is not expected to make large gains under $H_0$, irrespective of the betting strategy employed~(formally, the wealth process is a non-negative supermartingale). Consequently, if the bettor ends up making a large profit by betting on the observations, this can be considered as evidence against $H_0$; with the relative growth in wealth provide a precise measure of the strength of evidence. 

        To use the above principle for constructing CSs, we simultaneously play a continuum of betting games, indexed by $m \in [0,1]$, each with an initial wealth of $\$1$. For every $m \in [0,1]$, we bet against the claim $H_{0,m}$ that the true misspecified fraction $m^*$ is equal to $m$. We design the payoff functions of this betting game, such that the resulting wealth process, $\{W_t(m): t \geq 1\}$, is a non-negative martingale if $H_{0,m}$ were true, but grows at an exponential rate otherwise. Due to this property, the process at  $m^*$, denoted by $\{W_t(m^*): t \geq 1\}$,  is actually a \emph{test martingale}; that is, a nonnegative martingale with an initial value $1$. Hence, Ville's inequality~(recalled below in~\Cref{fact:ville}) implies that with probability at least $1-\alpha$, the process $(W_t(m^*))$ never exceeds the value $1/\alpha$. This fact, suggests a natural definition of a CS for $m^*$, consisting of sets $C_t = \{m: W_t(m)< 1/\alpha\}$, since these sets contain $m^*$ for all $t \geq 1$, with probability at least $1-\alpha$.
        To conclude, the betting-based approach breaks the task of constructing confidence sequences into two smaller tasks: 
        \begin{enumerate}
            \item Choosing a sequence of payoff functions that are fair under $H_{0,m}$: we achieve this by using the idea of importance weighting. 
            \item Developing a betting strategy that ensures fast growth of the wealth process for all $m \neq m^*$: we use the \kelly strategy for this in our construction. 
        \end{enumerate}
        An additional design choice that is unique to the problem studied in this paper, is that of the \emph{sampling strategy} to select the transaction indices. We discuss several strategies in~\Cref{sec:sampling-strategies}, whose performance is determined by the availability and accuracy of side-information. 

        We end this discussion by recalling a statement of Ville's inequality~\citep{ville1939etude}. 
        \begin{fact}[Ville's Inequality]
            \label{fact:ville} 
            Suppose $\{M_t: t \geq 0\}$ denotes a nonnegative supermartingale adapted to a filtration $\{\mc{F}_t: t \geq 0\}$. Then, for any $\alpha >0$, we have 
            \begin{align}
                \mathbb{P}\lp \exists t \geq 0: M_t \geq 1/\alpha \rp \leq \frac{\mathbb{E}[M_0]}{\alpha}.  
            \end{align}
        \end{fact}

    \subsection{Working with minibatches}
    \label{appendix:minibatch}
        In this paper, we have developed our methodology under the assumption that the transactions are sampled and sent to the human auditor, one at a time. In practical scenarios, it may be preferable for the human auditor to evaluate a minibatch of transactions a time, rather than querying the transactions one-by-one. This generalization can be easily handled by updating the wealth process with the averaged payoff (over the minibatch in each round). More specifically, we can proceed as follows, for any $m \in [0,1]$, and for $t=1, 2, \ldots$: 
        \begin{itemize}
            \item Calculate the next sampling distribution, $q_t$. 
            \item Sample the next batch of transactions, $\mc{B}_t \defined \{I_t^{(1)}, \ldots, I_t^{(B)}\}$, with $I_t^{(j)}$ is drawn according to $q_t$ restricted on $\mc{N}_t \setminus \{I_t^{(1)}, \ldots, I_t^{(j-1)}\}$. Recall that, we now have $\mc{N}_t = [N] \setminus \lp \cup_{s=1}^{t-1} \mc{B}_s \rp$. 
            \item Obtain the true $f$ values, $\{f(I_t^{(i)}: 1 \leq i \leq B\}$ from the oracle (i.e., the human auditor). 
            \item Update the wealth process: 
            \begin{align}
                W_t(m) = W_{t-1}(m) \times \lp 1 + \frac{1}{B} \sum_{i=1}^B \lambda_t^{(i)} \lp Z_t^{(i)} - \mu_t^{(i)}(m) \rp \rp, \quad \text{for } m \in [0,1].  
            \end{align}
            Here $\lambda_t^{(i)}$ denotes the bet based on $\cup_{s=1}^{t-1} \mc{B}_t \cup \{X_t^{(1)}, \ldots, X_t^{(i-1)} \}$, and $Z_t^{(i)}-\mu_t^{(i)}$ denotes the analogous payoff function, as defined in~\Cref{subsec:betting-CS-no-side-info}. 
            \item Update the CS, as $C_t = \{m \in [0,1]: W_t(m) < 1/\alpha\}$. 
        \end{itemize}

\section{Proofs}
\label{sec:Proofs}
\subsection{Proof of Proposition~\ref{theorem:oracle-strategy}}
\label{sec:oracle-proof}
    Recall that $\expect_{I_n \sim q}[Z_n] = \mu_n(m^*)$ for any sampling distribution $q \in \Delta^{\mc{N}_n}$. Now, we note the following equivalencies
    \begin{align}
        &\expect_{I_n \sim q}[B_n(\lambda, m)]\\ 
        &= \expect_{I_n \sim q}[\lambda(Z_n - \mu_n(m)) - \lambda^2(Z_n -\mu_n(m))^2]\\
        &=\expect_{I_n \sim q}[\lambda(Z_n - \mu_n(m))] - \expect_{I_n \sim q}[\lambda^2(Z_n -\mu_n(m))^2]\\
        &=\lambda(m^* - m) - \lambda^2\expect_{I_n \sim q}[(Z_n- \mu^*_n(m) -\mu_n(m) + \mu^*_n(m))^2]\\
        &=\lambda(m^* - m) - \lambda^2(m^* - m)^2 + 2\lambda^2\expect_{I_n \sim q}[Z_n- \mu_n(m^*)](m^*-m) - \lambda^2\expect_{I_n \sim q}[(Z_n- \mu^*_n(m))^2]\\
        &= \lambda(m^* - m) - \lambda^2(m^* - m)^2 - \mathbb{V}_{I_n \sim q}[Z_n].
    \end{align}
    
    The above equivalencies show that $q_n^* = \argmin_{q \in \Delta^{\mc{N}_{n}}}\ \mathbb{V}_{I_n \sim q}[Z_n]$ by definition of $q_n^*$ in \eqref{eq:max-bound}, and since $\lambda, m^*, m$ are fixed in the optimization problem. Consequently, the minimizer of $\mathbb{V}_{I_n \sim q}[Z_n]$ is when the distribution of $Z_n$ has support on only a single value and the variance is 0. This is achieved when $q_n(i) \propto \pi(i)f(i)$ for each $i \in \mc{N}_n$. Hence, we have shown our desired result.

    \subsection{Proof of~ Proposition~\ref{prop:control-variates-1}}\label{proof:control-variates-1}
        For any $t \geq 1$, introduce the random variable $D_t = t \times \lp \Mtilde_t - A_t\rp = \sum_{i=1}^{t} \beta_i U_i$.  By construction of the term $U_i$, we know that for any $t \geq 1$, we have
        \begin{align}
            \mathbb{E}[D_t|\mc{F}_{t-1}] = \sum_{i=1}^{t-1} \beta_i U_i + \beta_t \mathbb{E}[D_t|\mc{F}_{t-1}] = D_{t-1}.
        \end{align}
        Thus, $\{D_t: t \geq 1\}$ is a martingale process. Furthermore, since both $\beta_t$ and $U_t$ lie in the set $[-1,1]$, the martingale process $\{D_t: t \geq 1\}$ has bounded differences.  Hence, by using the time-uniform deviation inequality for martingales with bounded differences~\citep[Eq.~(11)]{howard2021time}, we have
        \begin{align}
            \mathbb{P} \lp \exists t \leq n: |D_t| > 1.7 \sqrt{t \lp  \log \log(2t) + 0.72\log(10.4/\delta) \rp } \rp < \delta.
        \end{align}
        Since $|\Mtilde_t - A_t|/t = |D_t|/t$, the result follows.

\section{Alternative Definitions of RLFA}
\label{sec:alt-defs}

\paragraph{Testing based $\boldsymbol{\rlfa}$.}  An alternative notion of RLFAs, that mirrors the corresponding definition of risk-limiting audits~\citep{stark_conservative_statistical_2008a} more closely, can be obtained by framing it  as the task of testing whether the overall misstated fraction, $m^*$, is ``small''.  This can be interpreted as auditing the claim that all the reported monetary values are accurate. 
Formally, we assume that the announced assertion is that $m^* \leq \varepsilon$, and we want to design a procedure that (i) with probability at least $(1-\delta)$, rejects this assertion if it is false, and (ii) confirms this assertion with probability $1$ if it is true. Our general, CS-based strategy that we developed in this paper can be easily adapted to this problem --- we simply define the stopping time $\tau(\varepsilon, \delta) \coloneqq \min\{t \in [N]:  \mc{C}_t \subseteq (\varepsilon, 1] \cup \{N\}$ and reject if $C_{\tau(\varepsilon, \delta)} \subseteq (\varepsilon, 1]$. Properties (i) and (ii) follow immediately from the definition of a CS, and the consistency of a CS, i.e., $|\mc{C}_t| \rightarrow 0$ as $t \rightarrow N$.

\paragraph{Auditing the remaining misstated fraction.} Another quantity of interest is the \textit{remaining} misstated fraction, i.e., what is the remaining misstated fraction assuming we correct the transaction values for the transactions we have audited? In many practical auditing scenarios, the company does correct their finances in accordance with their records for the audited transactions, i.e., $f(I_i) = 0$ for each $i \in [t]$ after querying the $t$th transaction. This is equivalent to estimating $\mu_t(m^*)$, a time varying quantity, instead of the static quantity of $m^*$. Since $\mu_t(m^*)$ is simply a shift of $m^*$ by a quantity that is known to the auditor, all estimates of $m^*$ we produce in this work can easily be transformed into estimates of the remaining misstated fraction as well (by subtracting $\sum_{i \in [t]} \pi(I_i)f(I_i)$ from both boundaries of our CSs). Thus, we can estimate the remaining misstated fraction as efficiently as we estimate $m^*$.

\section{Hoeffding and empirical-Bernstein Confidence sequences}\label{sec:hoeffding-empirical-bernstein}

    In this section, we present a different approach for constructing confidence sequences, that are based on nonnegative supermartingales (NSMs), instead of nonnegative martingales used by the betting based CS (\Cref{subsec:betting-CS-no-side-info}). While these CSs are typically looser than the betting CS defined in \eqref{eq:conf-seq-def-1}, they are computationally inexpensive and can be derived analytically. We will introduce two such CSs, the Hoeffding CS and empirical-Bernstein CS, and each will have boundaries that take on an explicit form. Thus, only constant time is needed to compute the boundaries for each new sample. In contrast, the betting CS computes its boundaries through a root finding procedure which requires $O(t)$ computations to derive updated boundaries after receiving the $t$th sample. We provide simulations comparing the Hoeffding and empirical-Bernstein CSs with the betting CS in \Cref{sec:HoefEBExperiments}.
    Before defining our CSs, we  first introduce the following quantities:
    \begin{align}
        \widehat{m}_t \coloneqq \frac{\pi(I_t)}{q_t(I_t)}f(I_t) + \sum\limits_{i = 1}^{t - 1}\pi(I_i) f(I_i), \quad
        \widehat{\mu}_t \coloneqq  \frac{\sum\limits_{i = 1}^t \widehat{m}_i}{t}, \qquad \widehat{\mu}_t(\lambda_1^t) \coloneqq  \frac{\sum\limits_{i = 1}^t \lambda_i\widehat{m}_i}{\sum\limits_{i = 1}^t \lambda_i}.
    \end{align} Note that \(\widehat{\mu}_t(\lambda_1^t)\) where \(\lambda_1 = \ldots= \lambda_t = 1\) is equivalent to \(\widehat{\mu}_t\).

    \subsection{Hoeffding CS}\label{subsec:hoeffding-cs}
        To define the Hoeffding CS, we first define the nonnegative supermartingale (NSM) associated with it.
        As noted in prior work \citep{howard2021time,waudby2020estimating}, let the following be a CGF-like function for Hoeffding:
        \begin{align}
            \psi_{\rmH}^c(\lambda) \coloneqq \frac{\lambda^2c^2}{8},
        \end{align} for any fixed \(c > 0\). Now we define the following Hoeffding NSM as follows:
        \begin{align}
            M_t^{\rmH}(m) &\coloneqq \exp\left(\sum\limits_{i= 1}^t \lambda_i(Z_i - \mu_i(m)) - \psi_{\rmH}^{c_i}(\lambda_i)\right) =\exp\left(\sum\limits_{i= 1}^t \lambda_i(\widehat{m}_i - m) - \psi_{\rmH}^{c_i}(\lambda_i)\right),
        \end{align} where  \(c_t \geq \max_{i \in \mathcal{U}_t} \pi(i) / q_t(i)\), and both \((\lambda_t)\) and \((c_t)\) are predictable w.r.t.\ \((\filtration_t)\).

        \begin{proposition}
        \label{prop:hoeffding-nsm}
            \((M_t^{\rmH}(m^*))_{t \in [N]}\) is an NSM.
        \end{proposition}
        \begin{proof}
            First, note that \(\expect[Z_t \mid \filtration_{t - 1}]=\mu_t(m)\) since we are assuming the null $H_{0,m}$ is true. Second, note that \(Z_t \in [0, \max_{i \in \mathcal{U}_t} \pi(i) / q_t(i)]\) is bounded.
            Thus, the desired statement follows directly from the MGF bound on bounded random variables i.e.\ if \(X \in [\ell, u]\) is a random variable, then \(\expect[\exp(\lambda(X - \expect[X]))] \leq \exp(\lambda^2(u - \ell)^2 / 8)\) for any \(\lambda \in \reals\).
        \end{proof}

        Consequently, we can derive the following Hoeffding CS:
        \begin{align}
            C^{\rmH}_t \coloneqq \left(\widehat{\mu}_t(\lambda_1^t) \pm \frac{\log(2 / \alpha) + \sum\limits_{i = 1}^t \psi_{\rmH}^{c_i}(\lambda_i)}{\sum\limits_{i = 1}^t \lambda_i}\right) \cap [0, 1],
            \label{eqn:hoeffding-cs}
        \end{align} where  \(c_t \geq \max_{i \in \mathcal{U}_t} \pi(i) / q_t(i)\), and both \((\lambda_t)\) and \((c_t)\) are predictable w.r.t.\ \((\filtration_t)\).

    \subsection{Empirical-Bernstein CS}

        Define the following CGF-like function for empirical-Bernstein:
        \begin{align}
            \psi_{\E}^c \coloneqq \frac{-\log(1 - c\lambda) - c\lambda}{c^2},
        \end{align} for any \(c > 0\). Now we define the following empirical-Bernstein NSM:
        \begin{align}
            M_t^{\EB}(m) &\coloneqq \exp\left(\sum\limits_{i= 1}^t \lambda_i(Z_i - \mu_i(m)) - (Z_i - \widehat{\mu}_{i - 1})^2\psi_{\E}^{c_i}(\lambda_i)\right)\\
            &=\exp\left(\sum\limits_{i= 1}^t \lambda_i(\widehat{m}_t - m) - (Z_i - \widehat{\mu}_{i - 1})^2\psi_{\E}^{c_i}(\lambda_i)\right),
        \end{align} where \(\lambda_t \in [0, 1 / c_t)\), \(c_t \geq \widehat{\mu}_{t - 1}\), and both \((\lambda_t)\) and \((c_t)\) are predictable w.r.t.\ \((\filtration_t)\).

         \paragraph{Constructing the upper CS through mirroring.} \(M_t^{\EB}\) can be used to construct a CS that lower bounds \(m^*\), but naively constructing an analog NSM (i.e., by negating $ Z_t - \mu_t(m)$ into $\mu_t(m) - Z_t$) results in a loose construction for the upper CS, since $c_t$ would need to lower bound $\widehat{\mu}_{t - 1} - Z_t$. $Z_t$ (and hence $c_t$) be quite large depending on the sampling probabilities $q_t$. Thus, we use the fact that $m^* \in [0, 1]$ and hence $1 - m^* \in [0, 1]$ to construct a ``mirroring'' lower CS for $1 - m^* = \sum_{i = 1}^N \pi(i)(1 - f(i))$. The lower CS for $1 - m^*$ is based upon the following NSM:
        \begin{align}
            {M'}_t^{\EB}(m) &\coloneqq \exp\left(\sum\limits_{i= 1}^t \lambda_i\left(\widetilde{Z}_i - \left(1 - m + \sum\limits_{j = 1}^{i - 1}\pi(I_j)(1 - f(I_j))\right) - (\widetilde{Z}_i - \widetilde{\mu}_{i - 1})^2\psi_{\E}^{c_i}(\lambda_i)\right)\right)\\ &=\exp\left(\sum\limits_{i= 1}^t \lambda_i(\widetilde{m}_t - (1 - m)) - (\widetilde{Z}_i - \widetilde{\mu}_{i - 1})^2\psi_{\E}^{c_i}(\lambda_i)\right).
        \end{align} Here, define $\widetilde{Z}_t \coloneqq \frac{\pi(I_t)}{q_t(I_t)}(1 - f(I_t))$, and let $\widetilde{m}_t$, and  $\widetilde{\mu}_t$ be counterparts of $\widehat{m}_t$ and $\widehat{\mu}_t$ where $f(i)$ is replaced with $1 - f(i)$ in the respective definitions for each \(i \in [N]\).
        \begin{proposition}
        \label{prop:EBNSM}
        \((M_t^{\EB}(m^*))_{t \in [N]}\) and \(({M_t'}^{\EB}(m^*))_{t \in [N]}\) are both NSMs.
        \end{proposition}
    To prove \Cref{prop:EBNSM}, we introduce the following key lemma from \citet{fan_exponential_inequalities_2015}.
\begin{lemma}[{\citealt[Lemma 4.1]{fan_exponential_inequalities_2015}}]
    \label{lemma:Fan}
    Let \(\xi\) be a number bounded from below i.e.\ satisfies \(\xi \geq -c\), where \(c \in \reals^+\) is a fixed constant. Let the following be true: \(\lambda \in [0, 1/c)\). Then,
    \begin{align}
    1 + \lambda \xi \geq \exp(\lambda \xi + \xi^2(\log(1 - c\lambda) + c\lambda)).
    \end{align}
    \end{lemma}
    \begin{proof}
    The proof revolves around the following function \(h\):
    \begin{align}
        h(x) \coloneqq \frac{\log(1 + x) - x}{x^2 / 2}, \qquad x > -1.
    \end{align}
    Note that \(f\) is increasing in its domain. Then, \(\lambda \xi \geq -c\lambda > -1\) by definition of \(\lambda\) and \(\xi\). Thus,
    \begin{align}
        h(\lambda \xi) &\geq h(-c\lambda) \Leftrightarrow \frac{\log(1 + \lambda \xi) - \lambda \xi}{\xi^2}  \geq \frac{\log(1 - c\lambda) + c\lambda}{c^2}.
    \end{align} The desired statement follows from expanding this inequality and rearranging terms.
    \end{proof}

    \begin{proof}[Proof of \Cref{prop:EBNSM}] We will only show the proof that $(M^{\EB}_t(m^*))$ is an NSM, since the proof that $({M'}_t^{\EB}(m^*))$ is an NSM will follow a similar derivation.
    Let \(Y_t = Z_i - \mu_i(m)\) and \(\delta_t = \widehat{\mu}_{t - 1} - \mu_i(m)\).
    Note that \(Y_t - \delta_t = Z_t - \widehat{\mu}_{t - 1}\). To prove our desired statement, it suffices to show the following is true:
    \begin{align}
    \expect\left[\exp(\lambda_t Y_t - (Y_t - \widehat{\mu}_{t - 1})^2\psi_{\E}^{\widehat{\mu}_{t - 1}}) \mid \filtration_{t - 1}\right] \leq 1.
    \label{eqn:EBIncrement}
    \end{align}

    We will now show that \eqref{eqn:EBIncrement} is indeed true:

    \begin{align}
    \expect\left[\exp(\lambda_t Y_t - (Y_t - \delta_t)^2\psi_{\E}^{\widehat{\mu}_{t - 1}}) \mid \filtration_{t - 1}\right] &= \expect\left[\exp(\lambda_t (Y_t - \delta_t) - (Y_t - \delta_t)^2\psi_{\E}^{\widehat{\mu}_{t - 1}}) \mid \filtration_{t - 1}\right]\exp(\lambda_t \delta_t)\\
    &\leq\expect\left[1 + \lambda_t (Y_t - \delta_t)  \mid \filtration_{t - 1}\right]\exp(\lambda_t \delta_t)\\
    &= \expect\left[1 - \lambda_t \delta_t  \mid \filtration_{t - 1}\right]\exp(\lambda_t \delta_t) \leq 1.
    \end{align}

    The 1st inequality is by application of \Cref{lemma:Fan} with \(\xi_t = Y_t - \delta_t\) and \(c = \widehat{\mu}_{t - 1}\) (\(Z_t \geq 0\), so \(\xi_t \geq -\widehat{\mu}_{t - 1}\)). The 2nd equality is the result of \(\expect[Y_t \mid \filtration_{t - 1}] = (m^* - m^*_t) - (m^* - m^*_t) = 0\). The last inequality is by \(1 - x \leq \exp(-x)\) for all \(x \in \reals\).
    Thus, we have proved \eqref{eqn:EBIncrement} and our desired statement as a result.
    \end{proof}
    As a result, we can construct the following empirical-Bernstein CS:
    \begin{align}
        C_t^{\EB} \coloneqq &\left(\widehat{\mu}({\lambda'}_1^t) - \frac{\log(2 / \alpha) + \sum\limits_{i = 1}^t(Z_i - \widehat
        {\mu}_{i - 1})^2\psi_E^c(\lambda_i')}{\sum\limits_{i= 1}^t \lambda_i'}, 1 - \widetilde{\mu}_t({\lambda}_1^t) + \frac{\log(2 / \alpha) + \sum\limits_{i = 1}^t(\widetilde{Z}_t - \widetilde{\mu}_{t - 1})^2\psi_E^c(\lambda_i)}{\sum\limits_{i = 1}^t \lambda_i} \right)
        \cap [0, 1].
    \end{align} The mirroring trick of constructing a lower CS for $1 - m^*$ was originally employed to construct confidence bounds for off policy evaluation in contextual bandits \cite{thomas_highconfidence_offpolicy_2015,waudby-smith_anytimevalid_offpolicy_2022}. To the best of our knowledge, this is the first use of the mirroring trick simply for mean estimation.

        For both the Hoeffding and empirical-Bernstein CSs, we show in \Cref{sec:HoefEBComparison} that we are able to recover the unweighted, uniform sampling versions introduced by \citet{waudby2020confidence} as a special case, i.e., when $q_t$ is the uniform distribution over the remaining items and $\pi$ is uniform over all items. Thus, our formulations of the Hoeffding and empirical-Bernstein CSs generalize the CSs for sampling without replacement in \cite{waudby2020confidence} to weighted sampling and estimation.

\section{Connections with Waudby-Smith and Ramdas \cite{waudby2020confidence,waudby2020estimating}}
\label{sec:HoefEBComparison}
CSs for estimation of the unweighted mean through uniform sampling without replacement are shown in \citet{waudby2020confidence} for Hoeffding and empirical-Bernstein style CSs, and \citet{waudby2020estimating} for betting style CSs. In this section, we show these results are a special case of our results that also account for non-uniform sampling strategies and weighted means. For simplicity, we will show the Hoeffding case, and the results for the other CSs follow a similar argument. Following the notation of \citet{waudby2020confidence}, let \((X(i))_{i \in [N]}\) be a finite population of values in $[0, 1]$ (without loss of generality to arbitrary bounds on the $X(i)$). Let \(X_1, \dots, X_N\) be random variables that are the the result of from sampling uniformly w/o replacement from this population. \citet{waudby2020confidence} construct the following NSM for $\mu \coloneqq \tfrac{1}{N}\sum\limits_{i = 1}^N X(i)$:

        \begin{align}
            M_t^{\WSR}(m) \coloneqq \exp\left(\sum\limits_{i = 1}^t\lambda_i\left(X_i- m + \frac{1}{N-i+1}\sum\limits_{j = 1}^{i - 1}(X_j - m)\right) - \psi_H(\lambda_i)\right),
        \end{align} and derive the CS
        \begin{align}
            C_t^{\WSR} \coloneqq \left(\widehat{\mu}_t^{\WSR}(\lambda_1^t) \pm \frac{\log(2 / \alpha) + \sum\limits_{i = 1}^t\psi_H(\lambda_i)}{\sum\limits_{i = 1}^t \lambda_i\left(1 + \frac{i - 1}{N - i + 1}\right)}\right). 
        \end{align}
        The center of this CS is defined as
        \begin{align}
            \widehat{\mu}_t^{\WSR}(\lambda_1^t) \coloneqq \frac{\sum\limits_{i = 1}^t\lambda_i\left(X_i + \frac{1}{N-i + 1}\sum\limits_{j = 1}^{i - 1}X_j\right)}{\sum\limits_{i = 1}^t \lambda_i(1 + \frac{i - 1}{N - i + 1})} = \frac{\sum\limits_{i = 1}^t\lambda_i\left(X_i + \frac{1}{N-i + 1}\sum\limits_{j = 1}^{i - 1}X_j\right)}{\sum\limits_{i = 1}^t \lambda_i \cdot \frac{N}{N - i + 1}}.
        \end{align}

        In the weighted setting, \(\pi(i) = 1 / N\) and \(f(i) = X(i)\) for each \(i \in [N]\) implies that we are estimating the uniformly weighted average $m^* = \mu$. For each \(t \in [N]\) and \(i \in \mc{N}_t\), set \(q_t(i)  = 1 / (N - t + 1)\) to be the uniform distribution over the remaining items. This gets us the following estimate of the mean from our Hoeffding CS:
        \begin{align}
            \widehat{\mu}_t({\lambda_1'}^t) = \frac{\sum\limits_{i = 1}^t \lambda_i'\left(\frac{N - i + 1}{N}X_i + \frac{1}{N}\sum\limits_{j = 1}^{i - 1}X_j\right)}{\sum\limits_{i = 1}^t \lambda_i'}.
        \end{align}

        By setting \(\lambda'_i = \lambda_i N / (N - i + 1)\), for each \(i \in [t]\), we get that $\widehat{\mu}_t({\lambda'}_1^t) = \widehat{\mu}^{\WSR}_t(\lambda_1^t)$. To see that $C_t^H = C_t^{\WSR}$, we set $c_t = (N - t + 1) / N$ for each $t \in [N]$. Note that that is minimum possible value that $c_t$ can be since $\pi(i) = 1 / N$ for each $i \in [N]$, and $q_t(i) = 1/ (N - t + 1)$ for each $i \in \mc{N}_t$. As a result, we are able to recover the Hoeffding CS from \cite{waudby2020confidence} as a special case of our Hoeffding CS.

\section{Experiments Comparing Different CS Constructions}
\label{sec:HoefEBExperiments}
    In \Cref{fig:HoefEBCS}, we compare the width of the Hoeffding, empirical-Bernstein, and betting CSs under the \propM sampling strategy, i.e., under a weighted sampling strategy. We follow the same setup as Experiment 1 in \Cref{sec:experiments}. We see that the empirical-Bernstein CS is tighter in cases where $\Nlarge$ is larger. In these cases, the support size, $c_t$, is large for $\Nlarge$ transactions, with makes the Hoeffding CS looser as a result. On the other hand, empirical-Bernstein is able to take advantage of the low-variance from a large number of transactions having similar misstated fractions $f(I_t)$. However, when $\Nlarge$ is small, most transactions will have a small support size, $c_t$, and the Hoeffding CS will be tighter than empirical-Bernstein as a result. The betting CS is tighter than both Hoeffding and empirical-Bernstein CSs in all our simulated setups. This trend is reflected in \Cref{fig:HoefEBHist} where we plot the histogram of the first time the CS reaches $\varepsilon=0.2$ width, i.e., the empirical-Bernstein CS reaches $\varepsilon$ width faster than the Hoeffding CS when $\Nlarge$ is larger, and the betting CS is the faster than both Hoeffding and empirical-Bernstein CSs.
    \begin{figure}[htb!]
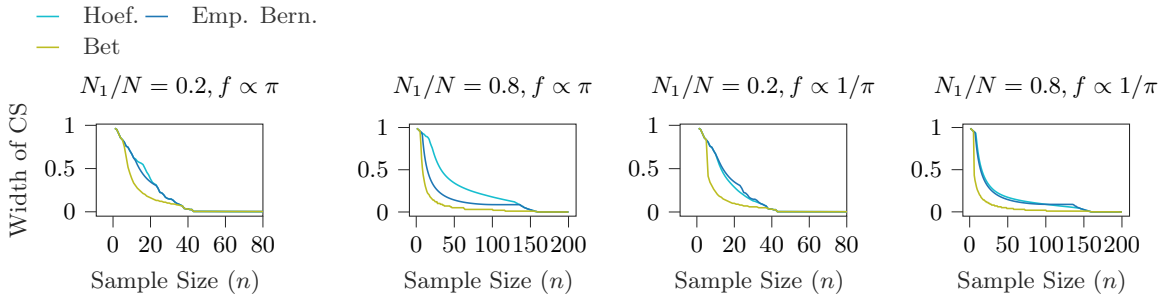

            \def\figwidth{0.23\textwidth}
            \def\figheight{0.17\textwidth} %
        \hspace*{-2pt}
        \input{Figures/hoef_eb/exp1_prop_2_cs.tex}
        \input{Figures/hoef_eb/exp1_prop_8_cs.tex}
        \input{Figures/hoef_eb/exp1_inv_2_cs.tex}
        \input{Figures/hoef_eb/exp1_inv_8_cs.tex}
        \caption{Plots showing the variation of the width of the betting, Hoeffding, and empirical-Bernstein CSs using \propM sampling strategies in different data regimes with $N=200$; all CSs here also are intersected with the logical CS of~\Cref{subsec:logical-CS}. We can see that empirical-Bernstein is tighter than Hoeffding in cases where the proportion of transactions with large weights is high (i.e., $\Nlarge / N$ is large), and vice versa. Across the board, the betting CS is tighter than both Hoeffding and empirical-Bernstein CSs.}
        \label{fig:HoefEBCS}
    \end{figure}
    \begin{figure}[htb!]
            \def\figwidth{0.23\textwidth}
            \def\figheight{0.17\textwidth} %
        \begin{tikzpicture}

\definecolor{darkslategray38}{RGB}{38,38,38}
\definecolor{darkturquoise23190207}{RGB}{23,190,207}
\definecolor{goldenrod18818934}{RGB}{188,189,34}
\definecolor{lightgray204}{RGB}{204,204,204}
\definecolor{steelblue31119180}{RGB}{31,119,180}

\begin{axis}[
axis line style={darkslategray38},
height=\figheight,
legend cell align={left},
legend style={fill opacity=0.8, draw opacity=1, text opacity=1, draw=none, legend columns=2, at={(0.5,2.4)}, anchor=north},
tick align=outside,
tick pos=left,
title={\(\displaystyle N_1 / N = 0.2, f \propto \pi\)},
width=\figwidth,
x grid style={lightgray204},
xlabel=\textcolor{darkslategray38}{Stopping Times},
xmin=16.15, xmax=34.85,
xtick style={color=darkslategray38},
y grid style={lightgray204},
ylabel=\textcolor{darkslategray38}{Density},
ymin=0, ymax=1,
ytick style={color=darkslategray38}
]
\draw[draw=none,fill=darkturquoise23190207,fill opacity=0.7] (axis cs:25,0) rectangle (axis cs:25.8,0.005);
\addlegendimage{ybar,ybar legend,draw=none,fill=darkturquoise23190207,fill opacity=0.7}
\addlegendentry{Hoef.}

\draw[draw=none,fill=darkturquoise23190207,fill opacity=0.7] (axis cs:25.8,0) rectangle (axis cs:26.6,0.02);
\draw[draw=none,fill=darkturquoise23190207,fill opacity=0.7] (axis cs:26.6,0) rectangle (axis cs:27.4,0.0875000000000003);
\draw[draw=none,fill=darkturquoise23190207,fill opacity=0.7] (axis cs:27.4,0) rectangle (axis cs:28.2,0.21);
\draw[draw=none,fill=darkturquoise23190207,fill opacity=0.7] (axis cs:28.2,0) rectangle (axis cs:29,0);
\draw[draw=none,fill=darkturquoise23190207,fill opacity=0.7] (axis cs:29,0) rectangle (axis cs:29.8,0.36);
\draw[draw=none,fill=darkturquoise23190207,fill opacity=0.7] (axis cs:29.8,0) rectangle (axis cs:30.6,0.3075);
\draw[draw=none,fill=darkturquoise23190207,fill opacity=0.7] (axis cs:30.6,0) rectangle (axis cs:31.4,0.152500000000001);
\draw[draw=none,fill=darkturquoise23190207,fill opacity=0.7] (axis cs:31.4,0) rectangle (axis cs:32.2,0.0874999999999995);
\draw[draw=none,fill=darkturquoise23190207,fill opacity=0.7] (axis cs:32.2,0) rectangle (axis cs:33,0.0200000000000001);
\draw[draw=none,fill=steelblue31119180,fill opacity=0.7] (axis cs:25,0) rectangle (axis cs:25.9,0.00666666666666668);
\addlegendimage{ybar,ybar legend,draw=none,fill=steelblue31119180,fill opacity=0.7}
\addlegendentry{Emp. Bern.}

\draw[draw=none,fill=steelblue31119180,fill opacity=0.7] (axis cs:25.9,0) rectangle (axis cs:26.8,0.0266666666666666);
\draw[draw=none,fill=steelblue31119180,fill opacity=0.7] (axis cs:26.8,0) rectangle (axis cs:27.7,0.0822222222222224);
\draw[draw=none,fill=steelblue31119180,fill opacity=0.7] (axis cs:27.7,0) rectangle (axis cs:28.6,0.204444444444444);
\draw[draw=none,fill=steelblue31119180,fill opacity=0.7] (axis cs:28.6,0) rectangle (axis cs:29.5,0.255555555555556);
\draw[draw=none,fill=steelblue31119180,fill opacity=0.7] (axis cs:29.5,0) rectangle (axis cs:30.4,0.277777777777778);
\draw[draw=none,fill=steelblue31119180,fill opacity=0.7] (axis cs:30.4,0) rectangle (axis cs:31.3,0.186666666666666);
\draw[draw=none,fill=steelblue31119180,fill opacity=0.7] (axis cs:31.3,0) rectangle (axis cs:32.2,0.0533333333333332);
\draw[draw=none,fill=steelblue31119180,fill opacity=0.7] (axis cs:32.2,0) rectangle (axis cs:33.1,0.0133333333333334);
\draw[draw=none,fill=steelblue31119180,fill opacity=0.7] (axis cs:33.1,0) rectangle (axis cs:34,0.00444444444444445);
\draw[draw=none,fill=goldenrod18818934,fill opacity=0.7] (axis cs:17,0) rectangle (axis cs:17.6,0.833333333333331);
\addlegendimage{ybar,ybar legend,draw=none,fill=goldenrod18818934,fill opacity=0.7}
\addlegendentry{Bet}

\draw[draw=none,fill=goldenrod18818934,fill opacity=0.7] (axis cs:17.6,0) rectangle (axis cs:18.2,0.793333333333336);
\draw[draw=none,fill=goldenrod18818934,fill opacity=0.7] (axis cs:18.2,0) rectangle (axis cs:18.8,0);
\draw[draw=none,fill=goldenrod18818934,fill opacity=0.7] (axis cs:18.8,0) rectangle (axis cs:19.4,0.0266666666666668);
\draw[draw=none,fill=goldenrod18818934,fill opacity=0.7] (axis cs:19.4,0) rectangle (axis cs:20,0);
\draw[draw=none,fill=goldenrod18818934,fill opacity=0.7] (axis cs:20,0) rectangle (axis cs:20.6,0.00333333333333333);
\draw[draw=none,fill=goldenrod18818934,fill opacity=0.7] (axis cs:20.6,0) rectangle (axis cs:21.2,0.00666666666666669);
\draw[draw=none,fill=goldenrod18818934,fill opacity=0.7] (axis cs:21.2,0) rectangle (axis cs:21.8,0);
\draw[draw=none,fill=goldenrod18818934,fill opacity=0.7] (axis cs:21.8,0) rectangle (axis cs:22.4,0);
\draw[draw=none,fill=goldenrod18818934,fill opacity=0.7] (axis cs:22.4,0) rectangle (axis cs:23,0.00333333333333333);
\end{axis}

\end{tikzpicture}
        \begin{tikzpicture}

\definecolor{darkslategray38}{RGB}{38,38,38}
\definecolor{darkturquoise23190207}{RGB}{23,190,207}
\definecolor{goldenrod18818934}{RGB}{188,189,34}
\definecolor{lightgray204}{RGB}{204,204,204}
\definecolor{steelblue31119180}{RGB}{31,119,180}

\begin{axis}[
axis line style={darkslategray38},
height=\figheight,
tick align=outside,
tick pos=left,
title={\(\displaystyle N_1 / N = 0.8, f \propto \pi\)},
width=\figwidth,
x grid style={lightgray204},
xlabel=\textcolor{darkslategray38}{Stopping Times},
xmin=13, xmax=101,
xtick style={color=darkslategray38},
y grid style={lightgray204},
ymin=0, ymax=1,
ytick style={color=darkslategray38}
]
\draw[draw=none,fill=darkturquoise23190207,fill opacity=0.7] (axis cs:85,0) rectangle (axis cs:86.2,0.00833333333333331);
\addlegendimage{ybar,ybar legend,draw=none,fill=darkturquoise23190207,fill opacity=0.7}

\draw[draw=none,fill=darkturquoise23190207,fill opacity=0.7] (axis cs:86.2,0) rectangle (axis cs:87.4,0.00833333333333331);
\draw[draw=none,fill=darkturquoise23190207,fill opacity=0.7] (axis cs:87.4,0) rectangle (axis cs:88.6,0.0400000000000004);
\draw[draw=none,fill=darkturquoise23190207,fill opacity=0.7] (axis cs:88.6,0) rectangle (axis cs:89.8,0.0699999999999998);
\draw[draw=none,fill=darkturquoise23190207,fill opacity=0.7] (axis cs:89.8,0) rectangle (axis cs:91,0.111666666666666);
\draw[draw=none,fill=darkturquoise23190207,fill opacity=0.7] (axis cs:91,0) rectangle (axis cs:92.2,0.346666666666666);
\draw[draw=none,fill=darkturquoise23190207,fill opacity=0.7] (axis cs:92.2,0) rectangle (axis cs:93.4,0.113333333333333);
\draw[draw=none,fill=darkturquoise23190207,fill opacity=0.7] (axis cs:93.4,0) rectangle (axis cs:94.6,0.0850000000000008);
\draw[draw=none,fill=darkturquoise23190207,fill opacity=0.7] (axis cs:94.6,0) rectangle (axis cs:95.8,0.0383333333333332);
\draw[draw=none,fill=darkturquoise23190207,fill opacity=0.7] (axis cs:95.8,0) rectangle (axis cs:97,0.0116666666666666);
\draw[draw=none,fill=steelblue31119180,fill opacity=1] (axis cs:35,0) rectangle (axis cs:35.2,0.0199999999999997);
\addlegendimage{ybar,ybar legend,draw=none,fill=steelblue31119180,fill opacity=0.7}

\draw[draw=none,fill=steelblue31119180,fill opacity=0.7] (axis cs:35.2,0) rectangle (axis cs:35.4,0);
\draw[draw=none,fill=steelblue31119180,fill opacity=0.7] (axis cs:35.4,0) rectangle (axis cs:35.6,0);
\draw[draw=none,fill=steelblue31119180,fill opacity=0.7] (axis cs:35.6,0) rectangle (axis cs:35.8,0);
\draw[draw=none,fill=steelblue31119180,fill opacity=0.7] (axis cs:35.8,0) rectangle (axis cs:36,0);
\draw[draw=none,fill=steelblue31119180,fill opacity=1] (axis cs:36,0) rectangle (axis cs:36.4,4.95999999999993);
\draw[draw=none,fill=steelblue31119180,fill opacity=0.7] (axis cs:36.2,0) rectangle (axis cs:36.4,0);
\draw[draw=none,fill=steelblue31119180,fill opacity=0.7] (axis cs:36.4,0) rectangle (axis cs:36.6,0);
\draw[draw=none,fill=steelblue31119180,fill opacity=0.7] (axis cs:36.6,0) rectangle (axis cs:36.8,0);
\draw[draw=none,fill=steelblue31119180,fill opacity=1] (axis cs:36.8,0) rectangle (axis cs:37.2,0.0199999999999997);
\draw[draw=none,fill=goldenrod18818934,fill opacity=1] (axis cs:17,0) rectangle (axis cs:17.4,4.99999999999993);
\addlegendimage{ybar,ybar legend,draw=none,fill=goldenrod18818934,fill opacity=0.7}

\draw[draw=none,fill=goldenrod18818934,fill opacity=0.7] (axis cs:17.1,0) rectangle (axis cs:17.2,0);
\draw[draw=none,fill=goldenrod18818934,fill opacity=0.7] (axis cs:17.2,0) rectangle (axis cs:17.3,0);
\draw[draw=none,fill=goldenrod18818934,fill opacity=0.7] (axis cs:17.3,0) rectangle (axis cs:17.4,0);
\draw[draw=none,fill=goldenrod18818934,fill opacity=0.7] (axis cs:17.4,0) rectangle (axis cs:17.5,0);
\draw[draw=none,fill=goldenrod18818934,fill opacity=0.7] (axis cs:17.5,0) rectangle (axis cs:17.6,0);
\draw[draw=none,fill=goldenrod18818934,fill opacity=0.7] (axis cs:17.6,0) rectangle (axis cs:17.7,0);
\draw[draw=none,fill=goldenrod18818934,fill opacity=0.7] (axis cs:17.7,0) rectangle (axis cs:17.8,0);
\draw[draw=none,fill=goldenrod18818934,fill opacity=0.7] (axis cs:17.8,0) rectangle (axis cs:17.9,0);
\draw[draw=none,fill=goldenrod18818934,fill opacity=1] (axis cs:17.9,0) rectangle (axis cs:18.3,4.99999999999993);
\end{axis}

\end{tikzpicture}
        \begin{tikzpicture}

\definecolor{darkslategray38}{RGB}{38,38,38}
\definecolor{darkturquoise23190207}{RGB}{23,190,207}
\definecolor{goldenrod18818934}{RGB}{188,189,34}
\definecolor{lightgray204}{RGB}{204,204,204}
\definecolor{steelblue31119180}{RGB}{31,119,180}

\begin{axis}[
axis line style={darkslategray38},
height=\figheight,
tick align=outside,
tick pos=left,
title={\(\displaystyle N_1 / N = 0.2, f \propto 1/\pi\)},
width=\figwidth,
x grid style={lightgray204},
xlabel=\textcolor{darkslategray38}{Stopping Times},
xmin=9.85, xmax=35.15,
xtick style={color=darkslategray38},
y grid style={lightgray204},
ymin=0, ymax=1,
ytick style={color=darkslategray38}
]
\draw[draw=none,fill=darkturquoise23190207,fill opacity=0.7] (axis cs:23,0) rectangle (axis cs:23.6,0.00333333333333333);
\addlegendimage{ybar,ybar legend,draw=none,fill=darkturquoise23190207,fill opacity=0.7}

\draw[draw=none,fill=darkturquoise23190207,fill opacity=0.7] (axis cs:23.6,0) rectangle (axis cs:24.2,0.0466666666666668);
\draw[draw=none,fill=darkturquoise23190207,fill opacity=0.7] (axis cs:24.2,0) rectangle (axis cs:24.8,0);
\draw[draw=none,fill=darkturquoise23190207,fill opacity=0.7] (axis cs:24.8,0) rectangle (axis cs:25.4,0.350000000000001);
\draw[draw=none,fill=darkturquoise23190207,fill opacity=0.7] (axis cs:25.4,0) rectangle (axis cs:26,0);
\draw[draw=none,fill=darkturquoise23190207,fill opacity=0.7] (axis cs:26,0) rectangle (axis cs:26.6,0.659999999999998);
\draw[draw=none,fill=darkturquoise23190207,fill opacity=0.7] (axis cs:26.6,0) rectangle (axis cs:27.2,0.443333333333335);
\draw[draw=none,fill=darkturquoise23190207,fill opacity=0.7] (axis cs:27.2,0) rectangle (axis cs:27.8,0);
\draw[draw=none,fill=darkturquoise23190207,fill opacity=0.7] (axis cs:27.8,0) rectangle (axis cs:28.4,0.143333333333334);
\draw[draw=none,fill=darkturquoise23190207,fill opacity=0.7] (axis cs:28.4,0) rectangle (axis cs:29,0.02);
\draw[draw=none,fill=steelblue31119180,fill opacity=0.7] (axis cs:25,0) rectangle (axis cs:25.9,0.00666666666666668);
\addlegendimage{ybar,ybar legend,draw=none,fill=steelblue31119180,fill opacity=0.7}

\draw[draw=none,fill=steelblue31119180,fill opacity=0.7] (axis cs:25.9,0) rectangle (axis cs:26.8,0.0266666666666666);
\draw[draw=none,fill=steelblue31119180,fill opacity=0.7] (axis cs:26.8,0) rectangle (axis cs:27.7,0.0822222222222224);
\draw[draw=none,fill=steelblue31119180,fill opacity=0.7] (axis cs:27.7,0) rectangle (axis cs:28.6,0.204444444444444);
\draw[draw=none,fill=steelblue31119180,fill opacity=0.7] (axis cs:28.6,0) rectangle (axis cs:29.5,0.255555555555556);
\draw[draw=none,fill=steelblue31119180,fill opacity=0.7] (axis cs:29.5,0) rectangle (axis cs:30.4,0.277777777777778);
\draw[draw=none,fill=steelblue31119180,fill opacity=0.7] (axis cs:30.4,0) rectangle (axis cs:31.3,0.186666666666666);
\draw[draw=none,fill=steelblue31119180,fill opacity=0.7] (axis cs:31.3,0) rectangle (axis cs:32.2,0.0533333333333332);
\draw[draw=none,fill=steelblue31119180,fill opacity=0.7] (axis cs:32.2,0) rectangle (axis cs:33.1,0.0133333333333334);
\draw[draw=none,fill=steelblue31119180,fill opacity=0.7] (axis cs:33.1,0) rectangle (axis cs:34,0.00444444444444445);
\draw[draw=none,fill=goldenrod18818934,fill opacity=0.7] (axis cs:11,0) rectangle (axis cs:11.4,0.005);
\addlegendimage{ybar,ybar legend,draw=none,fill=goldenrod18818934,fill opacity=0.7}

\draw[draw=none,fill=goldenrod18818934,fill opacity=0.7] (axis cs:11.4,0) rectangle (axis cs:11.8,0);
\draw[draw=none,fill=goldenrod18818934,fill opacity=0.7] (axis cs:11.8,0) rectangle (axis cs:12.2,2.18500000000001);
\draw[draw=none,fill=goldenrod18818934,fill opacity=0.7] (axis cs:12.2,0) rectangle (axis cs:12.6,0);
\draw[draw=none,fill=goldenrod18818934,fill opacity=0.7] (axis cs:12.6,0) rectangle (axis cs:13,0);
\draw[draw=none,fill=goldenrod18818934,fill opacity=0.7] (axis cs:13,0) rectangle (axis cs:13.4,0.165);
\draw[draw=none,fill=goldenrod18818934,fill opacity=0.7] (axis cs:13.4,0) rectangle (axis cs:13.8,0);
\draw[draw=none,fill=goldenrod18818934,fill opacity=0.7] (axis cs:13.8,0) rectangle (axis cs:14.2,0.135);
\draw[draw=none,fill=goldenrod18818934,fill opacity=0.7] (axis cs:14.2,0) rectangle (axis cs:14.6,0);
\draw[draw=none,fill=goldenrod18818934,fill opacity=0.7] (axis cs:14.6,0) rectangle (axis cs:15,0.00999999999999999);
\end{axis}

\end{tikzpicture}
        \begin{tikzpicture}

\definecolor{darkslategray38}{RGB}{38,38,38}
\definecolor{darkturquoise23190207}{RGB}{23,190,207}
\definecolor{goldenrod18818934}{RGB}{188,189,34}
\definecolor{lightgray204}{RGB}{204,204,204}
\definecolor{steelblue31119180}{RGB}{31,119,180}

\begin{axis}[
axis line style={darkslategray38},
height=\figheight,
tick align=outside,
tick pos=left,
title={\(\displaystyle N_1 / N = 0.8, f \propto 1/\pi\)},
width=\figwidth,
x grid style={lightgray204},
xlabel=\textcolor{darkslategray38}{Stopping Times},
xmin=11.875, xmax=47.625,
xtick style={color=darkslategray38},
y grid style={lightgray204},
ymin=0, ymax=1,
ytick style={color=darkslategray38}
]
\draw[draw=none,fill=darkturquoise23190207,fill opacity=0.7] (axis cs:41,0) rectangle (axis cs:41.5,0.02);
\addlegendimage{ybar,ybar legend,draw=none,fill=darkturquoise23190207,fill opacity=0.7}

\draw[draw=none,fill=darkturquoise23190207,fill opacity=0.7] (axis cs:41.5,0) rectangle (axis cs:42,0);
\draw[draw=none,fill=darkturquoise23190207,fill opacity=0.7] (axis cs:42,0) rectangle (axis cs:42.5,0.308);
\draw[draw=none,fill=darkturquoise23190207,fill opacity=0.7] (axis cs:42.5,0) rectangle (axis cs:43,0);
\draw[draw=none,fill=darkturquoise23190207,fill opacity=0.7] (axis cs:43,0) rectangle (axis cs:43.5,1.164);
\draw[draw=none,fill=darkturquoise23190207,fill opacity=0.7] (axis cs:43.5,0) rectangle (axis cs:44,0);
\draw[draw=none,fill=darkturquoise23190207,fill opacity=0.7] (axis cs:44,0) rectangle (axis cs:44.5,0.468);
\draw[draw=none,fill=darkturquoise23190207,fill opacity=0.7] (axis cs:44.5,0) rectangle (axis cs:45,0);
\draw[draw=none,fill=darkturquoise23190207,fill opacity=0.7] (axis cs:45,0) rectangle (axis cs:45.5,0.028);
\draw[draw=none,fill=darkturquoise23190207,fill opacity=0.7] (axis cs:45.5,0) rectangle (axis cs:46,0.012);
\draw[draw=none,fill=steelblue31119180,fill opacity=0.7] (axis cs:36,0) rectangle (axis cs:36.3,0.0266666666666669);
\addlegendimage{ybar,ybar legend,draw=none,fill=steelblue31119180,fill opacity=0.7}

\draw[draw=none,fill=steelblue31119180,fill opacity=0.7] (axis cs:36.3,0) rectangle (axis cs:36.6,0);
\draw[draw=none,fill=steelblue31119180,fill opacity=0.7] (axis cs:36.6,0) rectangle (axis cs:36.9,0);
\draw[draw=none,fill=steelblue31119180,fill opacity=0.7] (axis cs:36.9,0) rectangle (axis cs:37.2,3.26666666666662);
\draw[draw=none,fill=steelblue31119180,fill opacity=0.7] (axis cs:37.2,0) rectangle (axis cs:37.5,0);
\draw[draw=none,fill=steelblue31119180,fill opacity=0.7] (axis cs:37.5,0) rectangle (axis cs:37.8,0);
\draw[draw=none,fill=steelblue31119180,fill opacity=0.7] (axis cs:37.8,0) rectangle (axis cs:38.1,0.00666666666666657);
\draw[draw=none,fill=steelblue31119180,fill opacity=0.7] (axis cs:38.1,0) rectangle (axis cs:38.4,0);
\draw[draw=none,fill=steelblue31119180,fill opacity=0.7] (axis cs:38.4,0) rectangle (axis cs:38.7,0);
\draw[draw=none,fill=steelblue31119180,fill opacity=0.7] (axis cs:38.7,0) rectangle (axis cs:39,0.0333333333333336);
\draw[draw=none,fill=goldenrod18818934,fill opacity=0.7] (axis cs:13.5,0) rectangle (axis cs:13.6,0);
\addlegendimage{ybar,ybar legend,draw=none,fill=goldenrod18818934,fill opacity=0.7}

\draw[draw=none,fill=goldenrod18818934,fill opacity=0.7] (axis cs:13.6,0) rectangle (axis cs:13.7,0);
\draw[draw=none,fill=goldenrod18818934,fill opacity=0.7] (axis cs:13.7,0) rectangle (axis cs:13.8,0);
\draw[draw=none,fill=goldenrod18818934,fill opacity=0.7] (axis cs:13.8,0) rectangle (axis cs:13.9,0);
\draw[draw=none,fill=goldenrod18818934,fill opacity=0.7] (axis cs:13.9,0) rectangle (axis cs:14,0);
\draw[draw=none,fill=goldenrod18818934,fill opacity=1] (axis cs:14,0) rectangle (axis cs:14.4,10);
\draw[draw=none,fill=goldenrod18818934,fill opacity=0.7] (axis cs:14.1,0) rectangle (axis cs:14.2,0);
\draw[draw=none,fill=goldenrod18818934,fill opacity=0.7] (axis cs:14.2,0) rectangle (axis cs:14.3,0);
\draw[draw=none,fill=goldenrod18818934,fill opacity=0.7] (axis cs:14.3,0) rectangle (axis cs:14.4,0);
\draw[draw=none,fill=goldenrod18818934,fill opacity=0.7] (axis cs:14.4,0) rectangle (axis cs:14.5,0);
\end{axis}

\end{tikzpicture}
        \caption{Histograms of the first time for each CS to reach a width of $\varepsilon=0.2$ with $N=200$. We choose a larger $\varepsilon$ than in Experiment 1 for purposes of demonstrating the difference between the CSs, as all CSs converge to the logical CS and have nearly identical stopping time distribution for small values of $\varepsilon$. We can see that the empirical-Bernstein CS is reaches $\varepsilon$ width earlier than Hoeffding CS when $\Nlarge$ is large and the reverse is true when $\Nlarge$ is small, and the betting CS reaches $\varepsilon$ uniformly the fastest.}
        \label{fig:HoefEBHist}
    \end{figure}
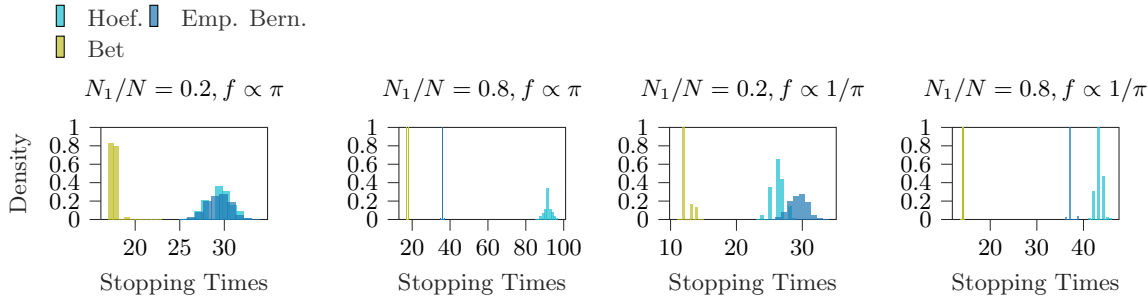

\section{Experiments with Housing Sales Data}
\label{appendix:housing-data} 

    We now apply our auditing scheme to the transactions in the `\href{https://www.kaggle.com/datasets/harlfoxem/housesalesprediction}{House Sales in King County}' dataset from Kaggle. The dataset consists of $21,616$ datapoints, each consisting of $21$ features describing the house (such as the number of bedrooms, the number of bathrooms, square footage, floors, condition, etc), and one target variable \texttt{price}.  We treat the \texttt{price} values as the `reported monetary values' for our framework. 
    
    \paragraph{Creating the ground truth.}  To adapt the dataset to our problem, we first need to generate the `ground truth', that is, the true $f$-values. To do this, we proceed in the following steps: 
    \begin{itemize}
        \item We first select $10\%$ of the dataset, and assign them some arbitrary $f$ values in the range $(0,0.7)$.
        
        \item Using this `labelled' dataset, we train a random forest regressor with $200$ trees, and mean-absolute error criterion.
        
        \item Finally, this trained regressor is then used to generate the ground truth for the remaining $90\%$ of the dataset.
    \end{itemize}
     The reason for using this approach for generating the ground truth, is that we want it to be dependent on the additional features associated with each transaction.
    
    \paragraph{Generating the side-information.} Having obtained the $M$ and $f$-values, we obtain the side-information~(i.e., $S$-values) by using $10\%$ of the remaining labelled data to train another predictor for generating side-information on the rest of the data. In our experiments, we either used a single decision-tree regressor, or a random forest with a small number of trees~(fewer than $50$). Informally, we expect that increasing the capacity of the regressor should lead to increased correlation between the ground truth and the side-information. 
    \begin{figure}[htb!]
            \def\figwidth{0.45\textwidth}
            \def\figheight{0.45\textwidth} %
        \hspace*{-2pt}
        \begin{tikzpicture}

\definecolor{crimson2143940}{RGB}{214,39,40}
\definecolor{darkslategray38}{RGB}{38,38,38}
\definecolor{lightgray204}{RGB}{204,204,204}
\definecolor{orchid227119194}{RGB}{227,119,194}
\definecolor{steelblue31119180}{RGB}{31,119,180}

\begin{axis}[
axis line style={darkslategray38},
height=\figheight,
legend cell align={left},
legend style={
  fill opacity=0.8,
  draw opacity=1,
  text opacity=1,
  at={(0.5,0.09)},
  anchor=south,
  draw=none
},
tick align=outside,
tick pos=left,
title={Stopping Times distribution ($\rho\approx 0.71$)},
width=\figwidth,
x grid style={lightgray204},
xlabel=\textcolor{darkslategray38}{Stopping Times},
xmin=70.95, xmax=248.05,
xtick style={color=darkslategray38},
y grid style={lightgray204},
ylabel=\textcolor{darkslategray38}{Density},
ymin=0, ymax=0.0750000000000001,
ytick style={color=darkslategray38}
]
\draw[draw=none,fill=steelblue31119180,fill opacity=0.4] (axis cs:79,0) rectangle (axis cs:84.2,0.00384615384615384);
\addlegendimage{ybar,ybar legend,draw=none,fill=steelblue31119180,fill opacity=0.4}
\addlegendentry{propM}

\draw[draw=none,fill=steelblue31119180,fill opacity=0.4] (axis cs:84.2,0) rectangle (axis cs:89.4,0.00384615384615384);
\draw[draw=none,fill=steelblue31119180,fill opacity=0.4] (axis cs:89.4,0) rectangle (axis cs:94.6,0.00769230769230771);
\draw[draw=none,fill=steelblue31119180,fill opacity=0.4] (axis cs:94.6,0) rectangle (axis cs:99.8,0.00769230769230769);
\draw[draw=none,fill=steelblue31119180,fill opacity=0.4] (axis cs:99.8,0) rectangle (axis cs:105,0.0230769230769231);
\draw[draw=none,fill=steelblue31119180,fill opacity=0.4] (axis cs:105,0) rectangle (axis cs:110.2,0.0384615384615384);
\draw[draw=none,fill=steelblue31119180,fill opacity=0.4] (axis cs:110.2,0) rectangle (axis cs:115.4,0.0384615384615384);
\draw[draw=none,fill=steelblue31119180,fill opacity=0.4] (axis cs:115.4,0) rectangle (axis cs:120.6,0.0307692307692308);
\draw[draw=none,fill=steelblue31119180,fill opacity=0.4] (axis cs:120.6,0) rectangle (axis cs:125.8,0.0192307692307692);
\draw[draw=none,fill=steelblue31119180,fill opacity=0.4] (axis cs:125.8,0) rectangle (axis cs:131,0.0192307692307693);
\draw[draw=none,fill=crimson2143940,fill opacity=0.4] (axis cs:83,0) rectangle (axis cs:86.3,0.0424242424242425);
\addlegendimage{ybar,ybar legend,draw=none,fill=crimson2143940,fill opacity=0.4}
\addlegendentry{propM+CV}

\draw[draw=none,fill=crimson2143940,fill opacity=0.4] (axis cs:86.3,0) rectangle (axis cs:89.6,0.0484848484848485);
\draw[draw=none,fill=crimson2143940,fill opacity=0.4] (axis cs:89.6,0) rectangle (axis cs:92.9,0.0484848484848483);
\draw[draw=none,fill=crimson2143940,fill opacity=0.4] (axis cs:92.9,0) rectangle (axis cs:96.2,0.0606060606060607);
\draw[draw=none,fill=crimson2143940,fill opacity=0.4] (axis cs:96.2,0) rectangle (axis cs:99.5,0.0363636363636364);
\draw[draw=none,fill=crimson2143940,fill opacity=0.4] (axis cs:99.5,0) rectangle (axis cs:102.8,0.0363636363636364);
\draw[draw=none,fill=crimson2143940,fill opacity=0.4] (axis cs:102.8,0) rectangle (axis cs:106.1,0.0121212121212121);
\draw[draw=none,fill=crimson2143940,fill opacity=0.4] (axis cs:106.1,0) rectangle (axis cs:109.4,0.0121212121212121);
\draw[draw=none,fill=crimson2143940,fill opacity=0.4] (axis cs:109.4,0) rectangle (axis cs:112.7,0);
\draw[draw=none,fill=crimson2143940,fill opacity=0.4] (axis cs:112.7,0) rectangle (axis cs:116,0.00606060606060607);
\draw[draw=none,fill=orchid227119194,fill opacity=0.4] (axis cs:184,0) rectangle (axis cs:189.6,0.00357142857142857);
\addlegendimage{ybar,ybar legend,draw=none,fill=orchid227119194,fill opacity=0.4}
\addlegendentry{uniform}

\draw[draw=none,fill=orchid227119194,fill opacity=0.4] (axis cs:189.6,0) rectangle (axis cs:195.2,0);
\draw[draw=none,fill=orchid227119194,fill opacity=0.4] (axis cs:195.2,0) rectangle (axis cs:200.8,0);
\draw[draw=none,fill=orchid227119194,fill opacity=0.4] (axis cs:200.8,0) rectangle (axis cs:206.4,0);
\draw[draw=none,fill=orchid227119194,fill opacity=0.4] (axis cs:206.4,0) rectangle (axis cs:212,0.00357142857142857);
\draw[draw=none,fill=orchid227119194,fill opacity=0.4] (axis cs:212,0) rectangle (axis cs:217.6,0.00357142857142857);
\draw[draw=none,fill=orchid227119194,fill opacity=0.4] (axis cs:217.6,0) rectangle (axis cs:223.2,0.0107142857142857);
\draw[draw=none,fill=orchid227119194,fill opacity=0.4] (axis cs:223.2,0) rectangle (axis cs:228.8,0.0357142857142856);
\draw[draw=none,fill=orchid227119194,fill opacity=0.4] (axis cs:228.8,0) rectangle (axis cs:234.4,0.0714285714285715);
\draw[draw=none,fill=orchid227119194,fill opacity=0.4] (axis cs:234.4,0) rectangle (axis cs:240,0.0500000000000001);
\end{axis}

\end{tikzpicture}
        \begin{tikzpicture}

\definecolor{crimson2143940}{RGB}{214,39,40}
\definecolor{darkslategray38}{RGB}{38,38,38}
\definecolor{lightgray204}{RGB}{204,204,204}
\definecolor{orchid227119194}{RGB}{227,119,194}
\definecolor{steelblue31119180}{RGB}{31,119,180}

\begin{axis}[
axis line style={darkslategray38},
height=\figheight,
legend cell align={left},
legend style={
  fill opacity=0.8,
  draw opacity=1,
  text opacity=1,
  at={(0.5,0.09)},
  anchor=south,
  draw=none
},
tick align=outside,
tick pos=left,
title={Stopping Times distribution ($\rho\approx 0.82$)},
width=\figwidth,
x grid style={lightgray204},
xlabel=\textcolor{darkslategray38}{Stopping Times},
xmin=67.95, xmax=245.05,
xtick style={color=darkslategray38},
y grid style={lightgray204},
ylabel=\textcolor{darkslategray38}{Density},
ymin=0, ymax=0.0855555555555555,
ytick style={color=darkslategray38}
]
\draw[draw=none,fill=steelblue31119180,fill opacity=0.4] (axis cs:88,0) rectangle (axis cs:92,0.02);
\addlegendimage{ybar,ybar legend,draw=none,fill=steelblue31119180,fill opacity=0.4}
\addlegendentry{propM}

\draw[draw=none,fill=steelblue31119180,fill opacity=0.4] (axis cs:92,0) rectangle (axis cs:96,0.015);
\draw[draw=none,fill=steelblue31119180,fill opacity=0.4] (axis cs:96,0) rectangle (axis cs:100,0.03);
\draw[draw=none,fill=steelblue31119180,fill opacity=0.4] (axis cs:100,0) rectangle (axis cs:104,0.03);
\draw[draw=none,fill=steelblue31119180,fill opacity=0.4] (axis cs:104,0) rectangle (axis cs:108,0.035);
\draw[draw=none,fill=steelblue31119180,fill opacity=0.4] (axis cs:108,0) rectangle (axis cs:112,0.045);
\draw[draw=none,fill=steelblue31119180,fill opacity=0.4] (axis cs:112,0) rectangle (axis cs:116,0.02);
\draw[draw=none,fill=steelblue31119180,fill opacity=0.4] (axis cs:116,0) rectangle (axis cs:120,0.01);
\draw[draw=none,fill=steelblue31119180,fill opacity=0.4] (axis cs:120,0) rectangle (axis cs:124,0.02);
\draw[draw=none,fill=steelblue31119180,fill opacity=0.4] (axis cs:124,0) rectangle (axis cs:128,0.025);
\draw[draw=none,fill=crimson2143940,fill opacity=0.4] (axis cs:76,0) rectangle (axis cs:78.7,0.0148148148148148);
\addlegendimage{ybar,ybar legend,draw=none,fill=crimson2143940,fill opacity=0.4}
\addlegendentry{propM+CV}

\draw[draw=none,fill=crimson2143940,fill opacity=0.4] (axis cs:78.7,0) rectangle (axis cs:81.4,0.0592592592592592);
\draw[draw=none,fill=crimson2143940,fill opacity=0.4] (axis cs:81.4,0) rectangle (axis cs:84.1,0.0222222222222223);
\draw[draw=none,fill=crimson2143940,fill opacity=0.4] (axis cs:84.1,0) rectangle (axis cs:86.8,0.0518518518518518);
\draw[draw=none,fill=crimson2143940,fill opacity=0.4] (axis cs:86.8,0) rectangle (axis cs:89.5,0.0814814814814814);
\draw[draw=none,fill=crimson2143940,fill opacity=0.4] (axis cs:89.5,0) rectangle (axis cs:92.2,0.0444444444444444);
\draw[draw=none,fill=crimson2143940,fill opacity=0.4] (axis cs:92.2,0) rectangle (axis cs:94.9,0.0518518518518518);
\draw[draw=none,fill=crimson2143940,fill opacity=0.4] (axis cs:94.9,0) rectangle (axis cs:97.6,0.0296296296296298);
\draw[draw=none,fill=crimson2143940,fill opacity=0.4] (axis cs:97.6,0) rectangle (axis cs:100.3,0.0074074074074074);
\draw[draw=none,fill=crimson2143940,fill opacity=0.4] (axis cs:100.3,0) rectangle (axis cs:103,0.0074074074074074);
\draw[draw=none,fill=orchid227119194,fill opacity=0.4] (axis cs:186,0) rectangle (axis cs:191.1,0.00392156862745098);
\addlegendimage{ybar,ybar legend,draw=none,fill=orchid227119194,fill opacity=0.4}
\addlegendentry{uniform}

\draw[draw=none,fill=orchid227119194,fill opacity=0.4] (axis cs:191.1,0) rectangle (axis cs:196.2,0);
\draw[draw=none,fill=orchid227119194,fill opacity=0.4] (axis cs:196.2,0) rectangle (axis cs:201.3,0);
\draw[draw=none,fill=orchid227119194,fill opacity=0.4] (axis cs:201.3,0) rectangle (axis cs:206.4,0.00784313725490197);
\draw[draw=none,fill=orchid227119194,fill opacity=0.4] (axis cs:206.4,0) rectangle (axis cs:211.5,0.00784313725490197);
\draw[draw=none,fill=orchid227119194,fill opacity=0.4] (axis cs:211.5,0) rectangle (axis cs:216.6,0.00784313725490197);
\draw[draw=none,fill=orchid227119194,fill opacity=0.4] (axis cs:216.6,0) rectangle (axis cs:221.7,0.011764705882353);
\draw[draw=none,fill=orchid227119194,fill opacity=0.4] (axis cs:221.7,0) rectangle (axis cs:226.8,0.0352941176470587);
\draw[draw=none,fill=orchid227119194,fill opacity=0.4] (axis cs:226.8,0) rectangle (axis cs:231.9,0.0470588235294118);
\draw[draw=none,fill=orchid227119194,fill opacity=0.4] (axis cs:231.9,0) rectangle (axis cs:237,0.0745098039215687);
\end{axis}

\end{tikzpicture}
        \caption{Histograms of the first time for each CS to reach a width of $\varepsilon=0.05$ with $N=250$. As expected, the \propM strategy~(both with and without control variates) is significantly more sample-efficient than the uniform baseline, and furthermore, the improvement by using control variates increases with increasing informativeness~(i.e., $\rho$) of the side-information.}
        \label{fig:housing}
    \end{figure}
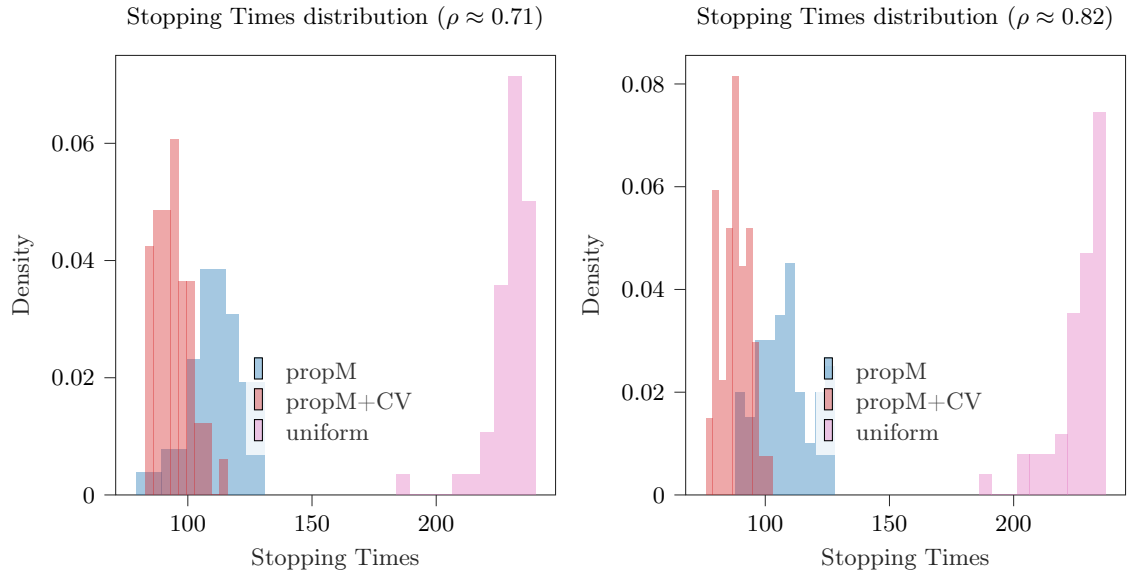
    \paragraph{Experimental Results.} In~\Cref{fig:housing}, we consider two instances of this problem: (i) side-information generated by a decision-tree regressor, and (ii) side-information generated by a random-forest regressor, consisting of $10$ trees. In the former case, the correlation between the side information and the ground-truth is approximately $0.71$, while in the latter it is around $0.82$. As shown in the plots, the \propM based strategies (both with and without control variates)  significantly outperform the uniform baseline strategy. Furthermore, the improvement by incorporating control variates increases with increasing correlation. 

\end{document}